\documentclass[sigconf]{acmart}
\setcopyright{none}
\renewcommand\footnotetextcopyrightpermission[1]{}
\usepackage{xspace}
\newcommand{\eat}[1]{}
\newcommand{\class}[1]{{\ensuremath{\mathsf{#1}}}}

\newcommand{\cw}[1]{{\color{blue}\footnotesize[cw: #1]}}
\newcommand{\kartik}[1]{{\color{brown} \footnotesize[Kartik: #1] }}

\newcommand{\re}[1]{{\color{black} #1} }

\newcommand{\boldparagraph}[1]{\vspace{7pt}\noindent\textbf{#1}}
\newcommand{\share}[1]{\llbracket #1 \rrbracket^m}

\newcommand{\system}{IncShrink\xspace}

\newcommand{\sx}{\mathcal{S}_0\xspace}
\newcommand{\sy}{\mathcal{S}_1\xspace}

\newcommand{\user}{owner\xspace}

\newcommand{\lcache}{\texttt{cache}\xspace}

\newcommand{\enc}{\ensuremath{Enc}\xspace}

\newcommand{\pp}{\mathsf{pp}\xspace}

\newcommand{\rd}{\ensuremath{\class{rd}}\xspace}
\newcommand{\msb}{\ensuremath{\class{msb}}\xspace}
\newcommand{\sort}{\ensuremath{\class{ObliSort}}\xspace}

\newcommand{\query}{\ensuremath{\class{Query}}\xspace}

\newcommand{\negl}{\ensuremath{\class{negl}}\xspace}

\newcommand{\vi}{\ensuremath{\class{VIEW}}\xspace}
\newcommand{\pk}{\ensuremath{\class{pk}}\xspace}

\newcommand{\Mod}[1]{\ (\mathrm{mod}\ #1)}
\newcommand\myeq{\stackrel{\mathclap{\normalfont\mbox{c}}}{=}}

\newcommand{\writec}{\ensuremath{\class{write}}\xspace}

\newcommand{\lap}{\ensuremath{\class{Lap}}\xspace}

\newcommand{\upatt}{\ensuremath{\class{UpdtPatt}}\xspace}

\newcommand{\sync}{\ensuremath{\class{Shrink}}\xspace}
\newcommand{\trans}{\ensuremath{\class{Transform}}\xspace}
\newcommand{\gs}{\ensuremath{\class{share}}\xspace}
\newcommand{\rec}{\ensuremath{\class{recover}}\xspace}
\newcommand{\timer}{\ensuremath{\class{sDPTimer}}\xspace}
\newcommand{\ant}{\ensuremath{\class{sDPANT}}\xspace}

\usepackage{hyperref}
\usepackage{tablefootnote}
\usepackage{algorithm}
\usepackage{algorithmicx}
\usepackage[noend]{algpseudocode}
\usepackage{mathtools}
\usepackage{caption}
\usepackage{balance}
\usepackage{bbm}
\usepackage{multirow}
\usepackage{subcaption}
\newtheorem{theorem}{Theorem}
\newtheorem{definition}{Definition}
\newtheorem{example}{Example}[section]

\AtBeginDocument{%
  \providecommand\BibTeX{{%
    \normalfont B\kern-0.5em{\scshape i\kern-0.25em b}\kern-0.8em\TeX}}}

\setcopyright{acmcopyright}
\copyrightyear{2022}
\acmYear{2022}
\acmDOI{10.1145/1122445.1122456}

\acmConference[SIGMOD '22]{SIGMOD '22: International Conference on Management of Data}{June 12--17, 2022}{Philadelphia, PA}
\acmPrice{15.00}
\acmISBN{978-1-4503-XXXX-X/18/06}

\begin{document}

\title{\system: Architecting Efficient Outsourced Databases using Incremental MPC and Differential Privacy}

\author{Chenghong Wang}
\affiliation{%
  \institution{Duke University}
  \city{}
  \country{}
}
\email{chwang@cs.duke.edu}

\author{Johes Bater}
\affiliation{%
  \institution{Duke University}
  \city{}
  \country{}
  }
\email{johes.bater@duke.edu}

\author{Kartik Nayak}
\affiliation{%
  \institution{Duke University}
  \city{}
  \country{}
}
\email{kartik@cs.duke.edu}

\author{Ashwin Machanavajjhala}
\affiliation{%
 \institution{Duke University}
 \city{}
  \country{}
 }
 \email{ashwin@cs.duke.edu}


\begin{abstract}
In this paper, we consider secure outsourced growing databases that support view-based query answering. These databases allow untrusted servers to privately maintain a materialized view, such that they can use only the materialized view to process query requests instead of accessing the original data from which the view was derived. To tackle this, we devise a novel view-based secure outsourced growing database framework, \system. The key features of this solution are:  (i) \system maintains the view using incremental MPC operators which eliminates the need for a trusted third party upfront, and (ii) to ensure high performance, \system guarantees that the leakage satisfies DP in the presence of updates. To the best of our knowledge, there are no existing systems that have these properties.  
We demonstrate \system's practical feasibility in terms of efficiency and accuracy with extensive empirical evaluations on real-world datasets and the TPC-ds benchmark. The evaluation results show that \system provides a 3-way trade-off in terms of privacy, accuracy and efficiency guarantees, and offers at least a 7,800$\times$ performance advantage over standard secure outsourced databases that do not support view-based query paradigm. 
\end{abstract}

\eat{

}


\setlength{\textfloatsep}{1pt plus 0.0pt minus 0.0pt}
\setlength{\floatsep}{1pt plus 0.0pt minus 0.0pt}
\setlength{\intextsep}{1pt plus 0.0pt minus 0.0pt}
\maketitle
\section{Introduction}
\label{sec:introduction}

There is a rapid trend of organizations moving towards outsourcing their data to cloud providers to take advantages of its cost-effectiveness, high availability, and ease of maintenance. Secure outsourced databases are designed to help organizations outsource their data to untrusted cloud servers while providing secure query functionalities without sacrificing data confidentiality and privacy. The main idea is to have the data owners upload the encrypted or secret-shared data to the outsourcing servers. Moreover, the servers are empowered with secure protocols which allow them to process queries over such securely provisioned data. A series of works such as CryptDB~\cite{popa2012cryptdb}, Cipherbase~\cite{arasu2013orthogonal}, EnclaveDB~\cite{priebe2018enclavedb}, and HardIDX~\cite{fuhry2017hardidx}  took the first step in the exploration of this scope by leveraging strong cryptographic primitives or secure hardware to accomplish the aforementioned design goals. Unfortunately, these solutions fail to provide strong security guarantees, as recent works on leakage-abuse attacks~\cite{cash2015leakage, zhang2016all, kellaris2016generic, blackstone2019revisiting} have found that they are vulnerable to a variety of reconstruction attacks that exploit side-channel leakages. For instance, an adversary can fully reconstruct the data distribution after observing the query processing transcripts.  

Although some recent efforts, such as~\cite{poddar2016arx, eskandarian2017oblidb, bater2018shrinkwrap, kamara2019computationally, patel2019mitigating, demertzis2020seal, kellaris2021accessing, naveed2014dynamic, xu2019hermetic, zheng2017opaque}, have shown potential countermeasures against leakage-abuse attacks, the majority of these works focus primarily on static databases. A more practical system often requires the support of updates to the outsourced data~\cite{agarwal2019encrypted, stefanov2014practical, ghareh2018new, wang2021dp}, which opens up new challenges. Wang et al.~\cite{wang2021dp} formulate a new type of leakage called {\it update pattern} that affects many existing outsourced database designs when the underlying data is dynamically growing. To mitigate such weakness, their solution dictates the data owners' update behavior to private record synchronization strategies, with which it perturbs the owners' logical update pattern. However, their solution only considers a na\"ive query answering mode such that each query is processed independently and evaluated directly over the entire outsourced data. This inevitably leads to a substantial amount of redundant computation by the servers. For example, consider the following use case where a courier company partners with a local retail store to help deliver products. The retail store has its sales data, and the courier company has its delivery records, both of which are considered to be the private property of each. Now assume the retail store owner wants to know how many of her products are delivered on time (i.e., within 48 hours of the courier accepting the package). With secure outsourced databases, the store owner and the courier company have the option to securely outsource their data and its corresponding computations to cloud servers. However, in a na\"ive query processing mode, the servers have to recompute the entire join relation between the outsourced data whenever a query is posted, which raises performance concerns. 



In this work, we take the next step towards designing a secure outsourced growing database (SOGDB) architecture with a more efficient query answering mechanism. Our proposed framework employs a novel secure query processing method in which the servers maintain a growing size materialized view corresponding to the owner's outsourced data. The upcoming queries will be properly answered using only the materialized view object. This brings in inherent advantages of view-based query answering~\cite{srivastava1996answering} paradigm, such as allowing the servers to cache important intermediate outputs, thus preventing duplicated computation. For instance, with our view-based SOGDB architecture, one can address the performance issues in the aforementioned use case by requiring the servers to maintain a materialized join table between the outsourced sales and delivery data. Moreover, whenever the underlying data changes, the materialized join table is updated accordingly. To this end, the servers only need to perform secure filtering over the materialized join table for processing queries, which avoids duplicated computation of join relations. 
\eat{
\begin{figure}[h]
\centering
\includegraphics[width=\linewidth]{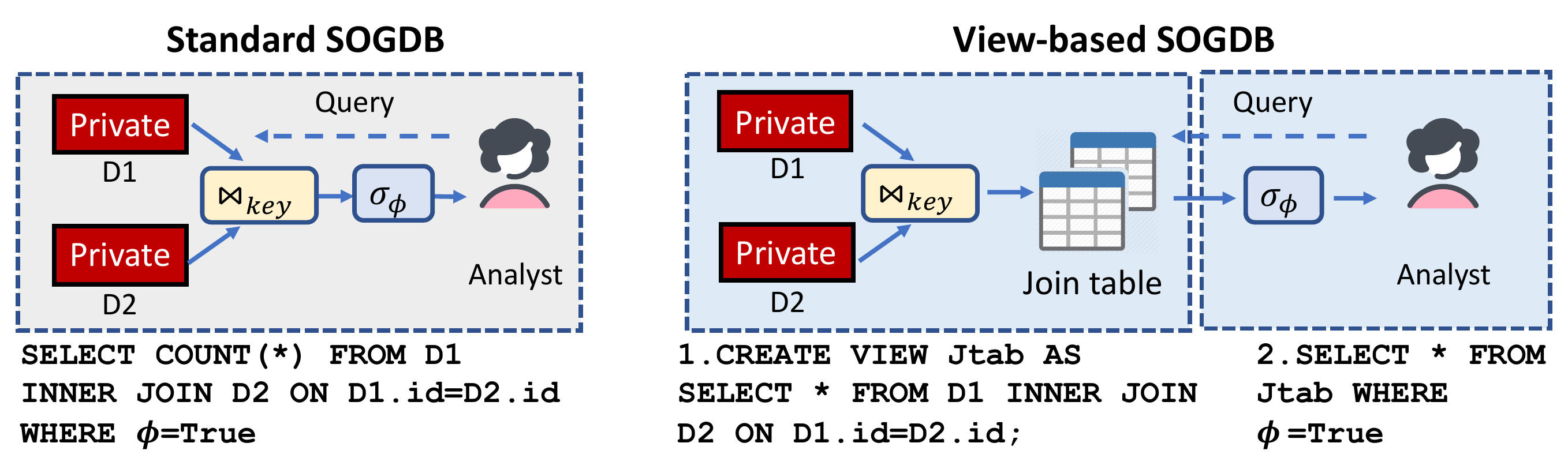}
\vspace{-8mm}
\caption{Standard vs. View-based SOGDB architecture}
\label{fig:ideacmp}
\end{figure}
}

There is no doubt one can benefit in many aspects from the view-based query answering paradigm. However, designing a practical view-based SOGDB is fraught with challenges. First, the servers that maintain the materialized view is considered to be potentially untrusted. Hence, we must explore the possibility of updating such materialized view without a trusted curator. A typical way is to leverage secure multi-party computation (MPC). However, na\"ively applying MPC for updating view instances over growing data would expose extra information leakage (i.e., update pattern~\cite{wang2021dp}). For example, consider the use case we mentioned before where the servers maintain a join table over the sales and the delivery data. Even with MPC, one can still obtain the true cardinality of newly inserted entries to the join table by looking at the output size from MPC. This would allow the adversary to learn the exact number of packages requested for delivery by the local retail store at any given time. Na\"ive methods, such as always padding newly generated view tuples to the maximum possible size or choosing never to update the materialized view, could prevent the aforementioned leakage. However, such an approach either introduces a large performance burden or does not provide us with the functionality of database updates. To combat this, we propose a novel view update methodology that leverages incremental MPC and differential privacy (DP), which hides the corresponding update leakage using DP while balancing between the efficiency and accuracy. 

This design pattern helps us to address the extra leakage, but also raises new challenges. The transformation from outsourced data to a view instance may have unbounded stability, i.e., an input record may contribute to the generation of multiple rows in the transformed output, which could cause unbounded privacy loss. To address this, we enforce that any individual data outsourced by the owner only contributes to the generation of a fixed number of view tuples. As the transformation after applying this constraint has bounded stability, thus we obtain a fixed privacy loss with respect to each insertion (logical update) to the owner's logical data.

\eat{Furthermore, in contrast to existing efforts, where they assume only trusted owners can execute private strategies. According to our design, the subject who maintains the materialized view is untrusted, hence we must ensure our view update strategy is available to untrusted parties without compromising our overall privacy guarantee. To address the third challenge, we compile the our private strategies as implementation-specified multi-party secure computing (MPC)~\cite{goldreich2009foundations} protocols and cache out-of-protocol parameters using secret-sharing techniques~\cite{beimel2011secret}. This allows untrustworthy parties to evaluate private view strategies while obtaining only the allowed information.}

Putting all these building blocks together, a novel view-based SOGDB framework, \system, falls into place. 
We summarize our contributions as follows:
\vspace{-1mm}
\begin{itemize}
    \item \system is a first of its kind, secure outsourced growing database framework that supports view-based query processing paradigm. Comparing with the standard SOGDB~\cite{wang2021dp} that employs na\"ive query answering setting, \system improves query efficiency, striking a balance between the guarantees of privacy, efficiency and accuracy, at the same time.
    \item \system integrates incremental MPC and DP to construct the view update functionality which (i) allows untrusted entities to securely build and maintain the materialized view instance (ii) helps to reduce the performance overhead of view maintenance, and (iii) provides a rigorous DP guarantee on the leakage revealed to the untrusted servers. 
    \item \system imposes constraints on the record contribution to view tuples which ensures the entire transformation from outsourced data to the view object over time has bounded stability. This further implies a  bounded privacy loss with respect to each individual logical update.
    \item We evaluate \system on use cases inspired by the {\it Chicago Police Data} and the TPC-ds benchmark. The evaluation results show at least 7800$\times$ and up to 1.5e+5$\times$ query efficiency improvement over standard SOGDB. Moreover, our evaluation shows that \system provides a 3-way trade-off between privacy, efficiency, and utility while allowing users to adjust the configuration to obtain their desired guarantees. 
    \vspace{-1mm}
\end{itemize}

\vspace{-2mm}
\section{Overview}
We design \system to meet three main goals:
\vspace{-1mm}
\begin{itemize}
    \item \textbf{View-based query answering.} \system enables view-based query answering for a class of specified queries over secure outsourced growing data.
    \item \textbf{Privacy against untrusted server.} Our framework allows untrusted servers to continuously update the materialized view while ensuring that the privacy of the owners' data is preserved against outsourcing servers.
    \item \textbf{Bounded privacy loss.} The framework guarantees an unlimited number of updates under a fixed privacy loss.
\end{itemize}
In this section, we first outline the key ideas that allow \system
to support the primary research goals in Section~\ref{sec:ideas}. Then we briefly review the framework components in Section~\ref{sec:compo}. We provide a running example in Section~\ref{sec:workflow} to illustrate the framework workflow as well as the overall architecture.

\vspace{-2mm}
\subsection{Key Ideas}\label{sec:ideas}
\vspace{-1mm}
\boldparagraph{KI-1. View-based query processing over secure outsourced growing data.} \system employs materialized view for answering pre-specified queries over secure outsourced growing data. The framework allows untrusted outsourcing servers to securely build and maintain a growing-size materialized view corresponding to the selected view definition. A typical materialized view can be either transformed solely based on the data provisioned by the owners, i.e., a join table over the outsourced data, or in combination with public information, i.e., a join table between the outsourced data and public relations. Queries posed to the servers are rewritten as queries over the defined view and answered using only the view object. Due to the existence of materialization, the outsourcing servers are exempted from performing redundant computations. 

\vspace{-1mm}
\boldparagraph{KI-2. Incremental MPC with DP update leakage.}  A key design goal of \system is to allow the untrusted servers to privately update the view instance while also ensuring the privacy of owners' logical data. As we mentioned before, compiling the view update functionality into the MPC protocol is not sufficient to ensure data privacy, as it still leaks the true cardinality of newly inserted view entries at each time, which is directly tied with the owners' record update patterns~\cite{wang2021dp}. Although na\"ive approaches such as exhaustive padding (EP) of MPC outputs or maintaining a one-time materialized view (OTM) could alleviate the aforementioned privacy risk. They are known to either incorporate a large amount of dummy data or provide poor query accuracy due to the lack of updates to the materialized view. This motivates our second key idea to design an incremental MPC protocol for view updates while balancing the privacy, accuracy, and efficiency guarantees.
 In our design, we adopt an innovative "Transform-and-Shrink" paradigm, where the protocol is composed of two sub-protocols, $\trans$ and $\sync$, that coordinate with each other. $\trans$ generates corresponding view entries based on newly outsourced data, and places them to an exhaustively padded secure cache to avoid information leakage. $\sync$, runs independently, periodically synchronizes the cached data to the materialized view according to its internal states. To prevent the inclusion of a large amount of dummy data, $\sync$ shrinks the cached data into a DP-sized secure array such that a subset of the dummy data is removed whereas the true cardinality is still preserved. As a result, the resulting protocol ensures any entity's knowledge with respect to the view instance is bounded by differential privacy. 

\vspace{-1mm}
\boldparagraph{KI-3. Fixed privacy loss through constraints on record contributions.} When \system releases noisy cardinalities, it ensures $\epsilon$-DP with respect to the view instance. However, this does not imply $\epsilon$-DP to the logical data where the view is derived. This is because an individual data point in the logical database may contribute to generating multiple view entries. As a result, the framework either incurs an unbounded privacy loss or has to stop updating the materialized view after sufficiently many synchronizations. This leads us to our third key idea, where we bound the privacy loss by imposing constraints on the contributions made by each individual record to the generation of the view object. Each data point in the logical database is allocated with a contribution budget, which is consumed whenever the data is used to generate a new view entry. Once the contribution budget for a certain record is exhausted, \system retires this data and will no longer use it to generate view entries. With such techniques, \system is able to constantly update the materialized view with a bounded privacy loss. On the other hand, despite such constraints over record contributions, \system is still able to support a rich class of queries with relatively small errors (Section~\ref{sec:exp}).

\vspace{-2mm}
\subsection{Framework Components}\label{sec:compo}
\vspace{-2mm}
\boldparagraph{Underlying database.} \system does not create a new secure outsourced database but rather builds on top of it. Therefore, as one of the major components, we assume the existence of an underlying secure outsourced database scheme.
Typically, secure outsourced databases can be implemented according to different architectural settings, such as the models utilizing server-aided MPC~\cite{kamara2011outsourcing, bater2017smcql,bater2018shrinkwrap,mohassel2017secureml, tan2021cryptgpu}, homomorphic encryption~\cite{chowdhury2019crypt}, symmetric searchable encryption~\cite{bellare2007deterministic, curtmola2011searchable, stefanov2014practical, cash2014dynamic, kamara2012dynamic, ghareh2018new, amjad2019forward} or trusted hardware~\cite{zheng2017opaque, eskandarian2017oblidb, priebe2018enclavedb, vinayagamurthy2019stealthdb}. For the ease of demonstration, we focus exclusively on the outsourced databases built upon the server-aided MPC model, where a set of data owners secretly share their data to two untrusted but non-colluding servers $\sx$ and $\sy$. The two outsourcing servers are able to perform computations (i.e., query processing) over the secret shared data by jointly evaluating a 2-party secure computation protocol. More details about this setting and its corresponding security definitions are provided in Section~\ref{sec:model}. We stress that, although the protocols described in this paper assumes an underlying database architected under the server-aided MPC setup, these protocols can be adapted to other settings as well. 


\vspace{-1mm}
\boldparagraph{Materialized view.}
A materialized view is a subset of a secure outsourced database, which is typically generated from a query and stored as an independent object (i.e., an encrypted or secret-shared data structure). The servers can process queries over the view instance just as they would in a persistent secure database. Additionally, changes to the underlying data are reflected in the entries shown in subsequent invocations of the materialized view.

\vspace{-1mm}
\boldparagraph{View update protocol.} The view update protocol is an incremental MPC protocol jointly evaluated by the outsourcing servers. It allows untrusted servers to privately update the materialized view but also ensures bounded leakage. More design and implementation details about the view update protocol can be found in Section~\ref{sec:sync}.

\vspace{-1mm}
\boldparagraph{Secure outsourced cache.}\label{sec:scache}
The secure outsourced cache is a secure array (i.e., memory blocks that are encrypted, secret-shared, or stored inside trusted hardware) denoted as $\boldsymbol{\sigma}[1,2,3,\ldots]$, which is used to temporarily store newly added view entries that will later be synchronized to the materialized view. In this work, as we focus on the server-aided MPC model, thus $\boldsymbol{\sigma}$ is considered as a secret shared memory block across two non-colluding servers. Each $\boldsymbol{\sigma}[i]$ represents a (secret-shared) view entry or a dummy tuple. Details on how our view update protocol interacts with the secure cache (i.e., read, write, and flush cache) are provided in Section~\ref{sec:sync}.
\subsection{\system Workflow}\label{sec:workflow}
We now briefly review the framework architecture and illustrate its workflow with a running example (as shown in Figure~\ref{fig:flow}), where an analyst is interested in a join query over the outsourced data from two data owners. 
\vspace{2mm}
\begin{figure}[h]
\centering
\includegraphics[width=0.95\linewidth]{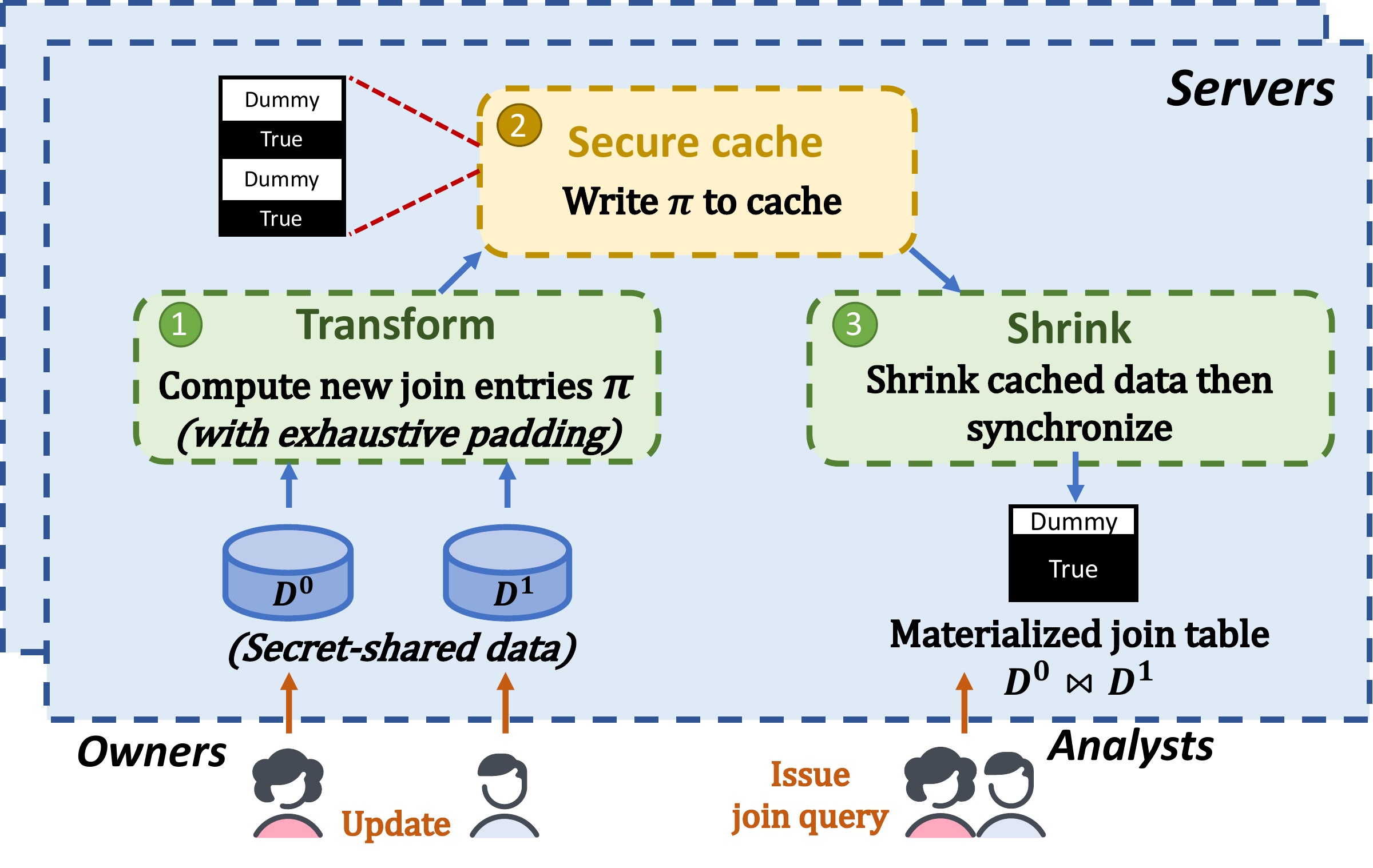}
\vspace{-3mm}
\caption{Framework workflow.}
\label{fig:flow}
\end{figure}

Initially, the analyst obtains authentications from the owners and registers the query with the outsourcing servers. The servers decide the view definition as a join table, set up the initial materialization structure, the secure cache, and then compile the corresponding secure protocols $\trans$ and $\sync$ for maintaining the view instance. Since then, the owners periodically update the outsourced data, securely provisioning the newly received data since the last update (through the update functionality defined by the underlying database). For demonstration purposes, we assume that the owners submit a fixed-size data block (possibly padded with dummy records) at predetermined intervals. We discuss potential extensions to support other update behaviors in a later section.
Whenever owners submit new data, the servers invoke $\trans$ to securely compute new joins. The join outputs will be padded to the maximum size then placed into a secure cache. Next, $\sync$ is executed independently, where it periodically synchronizes data from the secure cache to the materialized join table with DP resized cardinalities. Note that the DP noise used to distort the true cardinality can be either positive or negative. If the noise is negative, some of the real tuples in the cache are not fetched by $\sync$. We refer to the real view tuples left in the secure cache as the ``deferred data''. On the contrary, if the noise is positive, some deferred data or additional dummy tuples will be included to synchronize with the view.
On the other hand, the analyst can issue query requests to the servers, which process the issued queries over the materialized join table (through the query functionality defined by the underlying database) and return the resulting outputs back to the analyst.

\section{Preliminaries}
\label{sec:background}
\eat{
\vspace{-2mm}
\boldparagraph{Differential privacy (DP).} Differential privacy is currently the de-facto standard for achieving privacy in data analysis, where it defines a stability property for data processing mechanisms. 
\begin{definition}[$\epsilon$-differential privacy~\cite{dwork2014algorithmic}]\label{def:dp} $\mathcal{M}$ satisfies $\epsilon$-DP if for any pair of neighboring databases $D$ and $D'$, such that $D$ and $D'$ differ by addition or removal of one tuple, and for all $\forall O \subset \mathcal{O}$, where $\mathcal{O}$ denotes all possible outputs, the following holds
$$\textup{Pr}[\mathcal{M}(D)\in O] \leq e^{\epsilon}\textup{Pr}[\mathcal{M}(D')\in O]$$
\end{definition}
Note that if a single tuple in database $D$ corresponds one individual data subject (user), then it defines user-level privacy. If the tuples in $D$ are secret events, where multiple events can belong to one data subject, then Definition~\ref{def:dp} defines $\epsilon$-event level DP.
}
\vspace{-2mm}
\boldparagraph{Multi-party secure computation (MPC).} MPC utilizes cryptographic primitives to enable a set of participants $P_1, P_2, ..., P_n$ to jointly compute a function $f$ over private input data $x_i$ supplied by each party $P_i$, without using a trusted third party. The theory of MPC offers strong security guarantee similar as what can be achieved with a trusted third party~\cite{goldreich2009foundations}, i.e., absolutely no information leak to each participant $P_i$ beyond the desired output of $f(x_1, x_2, ..., x_n)$ and their  input $x_i$. In this work we focus mainly on the 2-party secure computing setting.




\boldparagraph{$(n,t)$-secret sharing.} Given ring $\mathbb{Z}_{m}$, and $m=2^{\ell}$. A $(n,t)$-secret sharing ($t$-out-of-$n$) over $\mathbb{Z}_{m}$ shares a secret value $x\in\mathbb{Z}_{m}$ with $n$ parties such that the sharing satisfies the following property:
\begin{itemize}
    \item {\bf Availability} Any $t'$ of the $n$ parties such that $t'\geq t$ can recover the secret value $x$.
    \item {\bf Confidentiality} Any $t'$ of the $n$ parties such that $t'< t$ have no information of $x$.
\end{itemize}
For any value $x\in\mathbb{Z}_{m}$, we denote it's secret sharing as $\share{x}\gets(x_1, x_2, ..., x_n)$. There are many existing efforts to implement such secret sharing design~\cite{beimel2011secret}, we focus on XOR-based $(2,2)$-secret sharing over $\mathbb{Z}_{2^{32}}$ with the following specifications.
\begin{itemize}
    \item {\bf Generate shares} $\gs(x)$: Given $x\in\mathbb{Z}_{m}$, sample random values \smash{$x_1\xleftarrow[]{\text{rd}}\mathbb{Z}_{m}$}, compute $x_2 \gets x \oplus x_1$, and return secret shares $\share{x}\gets(x_1, x_2)$.
    \item {\bf Recover shares} $\rec({\share{x}})$: Given secret shares $\share{x}\gets (x_1, x_2)$, compute $x\gets x_1\oplus x_2$, then return $x$.
\end{itemize}


\section{Privacy Model}\label{sec:model}
In general, we consider our framework supports dynamic updating of the materialized view while hiding the corresponding update leakage. More specifically, we consider the participants involved in the outsourcing phase are a set of data owners and two servers $\sx$, and $\sy$. We assume there exists a semi-honest {\it probabilistic polynomial time (p.p.t.) }  adversary $\mathcal{A}$ who can corrupt any subset of the owners and at most one of the two servers. Previous work~\cite{mohassel2017secureml} refers to this type of adversary as the {\it admissible adversary}, which captures the property of two non-colluding servers, i.e., if one is compromised
by the adversary, the other one behaves honestly.
Our privacy definition requires that the knowledge $\mathcal{A}$ can obtain about any single data of the remaining honest owners, by observing the view updates, is bounded by differential privacy. In this section, we first provide key terminologies and notations (Section~\ref{sec:notations}) then formalize our privacy model (Section~\ref{sec:privacy}) using {\it simulation-based computational differential privacy (SIM-CDP)}~\cite{mironov2009computational}.

\subsection{Notations}\label{sec:notations}
\vspace{-1mm}
\boldparagraph{Growing database.} A growing database is a dynamic relational dataset with insertion only updates, thus we define it as a collection of (logical) 
updates, $\mathcal{D}=\{u_i\}_{i\geq 0}$, where $u_i$ is a time stamped data. We write $\mathcal{D}=\{\mathcal{D}_t\}_{t\geq 0}$, such that $\mathcal{D}_t$ denotes the database instance of the growing database $\mathcal{D}$ at time $t$ and $\forall ~\mathcal{D}_t~, \mathcal{D}_t \subseteq \mathcal{D}$.

\vspace{-1mm}
\boldparagraph{Outsourced data.} The outsourced data is denoted as $\mathcal{DS}$, which stores the secret-shared entries corresponding to the records in the logical database, with the possibility to include additional dummy data. Similarly, we write $\mathcal{DS}=\{\mathcal{DS}_{t}\}_{t\geq0}$, where $\mathcal{DS}_{t}\subseteq\mathcal{DS}$ is the outsourced data at time $t$.

\vspace{-1mm}
\boldparagraph{Materialized view.} We use $\mathcal{V}$ to denote the materialized view instance which is a collection of secret-shared tuples. Each tuple in $\mathcal{V}$ is transformed from the outsourced data $\mathcal{DS}$ or in combination with public information. We define $\mathcal{V}=\{\mathcal{V}_t\}_{t\geq 0}$, where $\mathcal{V}_t$ denotes the materialized view structure at time $t$, and $\Delta\mathcal{V}_t$ denotes the changes between (newly generated view entries) $\mathcal{V}_t$ and $\mathcal{V}_{t-1}$

\vspace{-1mm}
\boldparagraph{Query.} Given a growing database $\mathcal{D}$ and a corresponding materialized view $\mathcal{V}$, we define the logical query posted at time $t$ as $q_{t}(\mathcal{D}_t)$ and the re-written view-based query as $\tilde{q}_{t}(\mathcal{V}_t)$. We refer the L1 norm of the difference between $\tilde{q}_{t}(\mathcal{V}_t)$ and $q_{t}(\mathcal{D}_t)$ as the {\it L1 query error}, denoted as $L_{q_t}\gets||\tilde{q}_t(\mathcal{V}_t) - {q}_t(\mathcal{D}_t)||_1$, which measures the difference between the server responded outputs and their corresponding logical results. Additionally, we call the elapsed time for processing $\tilde{q}_t(\mathcal{V}_t)$ as the {\it query execution time (QET)} of $q_t$. In this work, we use L1 error and QET as the main metrics to evaluate the accuracy and efficiency of our framework, respectively.

\subsection{Privacy Definition}\label{sec:privacy}
Based on the formalization of update pattern in~\cite{wang2021dp}, we first provide a more generalized definition of update pattern that captures updates to view instances.
\begin{definition}[Update pattern]\label{def:updt} Given a growing database $\mathcal{D}$, the update pattern for outsourcing $\mathcal{D}$ is the function family of $\upatt(\mathcal{D}) = \{\upatt_t\left(\mathcal{D}\right)\}_{t\in\mathbb{N}^{+}}$, with:
$$\upatt_{t}\left(\mathcal{D}\right) = \left(t, |T_t(\mathcal{D}|)\right)$$
where $T_t$ is a transformation function that outputs a set of tuples (i.e., new view entries) been outsourced to the server at time $t$. 
\end{definition}
In general, Definition~\ref{def:updt} defines the transcript of entire update history for outsourcing a  growing database $\mathcal{D}$. It may include information about the volume of the outsourced data and their corresponding insertion times. Moreover, if $T_t(\mathcal{D})\gets \mathcal{D}_t - \mathcal{D}_{t-1}$, then this simply indicates the record insertion pattern~\cite{wang2021dp}.
\begin{definition}[Neighboring growing databases]\label{def:ngdb} Given a pair of growing databases $\mathcal{D}$ and $\mathcal{D'}$, such that there exists some parameter $\tau \geq 0$, the following holds: (i) $\forall~ t \leq \tau, \mathcal{D}_t = \mathcal{D}'_t$ (ii) $\forall~ t > \tau$, $\mathcal{D}_t$ and $\mathcal{D}'_t$ differ by the addition or removal of a single logical update.
\end{definition}
\begin{definition}[DP mechanism over growing data]\label{def:dp-jpat} Let $F$ to be a mechanism applied over a growing database. $F$ is said to satisfy $\epsilon$-DP if for any neighboring growing databases $\mathcal{D}$ and $\mathcal{D}'$, and any $O \in \mathcal{O}$, where $\mathcal{O}$ is the universe of all possible outputs, it satisfies:
\begin{equation}
    \begin{split}
         \textup{Pr}\left[F(\mathcal{D}) \in O\right] \leq e^{\epsilon} \textup{Pr}\left[F(\mathcal{D}')\in O \right]
    \end{split}
\end{equation} 
\end{definition}
Definition~\ref{def:dp-jpat} ensures that, by observing the output of $F$, the information revealed by any single logical update posted by the owner is differentially private. Moreover, if the logical update corresponds to different entity's (owner’s) data, the same holds for $F$ over each owner's logical database (privacy is guaranteed for each entity). \re{Additionally, in this work, we assume each logical update $u_i\in\mathcal{D}$ as a secret event, therefore such mechanism $F$ achieves $\epsilon$-event level DP~\cite{dwork2010differential}. However, due to the group-privacy property~\cite{kifer2011no, xiao2015protecting, liu2016dependence} of DP, one can achieve privacy for properties across multiple updates as well as at the user level as long as the property depends on a finite number of updates. An overall $\epsilon$-user level DP can be achieved by setting the privacy parameter in Definition~\ref{def:dp-jpat} to $\frac{\epsilon}{\ell}$, where $\ell$ is the maximum number of tuples in the
growing database owned by a single user. In practice, if the number of tuples owned by each user is unknown, a pessimistic large value can be chosen as $k$. Moreover, recent works~\cite{cao2017quantifying, song2017pufferfish}  have provided methods for deriving an $\epsilon < \epsilon' \leq \ell\times\epsilon$, such that $\epsilon$-event DP algorithms provide an $\epsilon'$ bound on privacy loss when data are correlated. For certain correlations, $\epsilon'$ can be even close to $\epsilon$ and much smaller than $\ell\times\epsilon$. In general, we emphasize that for the remainder of the paper, we focus exclusively on developing algorithms that ensure event-level privacy with parameter $\epsilon$, while simultaneously satisfying all the aforementioned privacy guarantees, with possibly a different privacy parameter. 

}

\begin{definition}[SIM-CDP view update protocol]\label{def:protocol-sogdb}
A view update protocol $\Pi$ is said to satisfy $\epsilon$-SIM-CDP if there exists a p.p.t.\ simulator $\mathcal{S}$ with only access to a set of public parameters $\pp$ and the output of a mechanism $F$ that satisfies Definition~\ref{def:dp-jpat}. Then for any growing database instance $\mathcal{D}$, and any p.p.t. adversary $\mathcal{A}$, the adversary's advantage satisfies:
\begin{equation}
    \begin{split}
        & \textup{Pr}\left[\mathcal{A}\left(\vi^{\Pi}(\mathcal{D}, \pp)=1\right)\right] \\
        & \leq  \textup{Pr}\left[\mathcal{A}\left(\vi^{\mathcal{S}}(F(\mathcal{D}), \pp)\right)=1\right] + \negl(\kappa)
    \end{split}
\end{equation}
where $\vi^{\Pi}$, and $\vi^{\mathcal{S}}$ denotes the adversary's view against the protocol execution and the simulator's outputs, respectively. And $\negl(\kappa)$ is a negligible function related to a security parameter $\kappa$.
\end{definition}
Definition~\ref{def:protocol-sogdb} defines the secure protocol for maintaining the materialized view such that as long as there exists at least one honest owner, the privacy of her data's individual records is guaranteed. 
In addition, the remaining entities' knowledge about honest owner's data is preserved by differential privacy. In other words, any {\it p.p.t.} adversary's knowledge of such protocol $\Pi$ is restricted to the outputs of an $\epsilon$-DP mechanism $F$. We refer to the mechanism $F$ as the leakage profile of protocol $\Pi$, and a function related to the update pattern, i.e., $F(\mathcal{D})=f(\upatt(\mathcal{D}))$. In the rest of this paper, we focus mainly on developing view update protocols that satisfy this definition. Moreover, Definition~\ref{def:protocol-sogdb} is derived from the SIM-CDP definition which is formulated under the {\it Universal Composition (UC) } framework~\cite{canetti2001universally}. Thus Definition~\ref{def:protocol-sogdb} leads to a composability property such that if other protocols (i.e., $\query$ protocol) defined by the underlying databases also satisfy UC security, then privacy guarantee holds under the composed system.

\eat{
\begin{definition}[$(\omega, q)$-BC transformation] Given a transformation $T = \{T_i\}_{i\geq0}$, $T$ is said to be a $(\omega, q)$-BC transformation if each $T_i$ is a $q$-BC transformation, and for any database $\mathcal{D}_1$, $\mathcal{D}_2\in \mathcal{D}$:
\begin{equation}
    \begin{split}
    || T(\mathcal{D}_1) - T(\mathcal{D}_2)|| \leq \omega q || \mathcal{D}_1 - \mathcal{D}_2|| 
    \end{split}
\end{equation}
\end{definition}

To distinguish, we refer $q$ as one-time contribution and $\omega$ as the contribution over time. According to Theorem~\ref{tm:bcompose}, applying $\frac{\epsilon}{q\omega}$-DP mechanism over $(\omega, q)$-BC transformed data implies $\epsilon$-event level DP for the original data.
}

\section{Protocol Design}\label{sec:sync}
In general, our view update protocol is implemented as an incremental MPC across two non-colluding outsourcing servers. Specifically, this incremental MPC is composed of two sub-protocols, $\trans$, and $\sync$ that operate independently but collaborate with each other. The reason for having this design pattern is that we can decouple the view transformation functionality and its update behavior, which provides flexibility in the choice of different view update strategies. For example, one may want to update the materialized view at a fixed interval or update it when there are enough new view entries. Each time when one needs to switch between these two strategies, she only needs to recompile the $\sync$ protocol without making any changes to the $\trans$ protocol. In this section, we discuss the implementation details of these two protocols, in Section~\ref{sec:pre}and~\ref{sec:dp-st}, respectively. Due to space concerns, for theorems in this section, we defer the proofs to the appendix.
\subsection{$\trans$ Protocol}\label{sec:pre}
Whenever owners submit new data, $\trans$ is invoked to convert the newly outsourced data to its corresponding view instance based on a predefined view definition. Although, one could simply reuse the query capability of the underlying database to generate the corresponding view tuples. There are certain challenges in order to achieve our targeted design objectives. Here are two examples: (i) The view transformation might have unbounded stability which further leads to an unbounded privacy loss; \re{(ii) While existing work ~\cite{bater2018shrinkwrap} implements a technique similar to first padding the output and then reducing its size, they compile the functionality as a one-time MPC protocol, which makes it difficult for them to handle dynamic data. Our design of constructing ``Transform'' and ``Shrink''  as independent MPC protocols overcome this problem and introduce flexibility in the choice of view update policy, but it raises a new challenge in that the two independently operating protocols still need to collaborate with each other.  By default, the $\sync$ protocol is unaware of how much data can be removed from the secure cache, therefore $\trans$ must privately inform $\sync$ how to eliminate the dummy records while ensuring DP.
To address these challenges, we employ the following techniques:
\begin{itemize}
    \item We adopt a truncated view transformation functionality to ensure that each outsourced data contributes to a bounded number of rows in the transformed view instance.
    \item We track important parameters in the $\trans$ phase, secretly share them inside the $\trans$ protocol and store the corresponding shares onto each server. The parameters are later passed to $\sync$ protocol as auxiliary input and used to generate the DP resized cardinalities. 
\end{itemize}
}
Algorithm~\ref{algo:precom} provides an overview of $\trans$ protocol.  At the very outset, $\trans$ is tasked to: (i) convert the newly submitted data into its corresponding view entry from time to time; (ii) cache the new view entries to an exhaustively padded secure array; and (iii) maintain a cardinality counter of how many new view entries have been cached since the last view update. This counter is then privately passed (through secret shares) to the $\sync$ protocol. 
\begin{algorithm}[]
\caption{$\trans$ protocol}
\begin{algorithmic}[1]
\Statex
\textbf{Input}: $\omega$ (truncation bound); $\mathcal{DS}_t$; $\boldsymbol{\sigma}$.
\If{t == 0}
\State \smash{$\share{c}=({x\xleftarrow[]{\text{rd}}\mathbb{Z}_{m}, x\oplus 0})$, $\share{c}\xRightarrow[]{} (\sx, \sy)$}
\EndIf
    \State $\Delta \mathcal{V} \gets \mathsf{trans\_truncate}(\mathcal{DS}_t, \omega)$
    \State \smash{$\share{c}\xLeftarrow[]{} (\sx, \sy)$, ${c} \gets \rec(\share{c})$}
    \State \smash{$c \gets c + \sum_{v_i \in \Delta\mathcal{V} \wedge v_i \neq \mathsf{dummy}} \mathbbm{1}$}
    \State $\share{c} \gets \gs(c)$, $\share{c}\xRightarrow[]{} (\sx, \sy)$
    \State \smash{$\boldsymbol{\sigma} \gets \boldsymbol{\sigma} || \Delta\mathcal{V}$}
\end{algorithmic}
\label{algo:precom}
\end{algorithm} 

At the very beginning, $\trans$ initializes the cardinality counter $c=0$ and secret shares it to both servers (Alg~\ref{algo:precom}:1-2). $\trans$ uses $\mathsf{trans\_truncate}$ (Alg~\ref{algo:precom}:3) operator to obliviously compute the new view tuples and truncate the contribution of each record at the same time. More specifically, we assume the output of this operator is stored in an exhaustively padded secure array $\Delta\mathcal{V}$, where each $\Delta\mathcal{V}[i]$ is a transformed tuple with an extra {\it isView} bit to indicate whether the corresponding tuple is a view entry ({\it isView=1}) or a dummy tuple ({\it isView=0}). Additionally, the operator ensures that
\begin{equation}\label{eq:trunc}
    \forall ds_i \in \mathcal{DS}_t, || g^{\omega}\left(\mathcal{DS}_t\right) - g^{\omega}\left(\mathcal{DS}_t - \{ds_i\}\right)||\leq\omega
\end{equation}
where $g^{\omega}(\cdot)\gets\mathsf{truncate}\left(\mathsf{new\_entry}(\cdot), \omega \right)$. This indicates any input data only affects at most $\omega$ rows in the truncated $\Delta\mathcal{V}$. Once the truncated outputs are available, $\trans$ updates and re-shares the cardinality counter $c$, then appends the truncated outputs to the secure cache $\boldsymbol{\sigma}$ (Alg~\ref{algo:precom}:5-7). 

\vspace{-1mm}
\boldparagraph{$q$-stable transformation.} We now provide the following lemmas with respect to the $q$-stable transformation.
\begin{lemma}\textup{(\textsc{$q$-Stable Transformation}~\label{def:bct} \textup{\cite[][Definition~2]{mcsherry2009privacy}})}. Let $T:\mathbb{D}\rightarrow\mathbb{D}$ to be a transformation, we say $T$ is $q$-stable, if for any two databases $\mathcal{D}_1$, $\mathcal{D}_2\in \mathbb{D}$, it satisfies
$|| T(\mathcal{D}_1) - T(\mathcal{D}_2)|| \leq q || \mathcal{D}_1 - \mathcal{D}_2|| $
\end{lemma}
\begin{lemma}\label{def:stable-privacy} Given $T$ is $q$-stable $T$, and an $\epsilon$-DP mechanism $\mathcal{M}$. The composite computation $\mathcal{M} \circ T$ implies $q\epsilon$-DP\textup{~\cite[][Theorem 2]{mcsherry2009privacy}}.
\end{lemma}
According to Lemma~\ref{def:bct}, it's clear that protocol $\trans$ is $q$-stable, and thus by Lemma~\ref{def:stable-privacy}, applying an $\epsilon$-DP mechanism over the outputs of $\trans$ protocol (the cached data) implies $q\epsilon$-DP over the original data. Therefore, if $q$ is constant, then the total privacy loss with respect to the input of $\trans$ is bounded.


\vspace{-1mm}
\boldparagraph{Contribution over time.} According to the overall architecture, $\trans$ is invoked repeatedly for transforming outsourced data into view entries at each time step. Thus having a $q$-stable $\trans$ does not immediately imply bounded privacy loss with respect to the logical database. There are certain cases where one record may contribute multiple times over time as the input to $\trans$. For example, suppose the servers maintain a join table on both Alice's and Bob's data. When Alice submits new data, the servers need to compute new join tuples between her new data and Bob's entire database, which results in some of Bob's data being used multiple times. This could eventually lead to unbounded privacy loss.
\vspace{-1mm}
\begin{theorem}\label{tm:bcompose} Given a set of transformations $T = \{T_i\}_{i\geq0}$, where each $T_i$ is a $q_i$-stable transformation. Let $\{\mathcal{M}_i\}_{i\geq 0}$ be a set of mechanisms, where each $\mathcal{M}_i$ provides $\epsilon_i$-differential privacy. Let another mechanism $\mathcal{M}(\mathcal{D})$ that executes each $\mathcal{M}_i$ using independent randomness with input $T_i(\mathcal{D})$. Then $\mathcal{M}$ satisfies $\epsilon$-DP, where 
\begin{equation}\label{eq:compose}
\begin{split}
\epsilon = \max_{u, \mathcal{D}}\left(\sum_{i ~:~ \tau_i(u) > 0} q_i\epsilon_i\right)
\end{split}
\end{equation}
and $\tau_i(u) = ||T_i(\mathcal{D}) - T_i(\mathcal{D} - \{u\})||$, denotes the contribution of record $u$ to the transformation $T_i$'s outputs.
\end{theorem}
\vspace{-1mm}
Theorem~\ref{tm:bcompose} shows that the overall privacy loss may still be infinite when applying the DP mechanisms over a composition of $q$-stable transformations (i.e. repeatedly invoke $\trans$). However, if the composed transformation $T$ is also $q$-stable, then one can still obtain bounded privacy loss as \smash{$\max_{u, \mathcal{D}}\left(\sum_{i ~:~ \tau_i(u) > 0} q_i\epsilon_i\right) \leq q\max(\epsilon_i)$}. 
Inspired by this, the following steps could help to obtain fixed privacy loss over time: First we assign a total contribution budget $b$ to each outsourced data $ds_i \in \mathcal{DS}$. As long as a record is used as input to $\trans$ (regardless of whether it contributes to generating a real view entry), it is consumed with a fixed amount of budget (equal to the truncation limit $\omega$). Then $\trans$ keeps track of the available contribution budget for each record over time and ensures that only data with a remaining budget is used. According to this design, the ``life time'' contribution of each outsourced data to the materialized view object is bounded by $b$. 





\vspace{-1mm}
\boldparagraph{Implementation of $\mathsf{trans\_truncate}$ operator.} 
Na\"ively, this operator can be implemented via two separate steps. For example, one may first apply an oblivious transformation (i.e. oblivious filter~\cite{bater2018shrinkwrap, chowdhury2019crypt}, join~\cite{zheng2017opaque, eskandarian2017oblidb}, etc.) without truncation over the input data. The results are stored to an exhaustively padded array. Next, truncation can be implemented by linearly scanning the array and resets the {\it isView} bit from 1 to 0 for a subset of the data in the array such that the resulting output satisfies Eq~\ref{eq:trunc}. In practice, truncation can be integrated with the view transformation so that the protocol does not have to run an extra round of linear scan. In what follows, we provide an instantiated implementation of oblivious sort-merge join where the output is truncated with a contribution bound $b$, and we continue to provide more implementations for other operators such as filter, nested-loop join, etc. in our complete full version.

\vspace{-2mm}
\begin{example}[$b$-truncated oblivious sort-merge join]\label{exmp:smj}
Assume two tables $T_1$, $T_2$ to be joined, the algorithm outputs the join table between $T_1$ and $T_2$ such that each data owns at most $b$ rows in the resulting join table. The first step is to union the two tables and then obliviously sort~\cite{batcher1968sorting} them based on the join attributes. To break the ties, we consider $T_1$ records are ordered before $T_2$ records. Then similar to a normal sort-merge join, where the operator linearly scans
the sorted merged table then joins $T_1$ records with the corresponding records in $T_2$. There are some variations
to ensure obliviousness and bounded contribution. First, the operator keeps track of the contribution of each individual tuple. If a tuple has already produced $b$ join entries, then any subsequent joins with this tuple will be automatically discarded. Second, for linear scan, the operator outputs $b$ entries after accessing each tuple in the merged table, 
regardless of how many true joins are generated. If there are fewer joins, then pad them with additional dummy data, otherwise truncate the true joins and keep only the $b$ tuples. Figure~\ref{fig:tjoin} illustrates the aforementioned computation workflow.
\end{example}
\begin{figure}[h]
\centering
\includegraphics[width=0.7\linewidth]{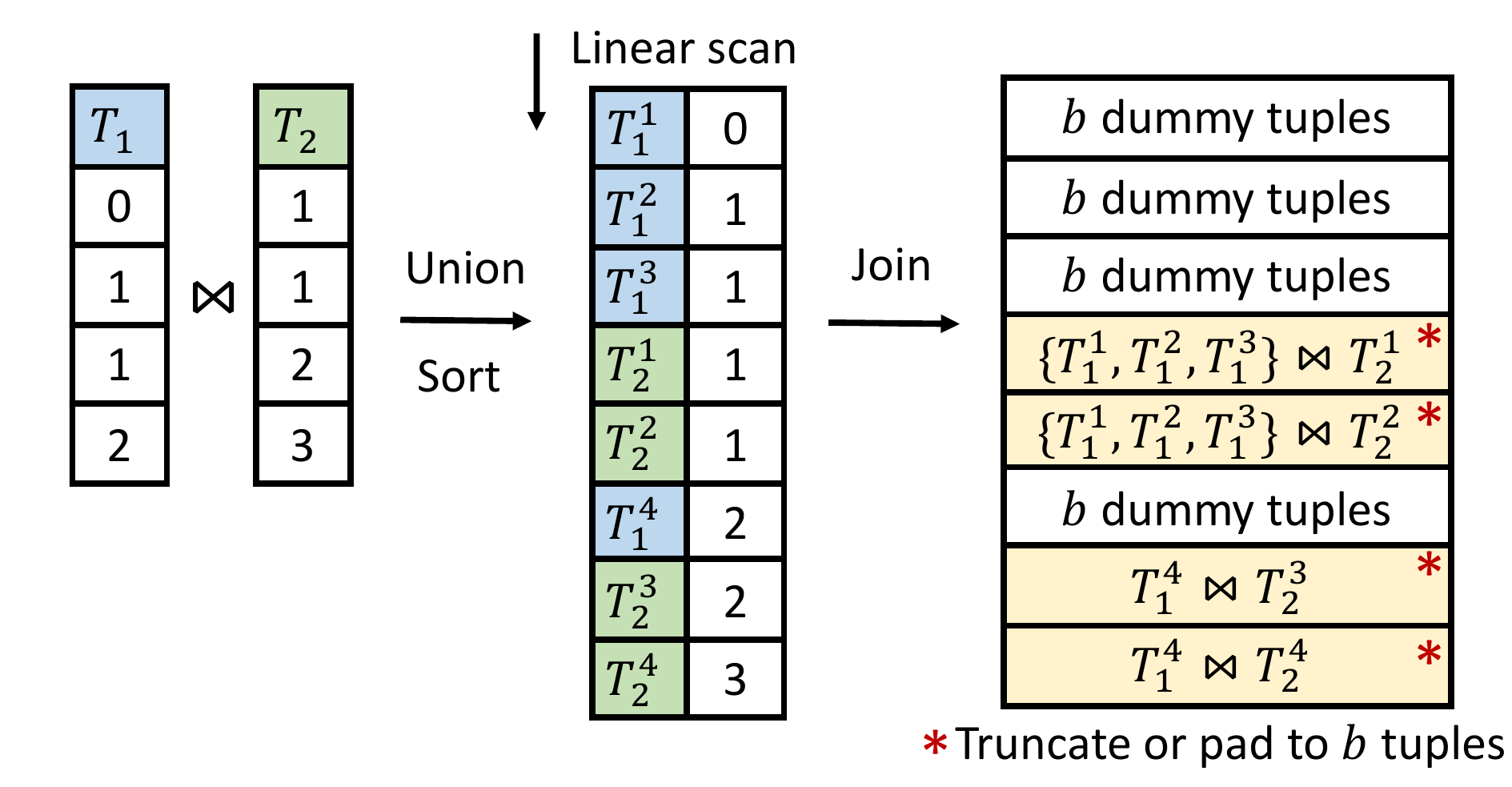}
\vspace{-4mm}
\caption{Oblivious truncated sort-merge join.}
\label{fig:tjoin}
\end{figure}
\vspace{-2mm}
\boldparagraph{Secret-sharing inside MPC.} When re-sharing the new cardinalities (Alg~\ref{algo:precom}:5-6), we must ensure none of the two servers can tamper with or predict the randomness for generating secret shares. This can be done with the following approach: each outsourcing server $\mathcal{S}_i$ chooses a value $z_i$ uniformly at random from the ring $\mathbb{Z}_{2^{32}}$, and contributes it as the input to $\trans$. The protocol then computes $\share{c}\gets\{c_0\gets z_0\oplus z_1, c_1\gets c_0\oplus c\}$ internally. By applying this, $\mathcal{S}_0$'s knowledge of the secret shared value is subject to the two random values $z_0, z_1$ while $\mathcal{S}_1$'s knowledge is bounded to $c\oplus z_0$ which is masked with a random value unknown to $\mathcal{S}_1$. A complete security proof of this technique can be found in our full version. 
\vspace{-1mm}
\subsection{$\sync$ Protocol}\label{sec:dp-st}
We propose two secure protocols that synchronize tuples from the secure cache to the materialized view. Our main idea originates from the {\it private synchronization strategy} proposed in DP-Sync~\cite{wang2021dp}, but with non-trivial variations and additional design. \re{ Typically, DP-Sync enforces trusted entities to execute private synchronization strategies, whereas in our scenario, the framework is supposed to allow untrusted servers to evaluate the view update protocol. Therefore, na\"ively adopting their techniques could lead to additional leakage and exposure to untrusted servers, such as the internal states (i.e., randomness) during protocol execution. Furthermore, DP-Sync considers that the subjects evaluating the synchronization strategies have direct access to a local cache in the clear, whereas in our setup the protocol must synchronize cached view tuples from an exhaustively padded (with dummy entries) secure array without knowing how the real data is distributed. To address these problems, we incorporate the following techniques:
\vspace{-1mm}
\begin{itemize}
    \item We utilize a joint noise-adding mechanism to generate DP noise, which ensures that no admissible adversary can obtain or tamper with the randomness used to generate the noise.  
    \item We implement a secure cache read operation that enforces the real data is always fetched before the dummy tuples when a cache read is performed.
\end{itemize}
}


In what follows, we review the technical details of two view update protocols, $\timer$ and $\ant$. 
\vspace{-2mm}
 \subsubsection{Timer-based approach ($\timer$)}\label{sec:dp-timer-detail}
 $\timer$ is a 2PC protocol among $\sx$, and $\sy$, parameterized by $T$, $\epsilon$ and $b$, where it updates the materialized view every $T$ time units with a batch of DP-sized tuples. Algorithm~\ref{algo:mpcdptimer} shows the corresponding details.
\begin{algorithm}[]
\caption{\timer}
\begin{algorithmic}[1]
\Statex
\textbf{Input}: $\epsilon$; $b$ (contribution budget); $T$ (update interval);  $\boldsymbol{\sigma}$; $\mathcal{V}$;
\For{$t \leftarrow 1,...$}
\If{$t \mod T == 0$}
\State recover $c$ internally
\State \smash{$(z_0, z_1)\xLeftarrow[]{} (\sx, \sy)$}, s.t. \smash{$\forall z_i \xleftarrow[]{\rd}\mathbb{Z}_{2^{32}}$}
\State $z\gets z_0 \oplus z_1$, $r \gets \mathsf{fixed\_point}(z)$, s.t. $r\in (0,1)$
\State \smash{$\texttt{sz} \gets c + \frac{b}{\epsilon}\ln{r}\times \mathsf{sign}(\msb(z))$} \Comment{$\texttt{sz}\gets\lap(\frac{b}{\epsilon})$}
\State $\boldsymbol{\hat{\sigma}}\gets \sort(\boldsymbol{\sigma}, \mathsf{key}=isView)$ 
\State ${\bf o}\gets \boldsymbol{\hat{\sigma}}[0,1,2,..,\texttt{sz}-1]$, $\mathcal{V}\gets\mathcal{V} \cup {\bf o}$, $\boldsymbol{\sigma}\gets\boldsymbol{\hat{\sigma}}[\texttt{sz},...]$
\State reset $c=0$ and re-share it to both servers.
\EndIf 
\EndFor
\end{algorithmic}
\label{algo:mpcdptimer}
\end{algorithm}

For every $T$ time steps $\timer$ obtains the secret-shared cardinality counter from both servers, and recovers the counter $c$ internally. 
The protocol then distort this cardinality with Laplace noise sampled from $\lap(\frac{b}{\epsilon})$. To prevent information leakage, we must ensure none of the entities can control or predict the randomness used to generate this noise. To this end, inspired by the idea in~\cite{dwork2006our}, we implement a joint noise generation approach (Alg~\ref{algo:mpcdptimer}:4-6), where each server generates a random value $z_i\in \mathbb{Z}_{2^{32}}$ uniformly at random and contributes it as an input to $\timer$. The protocol computes $z \gets z_0 \oplus z_1$ internally, and converts $z$ to a fixed-point random seed $r\in(0,1)$. Finally, $\timer$ computes \smash{$\lap(\frac{b}{\epsilon})\gets\frac{b}{\epsilon}\ln{r}\times\textup{sign}$}, using one extra bit of randomness to determine the sign, i.e. the most-significant bit of $z$. By applying this, as long as one server honestly chooses the value uniformly at random and does not share it with any others (which is captured by the non-colluding server setting), she can be sure that no other entity can know anything about the resulting noise. In our design, this joint noise adding technique is used whenever the protocol involves DP noise generation. For ease of notation, we denote this approach as $\tilde{x} \gets \mathsf{JointNoise}(\sx, \sy, \Delta, \epsilon, x)$, where $\tilde{x}\gets\lap(\frac{\Delta}{\epsilon})$.
\vspace{1mm}
\begin{figure}[h]
\centering
\includegraphics[width=0.9\linewidth]{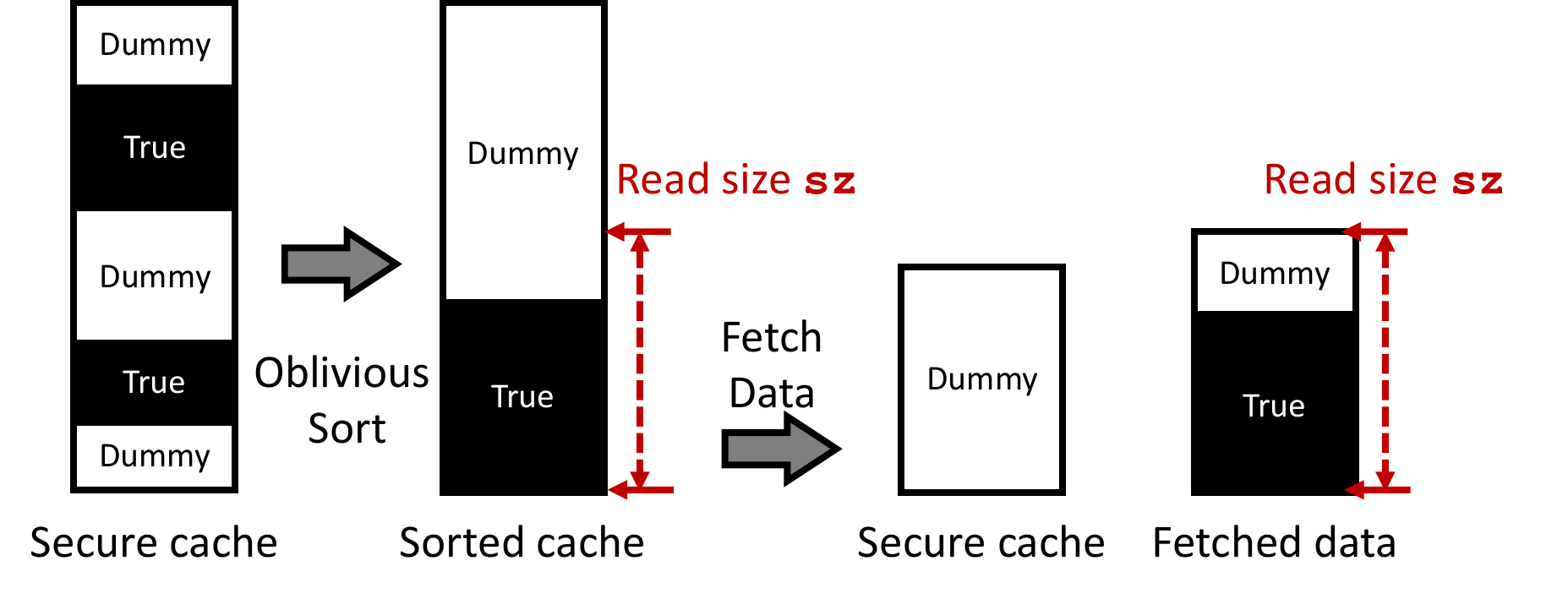}
\vspace{-5mm}
\caption{Cache read operation.}
\label{fig:cache}
\end{figure}

Next, $\timer$ obliviously sorts the exhaustively padded cache $\boldsymbol{\sigma}$ based on the $isView$ bit, moving all real tuples to the head and all dummy data to the tail, then cuts off the first $\texttt{sz}$ elements from the sorted array and stores them as a separate structure ${\bf o}$. $\timer$ then updates the materialized view $\mathcal{V}$ by appending ${\bf o}$ to the old view instance. Figure~\ref{fig:cache} shows an example of the aforementioned operation. Such secure array operation ensures that the real data is always fetched before the dummy elements, which allows us to eliminate a subset of the dummy data and shrink the size of the updated view entries. Finally, after each update, $\timer$ resets $c$ to 0 and re-shares it secretly to both participating servers. 


\re{
 Note that query errors of \system are caused by two factors, namely, the valid data discarded by the $\trans$ due to truncation constraints and the total amount of unsynchronized data in the cache. When the truncation bound is set too small, then a large number of valid view entries are dropped by $\trans$, resulting in inaccurate query answers. We further investigate how truncation bound would affect query errors in a later section (Section~\ref{sec:exp-omega}).
 
 However, one could still choose a relatively large truncation bound to ensure that no data is discarded by $\trans$. As a result, query errors under such circumstances will be caused primarily by unsynchronized cached view tuples. Typically, less unsynchronized cached data that satisfy the requesting query implies better accuracy and vice versa. For instance, when \system has just completed a new round of view update, the amount of cached data tends to be relatively small, and thus less data is missing from the materialized view. Queries issued at this time usually have better accuracy. 
}

\vspace{-2mm}
\begin{theorem}\label{lg:timer}
Given $\epsilon, b$, and \smash{$k \geq 4\log{\frac{1}{\beta}}$}, where $\beta \in (0,1)$. The number of deferred data $c_d$ after $k$-th updates satisfies \smash{$\textup{Pr}\left[c_d \geq \alpha \right] \leq \beta$}, where \smash{$\alpha=\frac{2b}{\epsilon}\sqrt{k\log{\frac{1}{\beta}}}$}.
\end{theorem}
As per Theorem~\ref{lg:timer}, we can derive the upper bound for total cached data at any time as \smash{$c^* + O(\frac{2b\sqrt{k}}{\epsilon})$}, where $c^*$ refers to the number of newly cached entries since last update, and the second term captures the upper bound for deferred data. The data in the cache, although stored on the server, is not used to answer queries. Therefore large amounts of cached data can lead to inaccurate query results. \re{
Although one may also adopt strategies such as returning the entire cache to the client, or scanning the cache while processing the query so that no data is missing. This will not break the security promise, however, will either increase the communication cost or the runtime overhead as the secure cache is exhaustively padded with dummy data. Moreover, in later evaluations (Section~\ref{sec:e2ecmp}), we observe that even without using the cache for query processing, the relative query errors are small (i.e. bounded by 0.04).
}

It's not hard to bound $c^*$ by picking a smaller interval $T$, however the upper bound for deferred data accumulates when $k$ increases. In addition, at each update, the server only fetches a batch of DP-sized data, leaving a large number of dummy tuples in the cache. Thus, to ensure bounded query errors and prevent the cache from growing too large, we apply an independent cache flushing mechanism to periodically clean the cache. To flush the cache, the protocol first sorts it, then fetches a set of data by cutting off a fixed number of tuples from the head of the sorted array. The fetched data is updated to the materialized view immediately and the remaining array is recycled (i.e. freeing the memory space).  As per Theorem~\ref{lg:timer}, we can set a proper flush size, such that with at most (a relatively small) probability $\beta$ there is non-dummy data been recycled.
\begin{theorem}\label{perf:timer}
Given $\epsilon, b$, and \smash{$k \geq 4\log{\frac{1}{\beta}}$}, where $\beta \in (0,1)$. Suppose the cache flush interval is $f$ with flush size $s$. Then the number of data entries inserted to the materialized view after the $k$-th update is bounded by \smash{$O(\frac{2b\sqrt{k}}{\epsilon}) + \frac{skT}{f}$}.
\end{theorem}

\vspace{-1mm}
\subsubsection{Above noisy threshold. ($\ant$)} The above noise threshold protocol (Algorithm~\ref{algo:mpcdpant}) takes $\theta$, $\epsilon$ and $b$ as parameters and updates the materialized view whenever the number of new view entries is approximately equal to a threshold $\theta$.
\begin{algorithm}[]
\caption{$\ant$}
\begin{algorithmic}[1]
\Statex
\textbf{Input}:  $\epsilon$; $b$; $\theta$ (sync threshold); $\boldsymbol{\sigma}$; $\mathcal{V}$;
\State $\epsilon_1 = \epsilon_2 =\frac{\epsilon}{2}$
\State $\tilde{\theta} \gets \mathsf{JointNoise}(\sx, \sy, b, \epsilon_1/2, \theta)$\Comment{distort the threshold}
\State $\share{\tilde{\theta}} \gets \gs(\tilde{\theta})$, $\share{\tilde{\theta}}\xRightarrow[]{} (\sx, \sy)$
\For{$t \leftarrow 1,...$}
\State recover $c$, and $\tilde{\theta}$ internally
\State $\tilde{c} \gets \mathsf{JointNoise}(\sx, \sy, b, \epsilon_1/4, c)$
\If{$\tilde{c}\geq \tilde{\theta}$}\Comment{updates if greater than noisy threshold}
    \State $\texttt{sz} \gets \mathsf{JointNoise}(\sx, \sy, b, \epsilon_2, c)$ 
    \State $\boldsymbol{\hat{\sigma}}\gets \sort(\boldsymbol{\sigma}, \mathsf{key}=isView)$
    \State ${\bf o}\gets \boldsymbol{\hat{\sigma}}[0,1,2,..,\texttt{sz}-1]$, $\mathcal{V}\gets\mathcal{V} \cup {\bf o}$, $\boldsymbol{\sigma}\gets\boldsymbol{\hat{\sigma}}[\texttt{sz},...]$
    \State $\tilde{\theta} \gets \mathsf{JointNoise}(\sx, \sy, b, \epsilon_1/2, \theta)$
     \State \smash{$\share{\tilde{\theta}} \gets \gs(\tilde{\theta})$, $ \share{\tilde{\theta}}\xRightarrow[]{} (\sx, \sy)$}
    \State reset $c=0$ and re-share it to both servers.
    \EndIf
\EndFor
\end{algorithmic}
\label{algo:mpcdpant}
\end{algorithm}

Initially, the protocol splits the overall privacy budget $\epsilon$ to two parts $\epsilon_1$, and $\epsilon_2$, where $\epsilon_1$ is used to construct the noisy condition check (Alg~\ref{algo:mpcdpant}:7) and $\epsilon_2$ is used to distort the true cardinalities (Alg~\ref{algo:mpcdpant}:8). The two servers then involve a joint noise adding protocol that securely distort $\theta$ with noise $\lap(\frac{2b}{\epsilon_1})$. This noisy threshold will remain unchanged until $\ant$ issues a new view update. An important requirement of this protocol is that such noisy threshold must remain hidden from untrusted entities. Therefore, to cache this value, $\ant$ generates secret shares of $\tilde{\theta}$ internally and disseminates the corresponding shares to each server (Alg~\ref{algo:mpcdpant}:3).

From then on, for each time step, the protocol gets the secret shares $\share{c}$ (true cardinality) and $\share{\tilde{\theta}}$ from $\sx$ and $\sy$, which are subsequently recovered inside the protocol. $\ant$ distorts the recovered cardinality \smash{$\tilde{c} \gets \lap(\frac{4b}{\epsilon_1})$} and compares the noisy cardinality $\tilde{c}$ with $\tilde{\theta}$. A view update is posted if $\tilde{c}\geq\tilde{\theta}$.  By issuing updates, $\ant$ distorts $c$ with another Laplace noise \smash{$\lap(\frac{b}{\epsilon_2})$} to obtain the read size $\texttt{sz}$. Similar to $\timer$, it obliviously sorts the secure cache and fetches as many tuples as specified by $\texttt{sz}$ from the head of the sorted cache. Note that, each time when an update is posted, $\ant$ must re-generate the noisy threshold with fresh randomness. Therefore, after each updates, $\ant$ resets $c=0$, produces a new $\tilde{\theta}$, and updates the corresponding secret shares stored on the two servers (Alg~\ref{algo:mpcdpant}:11-13). In addition, the same cache flush method in $\timer$ can be adopted by $\ant$ as well. The following theorem provides an upper bound on the cached data at each time, which can be used to determine the cache flush size.
\begin{theorem}\label{one:ant}
Given $\epsilon, b$, and let the cache flushes every $f$ time steps with fixed flushing size $s$, the number of deferred data at any time $t$ is bounded by \smash{$O(\frac{16b\log{t}}{\epsilon})$} and the total number of dummy data inserted to the materialized view is bounded by \smash{$O(\frac{16b\log{t}}{\epsilon}) + s\lfloor\frac{t}{f}\rfloor$}.
\end{theorem}

\eat{
\begin{algorithm}[]
\caption{\texttt{SecNewJoin}}
\begin{algorithmic}[1]
\Function{$\texttt{SecNewJoin}$}{$B^0_t \cup \gamma^0_t, B^1_t \cup \gamma^1_t, \bot$}
\State $\pi \gets B^0_t \bowtie \gamma^0_t \cup B^1_t \bowtie \gamma^0_t \cup \gamma^0_t \bowtie \gamma^1_t$ \Comment{compute new joins}
\State $\pi \gets \textup{Padd}(\pi, \textup{MAXSIZE})$ \Comment{Padding $\pi$ with dummy joins}
\State {\bf return} $(\bot, \bot, \pi)$
\EndFunction
\end{algorithmic}
\label{algo:secnewjoin}
\end{algorithm}
}

\vspace{-1mm}
\section{Security Proof}\label{sec:proof}
We provide the security and privacy proofs in this section.
\begin{theorem}\label{tm:pf-timer}
\system with $\timer$ satisfies Definition~\ref{def:protocol-sogdb}.
\end{theorem}
\begin{proof}\label{pf:dp-timer}
We prove this theorem by first providing a mechanism $\mathcal{M}$ that simulates the update pattern leakage of the view update protocol and proving that $\mathcal{M}$ satisfies $\epsilon$-DP. Second, we construct a {\it p.p.t} simulator that accepts as input only the output of $\mathcal{M}$ which can simulate the outputs that are computationally indistinguishable compared to the real protocol execution. In what follows, we provide the $\mathcal{M}_{\mathsf{timer}}$ that simulates the update pattern of $\timer$.
\begin{align*}
\begin{array}{l}
\mathcal{M}_{\mathsf{timer}}\\
\forall~ t ~:
\end{array}~\  
\begin{array}{l}
\textbf{return} ~\texttt{count}(\sigma_{t-T<t_{tid}}(\mathcal{D})) + \lap(\frac{b}{\epsilon}), \textbf{if}~0\equiv t\Mod T \\
\textbf{return} ~0, \text{otherwise}
\end{array}
\end{align*}
where $t_{tid}$ denotes the time stamp when tuple $rid$ is inserted to $\mathcal{D}$, and $\sigma_{t-T<t_{tid}}$ is a filter operator that selects all tuples inserted within the time interval $(t-T,t]$. In general, $\mathcal{M}_{\mathsf{timer}}$ can be formulated as a series of $\frac{\epsilon}{b}$-DP Laplace mechanisms that applies over disjoint data (tuples inserted in non-overlapping intervals). Thus by parallel composition theorem~\cite{dwork2014algorithmic}, $\mathcal{M}_{\mathsf{timer}}$ satisfies $\frac{\epsilon}{b}$-DP. Moreover, by Lemma~\ref{def:stable-privacy} given a $q$-stable transformation $\hat{T}$ such that $q=b$ (i.e. the $\trans$ protocol), then $\mathcal{M}_{\mathsf{timer}}(\hat{T}(\mathcal{D}))$ achieves $b\times\frac{\epsilon}{b} = \epsilon$-DP. 
We abstract $\mathcal{M}_{\mathsf{timer}}$'s output as $ \{\left(t, v_t\right)\}_{t\geq 0}$, where $v_t$ is the number released by $\mathcal{M}_{\mathsf{timer}}$ at time $t$. Also we assume the following parameters are publicly available: $\epsilon$, $\mathbb{Z}_{m}$, $C_r$ (batch size of owner uploaded data), $b$ (contribution bound), $s$ (cache flush size), $f$ (cache flush rate), $T$ (view update interval). In what follows, we construct a {\it p.p.t.} simulator $\mathcal{S}$ that simulates the protocol execution with only access to $\mathcal{M}_{\mathsf{timer}}$'s outputs and public parameters (Table~\ref{tab:sim1}).

\begin{table}[]
\scalebox{0.98}{\small
\begin{tabular}{|lll|}
\hline
\multicolumn{3}{|c|}{Simulator $\mathcal{S} \left(t, v_t, \epsilon, C_r, b, s, f, T\right)$}                      \\
\multicolumn{3}{|l|}{1. Initialize internal storage $B \gets \{\emptyset\}$.}            \\
\multirow{5}{*}{2. $\forall t > 0$}   & \multicolumn{2}{l|}{i. $\smash{(B_1, B_2) \xleftarrow[]{\text{rd}}\mathbb{Z}_{m}~:~|B_1|=C_r, |B_2|=bC_r}$ } \\
                           & \multicolumn{2}{l|}{ii. ${B}_3 \gets \mathsf{rd\_fetch}({B}_2 \cup B)~:~|{B}_3|=v_t$} \\
                           & \multicolumn{2}{l|}{iii. \textbf{reveal} $\smash{({B_1}, {B_2}, {B}_3, x\xleftarrow[]{\text{rd}}\mathbb{Z}_{m}) ~\textbf{if}~0\equiv t\Mod T}$} \\
                           &         & a.   \smash{$B' \gets\mathsf{rd\_fetch}(B)~:~|B'|=s$}      \\
                           & iv. $\textbf{if}~0 \equiv f \Mod t$             & b. $B \gets \{\emptyset\}$        \\
                           &              & c. \textbf{reveal} \smash{$({B_1}, {B_2}, {B}_3,B', x\xleftarrow[]{\text{rd}}\mathbb{Z}_{m})$}        \\
                           & \multicolumn{2}{l|}{v. \textbf{reveal} $({B_1}, {B_2}, {B_3})$ ~\textbf{otherwise}} \\ \hline
\end{tabular}
}
\caption{Simulator construction ($\timer$)}
\vspace{-4mm}
\label{tab:sim1}
\end{table}
\eat{
\begin{align*}
\begin{array}{l}
\mathcal{S} \left(t, v_t, \epsilon, C_r, b, s, f, T\right):\\
1. B \gets \{\emptyset\}\\
2. \forall t > 0:\\
3. \smash{(B_1, B_2) \xleftarrow[]{\text{rd}}\mathbb{Z}_{m}~:~|B_1|=C_r, |B_2|=bC_r}\\
4. {B}_3 \gets \mathsf{rd\_fetch}({B}_2 \cup B)~:~|{B}_3|=v_t\\
\textbf{Outputs}~ \\
\smash{({B_1}, {B_2}, {B}_3, x\xleftarrow[]{\text{rd}}\mathbb{Z}_{m}) ~\textbf{if}~0\equiv t\Mod T}\\
({B_1}, {B_2}, {B_3}) ~\textbf{otherwise}
\end{array}
\end{align*}
}

Initially, $\mathcal{S}$ initializes the internal storage $B$. Then for each time step,  $\mathcal{S}$ randomly samples 2 batches of data $B_1$, $B_2$ from $\mathbb{Z}_{m}$, where the cardinality of $B_1$, and $B_2$ equals to $C_r$, and $bC_r$, respectively. These two batches simulate the secret shared data uploaded by owners and the transformed tuples placed in cache at time $t$. Next, $\mathcal{S}$ appends $B_2$ to $B$ and then  samples $B_3$ from $B$ such that the resulting cardinality of $B_3$ equals to $v_t$ (if $v_t=0$ then $\mathcal{S}$ sets $B_3=\{\emptyset\}$). This step simulates the data synchronized by $\sync$ protocol. Finally, if $t\Mod T=0$, $\mathcal{S}$ generates one additional random value $x$ to simulate the secret share of the cardinality counter and reveals $x$ together with the generated data batches to the adversary (2.iii). If $0\equiv f\Mod t$, then $\mathcal{S}$ performs another random sample from internal storage $B$ to fetch $B'$ such that $|B'|=s$, followed by resetting $B$ to empty. $\mathcal{S}$ reveals $B_1, B_2, B_3, B'$ and one additional random value $x$ to the adversary. These steps (2.iv) simulate the cache flush. Otherwise $\mathcal{S}$ only reveals $B_1, B_2$ and $B_3$. The computational indistinguishability between the $B_1$, $B_2$, $B_3, x$ and the messages the adversary can obtain from the real protocol follows the security of $(2,2)$-secret-sharing scheme and the security of secure 2PC~\cite{lindell2017simulate}. 
\end{proof}
\vspace{-3mm}
\begin{theorem}\label{tm:pfant}
\vspace{-3mm}
\system with $\ant$ satisfies Definition~\ref{def:protocol-sogdb}.
\end{theorem}
\begin{proof}Following the same proof paradigm, we first provide $\mathcal{M}_{\mathsf{ant}}$ that simulates the view update pattern under $\ant$.
\begin{align*}
\begin{array}{l}
\mathcal{M}_{\mathsf{ant}}\\
\mathcal\forall~ t ~: ~
\end{array}\  
\begin{array}{l}
\tilde{\theta} \gets \theta + \lap(\frac{4b}{\epsilon}), \text{if}~t=0\\
c_t \gets \texttt{count}(\sigma_{t^* < t_{tid} \leq t}(\mathcal{D})), \tilde{c}_t \gets c_t + \lap(\frac{8b}{\epsilon})\\
\textbf{return} ~c_t + \lap(\frac{4b}{\epsilon}), \theta + \lap(\frac{4b}{\epsilon}), \text{if}~\tilde{c}_t\geq\tilde{\theta}\\
\textbf{return} ~0, \text{if}~\tilde{c}_t < \tilde{\theta}
\end{array}
\end{align*}
where $\sigma_{t^* < t_{tid}\leq t}$ is a filter that selects all data received since last update. In general, $\mathcal{M}_{\mathsf{ant}}$ is a mechanism utilizes sparse vector techniques (SVT) to periodically release a noisy count. According to~\cite{wang2021dp} (Theorem 11), this mechanism satisfies \smash{$\frac{\epsilon}{b}$}-DP (interested readers may refer to our full version for complete privacy proof of $\mathcal{M}_{\mathsf{ant}}$). As per Lemma~\ref{def:stable-privacy} we can obtain the same conclusion that $\mathcal{M}_{\mathsf{ant}}(\hat{T}(\mathcal{D}))$ achieves $\epsilon$-DP, if $\hat{T}$ is $q$-stable and $q=b$. 
Similarly, we abstract $\mathcal{M}_{\mathsf{ant}}$'s output as $ \{\left(t, v_t\right)\}_{t\geq 0}$ and we assume the following parameters are publicly available: $\epsilon$, $C_r$, $b$, $s$, $f$, $\theta$ (threshold). For simulator construction one can reuse most of the components in Table~\ref{tab:sim1} but with the following modifications: (i) For step 2.iii, one should replace the condition check as whether $v_t>0$. (ii) For step 2.iii and 2.iv, the simulator outputs one additional random value \smash{$y\xleftarrow[]{\text{rd}}\mathbb{Z}_{m})$} which simulates the secret shares of refreshed noisy threshold. Similar, the indistinguishability follows the security property of XOR-based secret-sharing scheme.

\end{proof}

\vspace{-6mm}
\section{Experiments}\label{sec:exp}
In this section, we present evaluation results of our proposed framework. Specifically, we address the following questions:
\begin{itemize}
    \item \textbf{Question-1}: Do the view-based query answering approaches have efficiency advantages over the non-materialization (NM) approach? Also, how does the DP-based view update protocol compare with the na\"ive ones? 
    \item \textbf{Question-2}: For DP-protocols, is there a trade-off between privacy, efficiency and accuracy? Can we adjust the privacy parameters to achieve different efficiency or accuracy goals?
    \item \textbf{Question-3}: How do $\timer$ compare to the $\ant$? Under what circumstances is one better than the other one?
\end{itemize}
\vspace{-2mm}
\boldparagraph{Implementation and configuration.} We implement a prototype \system, and evaluate it with real-world datasets. We build the prototype \system based on Shrinkwrap~\cite{bater2018shrinkwrap}, a typical secure outsourced database scheme under the server-aided MPC setting. Shrinkwrap only supports static data and a standard query answering method. \system extends it to a view-based SOGDB that handles growing data.
In addition to the prototype \system, we also implement client programs that consume data from the given datasets and outsource them to the server, which simulate how real-world data owner devices would receive and outsource new data. We implement all secure 2PC protocols using EMP-Toolkit-0.2.1 package and conduct all experiments on the GCP instance with 3.8GHz Xeon CPU, 32Gb RAM, and 64 bit Ubuntu 18.04.1 OS. 

\vspace{-2mm}
\boldparagraph{Data.}
We evaluate the system using  two datasets: {\it TPC Data Stream (TPC-ds)}~\cite{tpcds}, and {\it Chicago Police Database (CPDB)}~\cite{cpdb}. TPC-ds collects the retail records for several product suppliers over a five-year period. In our evaluation, for TPC-ds, we mainly use two relational tables, the {\it Sales} and the {\it Return} table. After eliminating invalid data points with incomplete or missing values, the {\it Sales} and {\it Return} tables contain 2.2 million and 270,000 records, respectively.  CPDB is a living repository of public data about Chicago’s police officers and their interactions with the public. We primarily use two relations, the {\it  Allegation} table, which documents the results of investigations into allegations of police misconduct, and the {\it Award} table, which collects information on awards given to certain officers. The cleaned data (after eliminating invalid entries) contains 206,000 and 656,000 records for {\it Allegation} and {\it Award} table, respectively.


\vspace{-2mm}
\boldparagraph{Execution scenario \& Testing query.}  For TPC-ds data, we delegate each relational table to a client program, which then independently outsources the data to the servers. We multiplex the sales time ({\it Sales} table) or return time ({\it Return} table) associated with each data as an indication of when the client received it. In addition, we assume that the client program uploads a batch of data every single day and the uploaded data is populated to the maximum size.
In addition, we pick the following query for evaluating with TPC-ds.
\begin{itemize}
    \item Q1-Count the total number of products returned within 10 days after purchasing: ``\texttt{SELECT COUNT(*) FROM Sales INNER JON Returns ON Sales.PID = Returns.PID\\ WHERE Returns.ReturnDate - Sales.SaleDate <= 10}''
\end{itemize}
According to the testing query, we set the materialized view as a join table for all products that returned within 10 days. Furthermore, as Q1 has multiplicity 1, thus we set the truncation bound as $\omega=1$, and the total contribution budget for each data as $b=10$.

For CPDB data, we consider only {\it  Allegation} table is private and is delegated to a client program. The {\it  Award} table will be treated as a public relation. Again, we use the investigation case end time to indicate when the client program received this record, and we assume that the client outsource data once every 5 days (minimum time span is 5 days), and the data is padded to maximum possible size as well. For evaluation, we select the following query.

\begin{itemize}
    \item Q2-Count how many times has an officer received an award from the department despite the fact that the officer had been found to have misconduct in the past 10 days: ``\texttt{SELECT COUNT(*) FROM Allegation INNER JON Award ON \\ Allegation.officerID = Award.officerID\\ WHERE Award.Time - Allegation.officerID <= 10}''.
\end{itemize}
Similarly, the materialized view is a join table that process Q2. We set the truncation bound $\omega=10$ and budget for each data as $b=20$.

\vspace{-2mm}
\boldparagraph{Default setting.} Unless otherwise specified, we assume the following default configurations. For both DP protocols, we set the default privacy parameter $\epsilon=1.5$, and cache flush parameters as $f=2000$ (flush interval) and $s=15$ (flush size). We fix the $\ant$ threshold $\theta$ as 30 for evaluating both datasets. Since the average number of new view entries added at each time step is 2.7 and 9.8, respectively for TPC-ds and CPDB, thus for consistency purpose we set the timer $T$ to $10\gets\lfloor\frac{30}{2.7}\rfloor$ and $3\gets\lfloor\frac{30}{9.8}\rfloor$. For each test group, we issue one test query at each time step and report the average L1 error and query execution time (QET) for all issued testing queries.

\vspace{-1mm}
\subsection{End-to-end Comparison}\label{sec:e2ecmp}
We address {\bf Question-1} by performing a comparative analysis between DP protocols, na\"ive protocols (one-time materialization and exhaustive padding method), and the non-materialization approach (standard SOGDB model~\cite{wang2021dp}). The comparison results are summarized in Table~\ref{tab:agg} and Figure~\ref{fig:cmp}.
\begin{table}[]
\scalebox{0.86}{\small
\begin{tabular}{lllccccc}
\hline
\multicolumn{3}{|c|}{\textbf{Comparison Cat.}} &
  \textbf{DP-Timer} &
  \textbf{DP-ANT} &
  \textbf{OTM} &
  \textbf{EP} &
  \multicolumn{1}{c|}{\textbf{NM}} \\ \hline
\multicolumn{8}{|c|}{\textbf{Average query error}} \\ \hline
\multicolumn{1}{|l|}{\textbf{}} &
  \multicolumn{2}{l|}{\textbf{L1 Error}} &
  40.02 &
  32.01 &
  2008.92 &
  0 &
  \multicolumn{1}{c|}{0} \\
\multicolumn{1}{|l|}{\textbf{TPCds}} &
  \multicolumn{2}{l|}{\textbf{Relative Error}} &
  0.03 &
  0.029 &
  1 &
  N/A &
  \multicolumn{1}{c|}{N/A} \\
\multicolumn{1}{|l|}{} &
  \multicolumn{2}{l|}{\textbf{Imp.}$^\dag$} &
  {\color[HTML]{009901} \textbf{50$\times$}} &
  {\color[HTML]{009901} \textbf{63$\times$}} &
  {\color[HTML]{3531FF} \textbf{1$\times$}}$^\ddag$ &
  N/A &
  \multicolumn{1}{c|}{N/A} \\ \hline
\multicolumn{1}{|l|}{\textbf{}} &
  \multicolumn{2}{l|}{\textbf{L1 Error}} &
  61.93 &
  52.45 &
  6595.6 &
  0 &
  \multicolumn{1}{c|}{0} \\
\multicolumn{1}{|l|}{\textbf{CPDB}} &
  \multicolumn{2}{l|}{\textbf{Relative Error}} &
  0.043 &
  0.038 &
  1 &
  N/A &
  \multicolumn{1}{c|}{N/A} \\
\multicolumn{1}{|l|}{} &
  \multicolumn{2}{l|}{\textbf{Imp.}} &
  {\color[HTML]{009901} \textbf{107$\times$}} &
  {\color[HTML]{009901} \textbf{126$\times$}} &
  {\color[HTML]{3531FF} \textbf{1$\times$}} &
  N/A &
  \multicolumn{1}{c|}{N/A} \\ \hline
\multicolumn{8}{|c|}{\textbf{Average execution time (s)}} \\ \hline
\multicolumn{1}{|l|}{} &
  \multicolumn{2}{l|}{\textbf{$\trans$}} &
  9.72 &
  9.69 &
  N/A &
  9.71 &
  \multicolumn{1}{c|}{N/A} \\
\multicolumn{1}{|l|}{} &
  \multicolumn{2}{l|}{\textbf{$\sync$}} &
  0.34 &
  0.37 &
  N/A &
  N/A &
  \multicolumn{1}{c|}{N/A} \\
\multicolumn{1}{|l|}{\textbf{TPC-ds}} &
  \multicolumn{2}{l|}{\textbf{QET}} &
  0.051 &
  0.052 &
  0 &
  5.84 &
  \multicolumn{1}{c|}{7982} \\
\multicolumn{1}{|l|}{} &
  \multicolumn{2}{l|}{\textbf{Imp. (over NM)}} &
  {\color[HTML]{009901} \textbf{1.5e+5$\times$}} &
  {\color[HTML]{009901} \textbf{1.5e+5$\times$}} &
  N/A &
  {\color[HTML]{333333} 1366$\times$} &
  \multicolumn{1}{c|}{{\color[HTML]{3531FF} \textbf{1$\times$}}} \\
\multicolumn{1}{|l|}{} &
  \multicolumn{2}{l|}{\textbf{Imp. (over EP)}} &
  {\color[HTML]{009901} \textbf{115$\times$}} &
  {\color[HTML]{009901} \textbf{112$\times$}} &
  N/A &
  {\color[HTML]{3531FF} \textbf{1$\times$}} &
  \multicolumn{1}{c|}{{\color[HTML]{333333} N/A}} \\ \hline
\multicolumn{1}{|l|}{} &
  \multicolumn{2}{l|}{\textbf{$\trans$}} &
  2.93 &
  2.91 &
  0 &
  2.93 &
  \multicolumn{1}{c|}{N/A} \\
\multicolumn{1}{|l|}{} &
  \multicolumn{2}{l|}{\textbf{$\sync$}} &
  3.93 &
  3.77 &
  N/A &
  N/A &
  \multicolumn{1}{c|}{N/A} \\
\multicolumn{1}{|l|}{\textbf{CPDB}} &
  \multicolumn{2}{l|}{\textbf{QET}} &
  0.17 &
  0.17 &
  0 &
  51.36 &
  \multicolumn{1}{c|}{1341} \\
\multicolumn{1}{|l|}{} &
  \multicolumn{2}{l|}{\textbf{Imp. (over NM)}} &
  {\color[HTML]{009901} \textbf{7888$\times$}} &
  {\color[HTML]{009901} \textbf{7888$\times$}} &
  N/A &
  26.1$\times$&
  \multicolumn{1}{c|}{{\color[HTML]{3531FF} \textbf{1$\times$}}} \\
\multicolumn{1}{|l|}{} &
  \multicolumn{2}{l|}{\textbf{Imp. (over EP)}} &
  {\color[HTML]{009901} \textbf{302$\times$}} &
  {\color[HTML]{009901} \textbf{302$\times$}} &
  N/A &
  {\color[HTML]{3531FF} \textbf{1$\times$}} &
  \multicolumn{1}{c|}{{\color[HTML]{333333} N/A}} \\ \hline
\multicolumn{8}{|c|}{\textbf{Materialized view size (Mb)}} \\ \hline
\multicolumn{1}{|l|}{\textbf{TPC-ds}} &
  \multicolumn{2}{l|}{\textbf{Avg. Size}} &
  2.01 &
  2.04 &
  0.01 &
  229.65 &
  \multicolumn{1}{c|}{N/A} \\
\multicolumn{1}{|l|}{\textbf{}} &
  \multicolumn{2}{l|}{\textbf{Imp.}} &
  {\color[HTML]{009901} \textbf{114$\times$}} &
  {\color[HTML]{009901} \textbf{113$\times$}} &
  N/A &
  {\color[HTML]{3531FF} \textbf{1$\times$}} &
  \multicolumn{1}{c|}{N/A} \\ \hline
\multicolumn{1}{|l|}{\textbf{CPDB}} &
  \multicolumn{2}{l|}{\textbf{Avg. Size}} &
  6.63 &
  6.68 &
  0.01 &
  2017.38 &
  \multicolumn{1}{c|}{N/A} \\
\multicolumn{1}{|l|}{} &
  \multicolumn{2}{l|}{\textbf{Imp.}} &
  {\color[HTML]{009901} \textbf{304$\times$}} &
  {\color[HTML]{009901} \textbf{302$\times$}} &
  N/A &
  {\color[HTML]{3531FF} \textbf{1$\times$}} &
  \multicolumn{1}{c|}{N/A} \\ \hline
\multicolumn{8}{l}{$^\dag$ {\bf Imp.} denotes the improvements; $^\ddag$ {\color[HTML]{3531FF} \textbf{1$\times$}} denotes the comparison baseline;}
\end{tabular}}
\caption{Aggregated statistics for comparison experiments}
\label{tab:agg}
\vspace{-4mm}
\end{table}

\vspace{-1mm}
\boldparagraph{Observation 1. View-based query answering provides a significant performance improvements over NM method.} As per Table~\ref{tab:agg}, we observe that the non-materialization method is the least efficient group among all groups. In terms of the average QET, the DP protocols achieve performance improvements of up to 1.5e+5$\times$ and 7888$\times$ on TPC-ds and CPDB data, respectively, in contrast to the NM approach. Even the EP method provides a performance edge of up to 1366$\times$ over NM approach. This result further demonstrates the necessity of adopting view-based query answering mechanism. A similar observation can be learned from Figure~\ref{fig:cmp} as well, where in each figure we compare all test candidates along with the two dimensions of accuracy (x-axis) and efficiency (y-axis). In all figures, the view-based query answering groups lie beneath the NM approach, which indicates better performance.
\begin{figure}[ht]
\captionsetup[sub]{font=small,labelfont={bf,sf}}
    \begin{subfigure}[b]{0.48\linewidth}
    \centering    \includegraphics[width=1\linewidth]{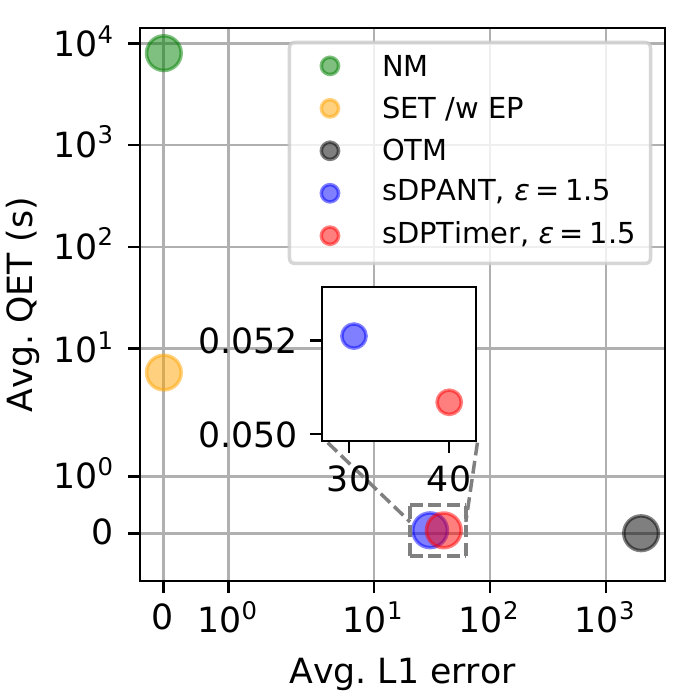}
        \caption{TPC-ds group}
        \label{fig:cmp-acc}\end{subfigure}
      \begin{subfigure}[b]{0.48\linewidth}
    \centering    \includegraphics[width=1\linewidth]{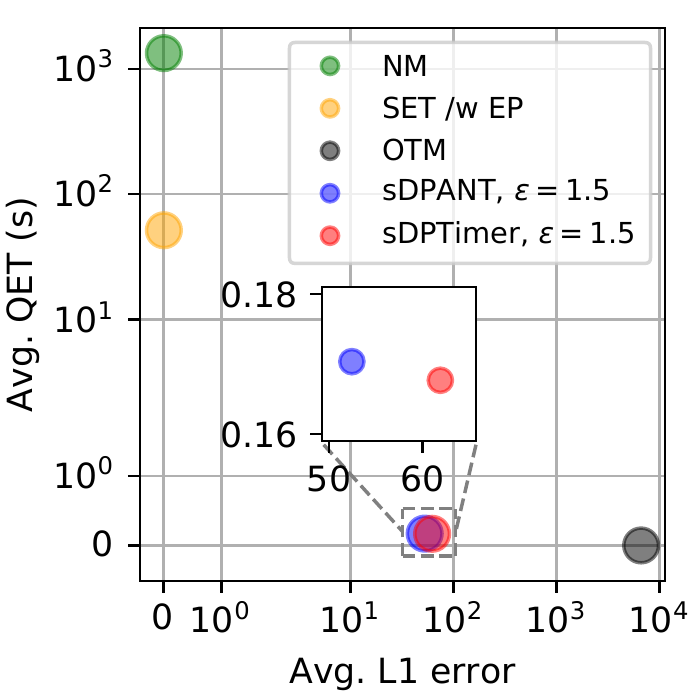}  
        \caption{CPDB group}
        \label{fig:cmp-perf}
    \end{subfigure}
    \vspace{-3mm}
   \caption{End-to-end comparison.}
   \label{fig:cmp}
\end{figure}

\vspace{-2mm}
\boldparagraph{Observation 2. DP protocols provide a balance between the two dimensions of accuracy and efficiency.} 
According to Table~\ref{tab:agg}, the DP protocols demonstrate at least 50$\times$ and 107$\times$ accuracy advantages (in terms of L1-error), respectively for TPC-ds and CPDB, over the OTM method. Meanwhile, in terms of performance, the DP protocols show a significant improvement in contrast to the EP method. For example, in TPC-ds group, the average QETs of both $\timer$ (0.051s) and $\ant$ (0.052s) are almost 120$\times$ smaller than that of EP method (5.84s). Such performance advantage is even evident (up to 302$\times$) over the CPDB data as testing query Q2 has join multiplicity greater than 1. Although, the DP approaches cannot achieve a complete accuracy guarantee, the average relative errors of all tested queries under DP protocols are below 4.3\%. These results are sufficient to show that DP approaches do provide a balance between accuracy and efficiency. This conclusion can be better illustrated with Figure~\ref{fig:cmp}, where we can observe that EP and OTM are located in the upper left and lower right corners of each plot, respectively, which indicates that they either completely sacrifice efficiency (EP) or accuracy (OTM) guarantees. However, both DP methods lie at the bottom-middle position of both figures, which further reveals that the DP protocols are optimized for the dual objectives of accuracy and efficiency. 

\vspace{-2mm}
\subsection{3-Way Trade-off}
We address {\bf Question-2} by evaluating the DP protocols with different $\epsilon$ ranging from $0.01$ to $50$. 
\begin{figure}[ht]
\captionsetup[sub]{font=small,labelfont={bf,sf}}
    \begin{subfigure}[b]{0.48\linewidth}
    \centering    \includegraphics[width=1\linewidth]{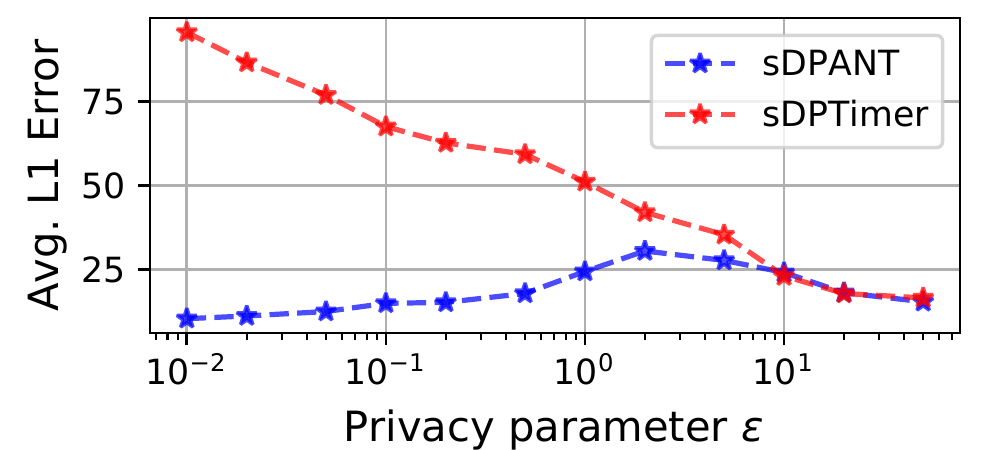}
    
        \caption{Privacy vs. Accuracy (TPC-ds)}
        \label{fig:tpc-eps-acc}\end{subfigure}
      \begin{subfigure}[b]{0.48\linewidth}
    \centering    \includegraphics[width=1\linewidth]{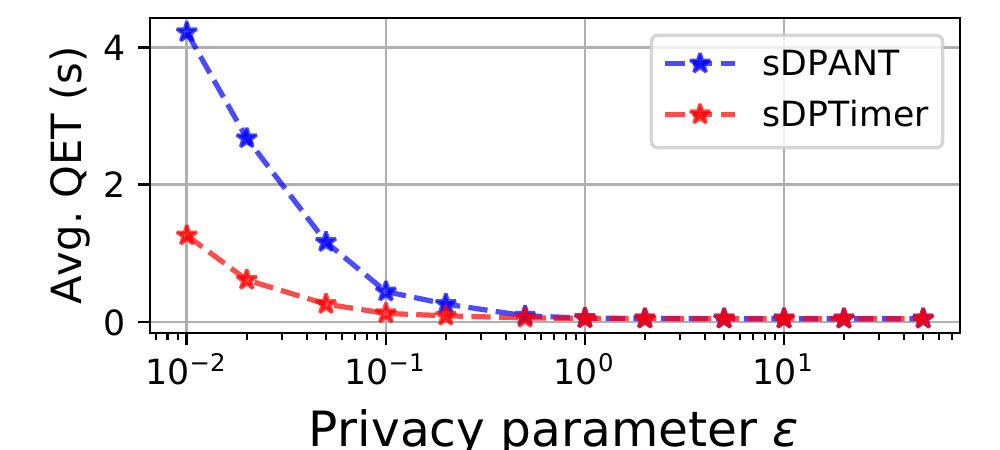}  
        \caption{Privacy vs. Efficiency (TPC-ds)}
        \label{fig:tpc-eps-perf}
    \end{subfigure}
    \newline
    \begin{subfigure}[b]{0.48\linewidth}
    \centering    \includegraphics[width=1\linewidth]{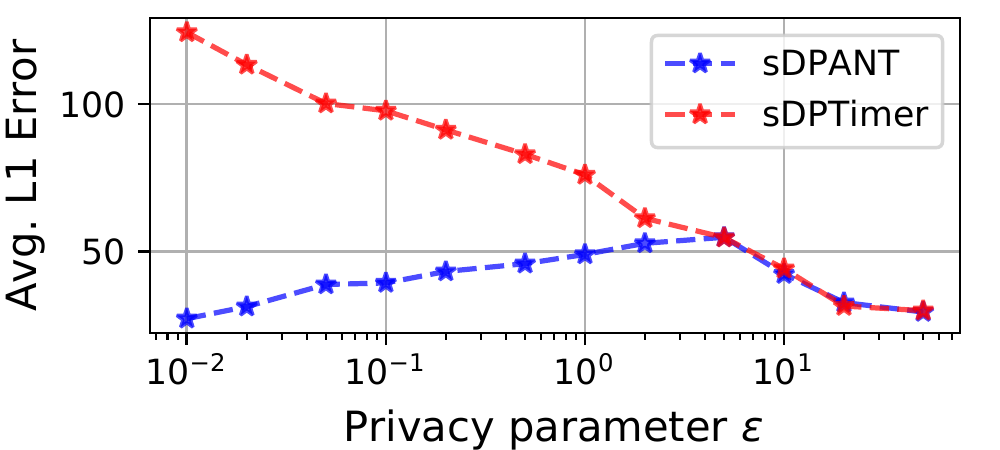}
        \caption{Privacy vs. Accuracy (CPDB)}
        \label{fig:cpdb-eps-acc}\end{subfigure}
      \begin{subfigure}[b]{0.48\linewidth}
    \centering    \includegraphics[width=1\linewidth]{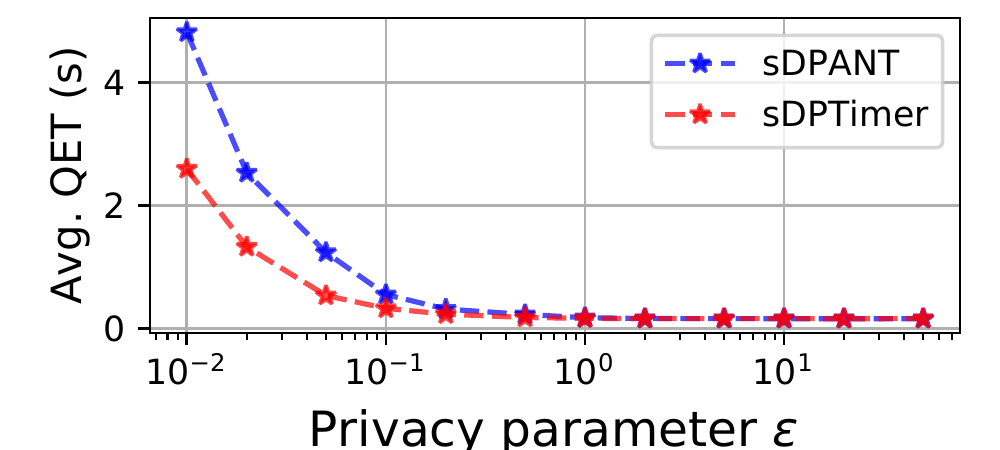}  
        \caption{Privacy vs. Efficiency (CPDB)}
        \label{fig:cpdb-eps-perf}
    \end{subfigure}
    \vspace{-3mm}
   \caption{Trade-off experiment.}
   \label{fig:trade-off}
\end{figure}

\vspace{-1mm}
\boldparagraph{Observation 3. $\timer$ and $\ant$ exhibit different privacy-accuracy trade-off.} The accuracy-privacy trade-off evaluation is summarized in Figure~\ref{fig:tpc-eps-acc} and~\ref{fig:cpdb-eps-acc}. In general, as $\epsilon$ increases from 0.01 to 50, we observe a consistent decreasing trend in the average L1 error for $\timer$, while the mean L1 error for $\ant$ first increases and then decreases. According to our previous discussion, the error upper bound of $\timer$ is given by \smash{$c^* + O(\frac{2b\sqrt{k}}{\epsilon})$}, where $c^*$ denotes the cached new entries since last update, and \smash{$O(\frac{2b\sqrt{k}}{\epsilon})$} is the upper bound for the deferred data (Theorem~\ref{lg:timer}). As $\timer$ has a fixed update frequency, thus $c^*$ is independent of $\epsilon$. However, the amount of deferred data is bounded by \smash{$O(\frac{2b\sqrt{k}}{\epsilon})$}, which leads to a decreasing trend in the L1 error of $\timer$ as $\epsilon$ increases. On the other hand, the update frequency of $\ant$ is variable and will be affected accordingly when $\epsilon$ changes. For example, a relatively small $\epsilon$  (large noise) will result in more frequent updates. As large noises can cause $\ant$ to trigger an update early before enough data has been placed in the secure cache. As a result, a relatively small $\epsilon$ will lead to a correspondingly small $c^*$, which essentially produces smaller query errors. Additionally, when $\epsilon$ reaches a relatively large level, its effect on $\ant$'s update frequency becomes less significant. Increasing $\epsilon$ does not affect $c^*$ much, but causes a decrease in the amount of deferred data (bouned by \smash{$O(\frac{16\log{t}}{\epsilon})$} as shown in  Theorem~\ref{one:ant}). This explains why there is a decreasing trend of $\ant$'s L1 error after $\epsilon$ reaches a relatively large level. Nevertheless, both protocols show a privacy-accuracy trade-off, meaning that users can actually adjust privacy parameters to achieve their desired accuracy goals.

\vspace{-1mm}
\boldparagraph{Observation 4. DP protocols have similar privacy-efficiency trade-off.} Both DP protocols show similar trends in terms of efficiency metrics (Figure~\ref{fig:tpc-eps-perf} and ~\ref{fig:cpdb-eps-perf}), that is when $\epsilon$ increases, the QET decreases. It is because with a relatively large $\epsilon$, the number of dummy data included in the view will be reduced, thus resulting in a subsequent improvement in query efficiency. Thus, similar to the accuracy-privacy trade-off, the DP protocols also provide a privacy-efficiency trade-off that allow users to tune the privacy parameter $\epsilon$ in order to obtain their desired performance goals.

\vspace{-3mm}
\subsection{Comparison Between DP Protocols}
We address {\bf Question-3} by comparing the two DP protocols over different type of workloads. In addition to the standard one, for each dataset, we create two additional datasets. 
First, we sample data from the original data and create a {\it Sparse} one, where the total number of view entries is 10\% of the standard one. Second, we process {\it Burst} data by adding data points to the original dataset, where the resulting data has 2$\times$ more view entries. 
\begin{figure}[ht]
\captionsetup[sub]{font=small,labelfont={bf,sf}}
    \begin{subfigure}[b]{0.48\linewidth}
    \centering    \includegraphics[width=1\linewidth]{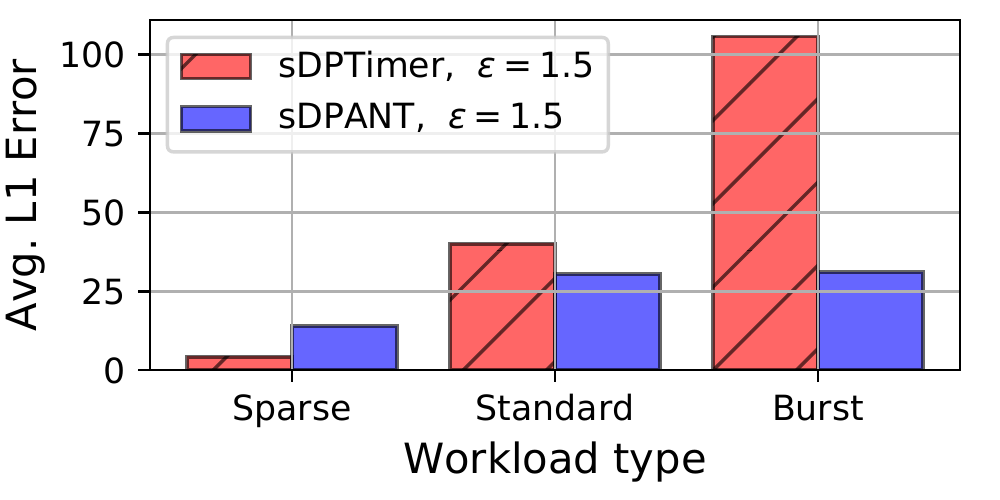}
        \caption{Workload vs. Accuracy (TPC-ds)}
        \label{fig:tpc-ld-acc}\end{subfigure}
      \begin{subfigure}[b]{0.48\linewidth}
    \centering    \includegraphics[width=1\linewidth]{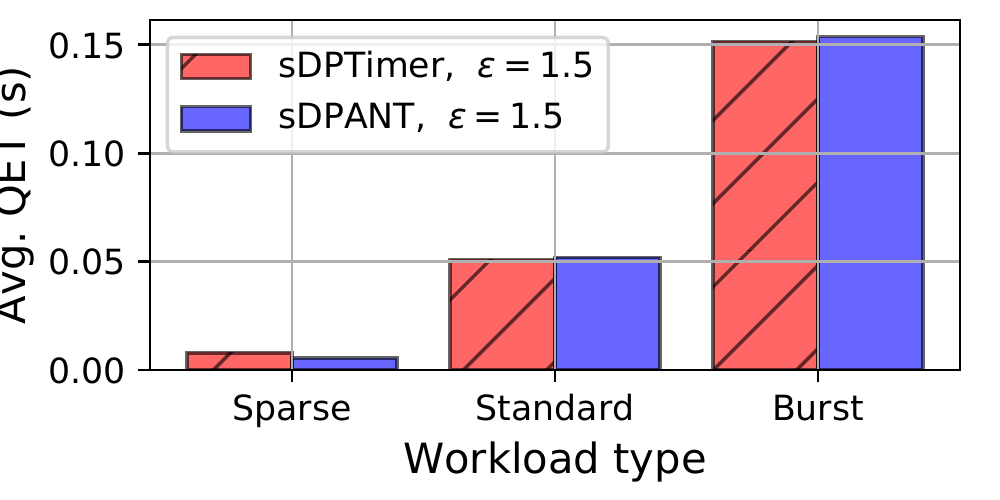}  
        \caption{Workload vs. Efficiency (TPC-ds)}
        \label{fig:tpc-ld-perf}
    \end{subfigure}
    \newline
    \begin{subfigure}[b]{0.48\linewidth}
    \centering    \includegraphics[width=1\linewidth]{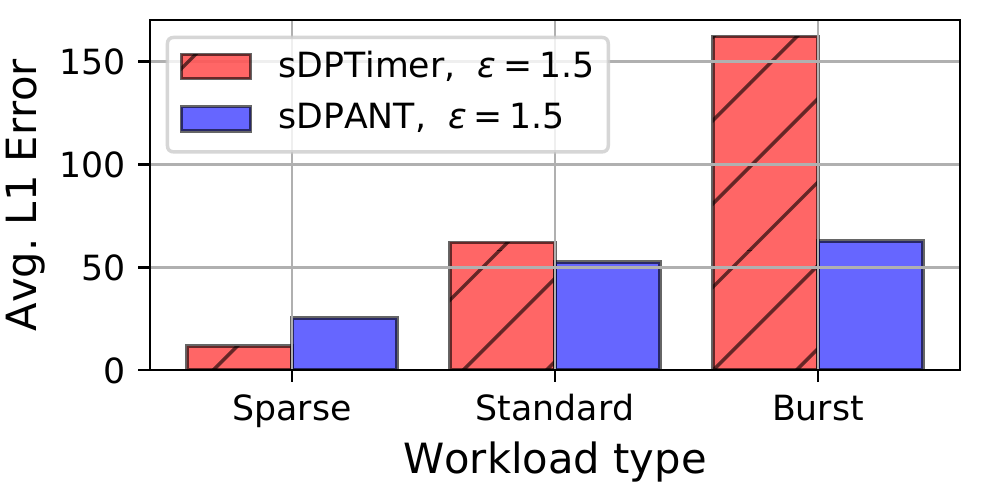}
        \caption{Workload vs. Accuracy (CPDB)}
        \label{fig:cmp-acc}\end{subfigure}
      \begin{subfigure}[b]{0.48\linewidth}
    \centering    \includegraphics[width=1\linewidth]{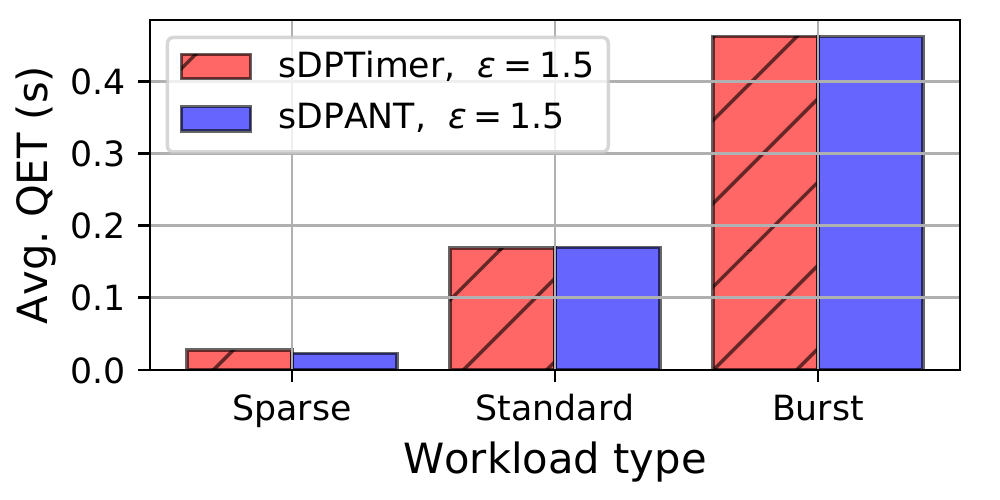}  
        \caption{Workload vs. Efficiency (CPDB)}
        \label{fig:cmp-perf}
    \end{subfigure}
    \vspace{-3mm}
   \caption{DP protocols under different workloads.}
   \label{fig:ld}
\end{figure}

\vspace{-1mm}
\boldparagraph{Observation 5. $\timer$ and $\ant$ show accuracy advantages in processing {\it Sparse} and {\it Burst} data, respectively.} According to Figure~\ref{fig:ld}, $\timer$ shows a relatively lower L1 error in the {\it Sparse} group than $\ant$. It is because it can take a relatively long time to have a new view entry when processing {\it Sparse} data. Applying $\ant$ will cause some data to be left in the secure cache for a relatively long time. However, $\timer$'s update schedule is independent of the data workload type, so when the load becomes very sparse, the data will still be synchronized on time. This explains why $\timer$ shows a better accuracy guarantee against $\ant$ for sparse data. On the contrary, when the data becomes very dense, i.e., there is a burst workload, the fixed update rate of $\timer$ causes a large amount of data to be stagnant in the secure cache. And thus causes significant degradation of the accuracy guarantee. However, $\ant$ can adjust the update frequency according to the data type, i.e., the denser the data, the faster the update. This feature gives $\ant$ an accuracy edge over $\timer$ when dealing with burst workloads. On the other hand, both methods show similar efficiency for all types of test datasets.
\begin{figure}[ht]
\captionsetup[sub]{font=small,labelfont={bf,sf}}
    \begin{subfigure}[b]{0.3\linewidth}
    \centering    \includegraphics[width=1\linewidth]{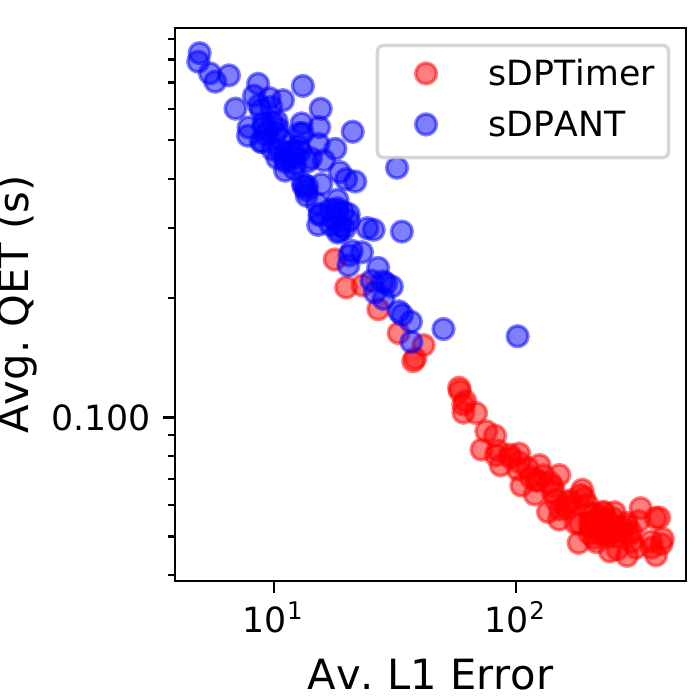}
        \caption{TPC-ds, $~~~\epsilon=0.1$}
        \label{fig:tpc-0.1-cmp}\end{subfigure}
      \begin{subfigure}[b]{0.3\linewidth}
    \centering    \includegraphics[width=1\linewidth]{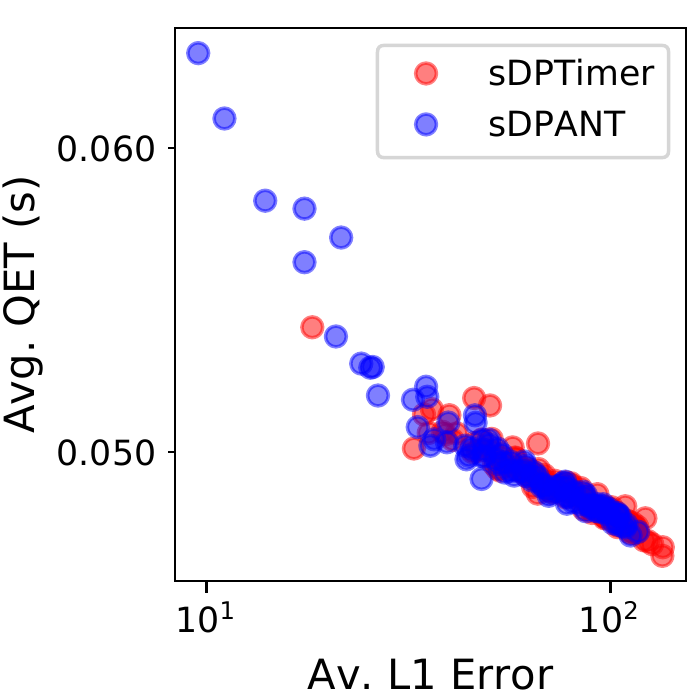}
        \caption{TPC-ds, $~~~\epsilon=1$ }
        \label{fig:tpc-1-cmp}
    \end{subfigure}
    \begin{subfigure}[b]{0.3\linewidth}
    \centering    \includegraphics[width=1\linewidth]{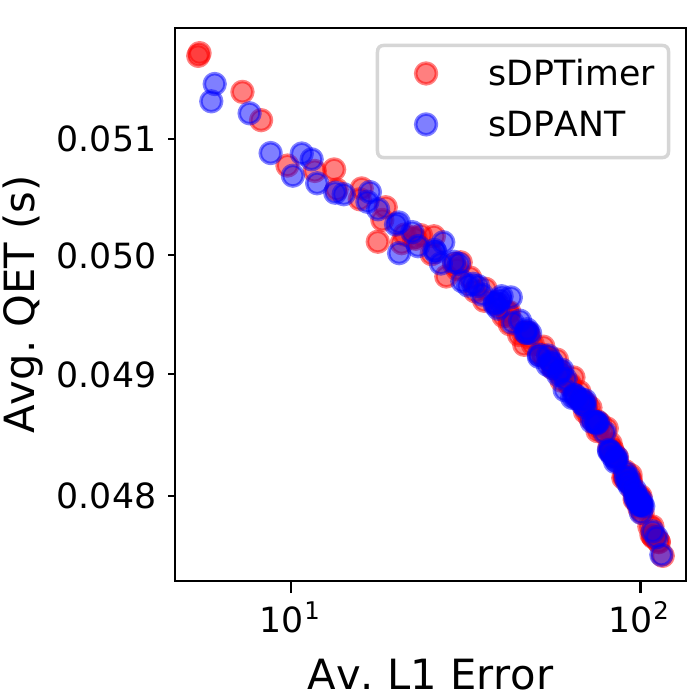}  
        \caption{TPC-ds, $~~~\epsilon=10$ }
        \label{fig:tpc-10-cmp}
    \end{subfigure}
    \newline
    \begin{subfigure}[b]{0.3\linewidth}
    \centering    \includegraphics[width=1\linewidth]{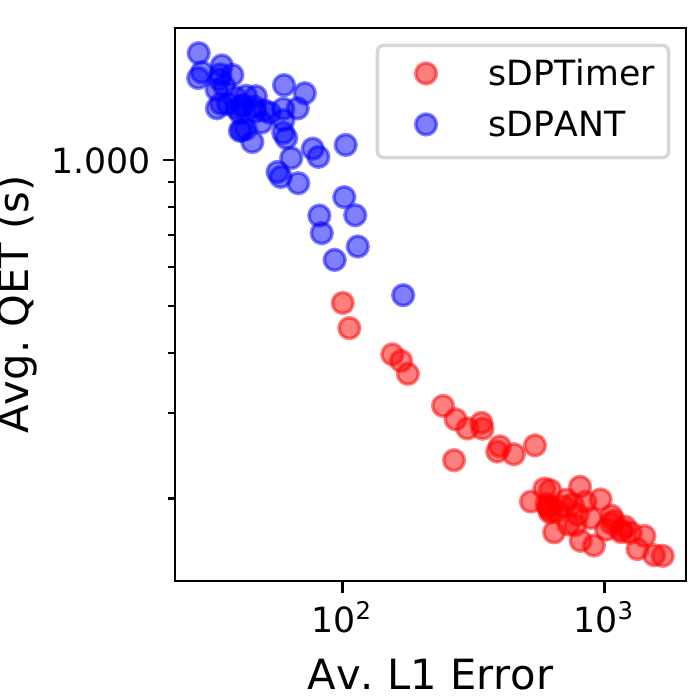}
        \caption{CPDB, $\epsilon=0.1$}
        \label{fig:cpdb-0.1-cmp}\end{subfigure}
      \begin{subfigure}[b]{0.3\linewidth}
    \centering    \includegraphics[width=1\linewidth]{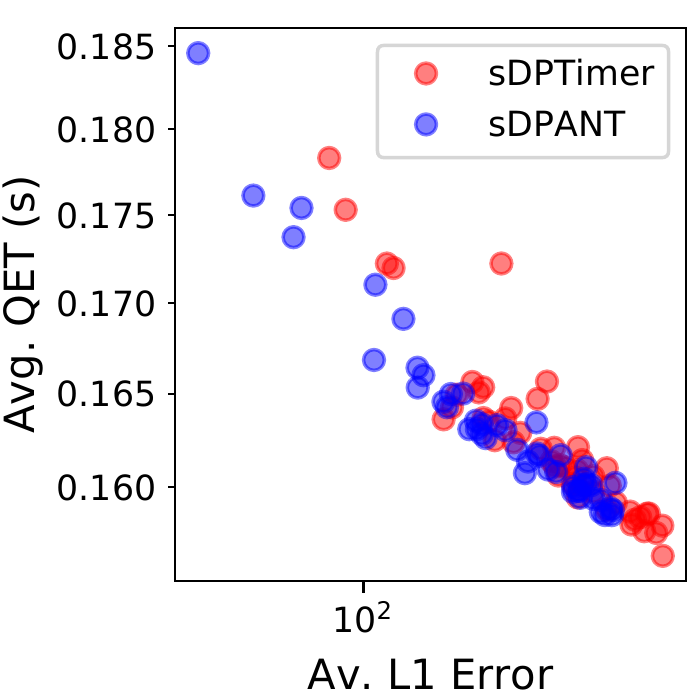}
        \caption{CPDB, $\epsilon=1$}
        \label{fig:cpdb-1-cmp}
    \end{subfigure}
    \begin{subfigure}[b]{0.3\linewidth}
    \centering    \includegraphics[width=1\linewidth]{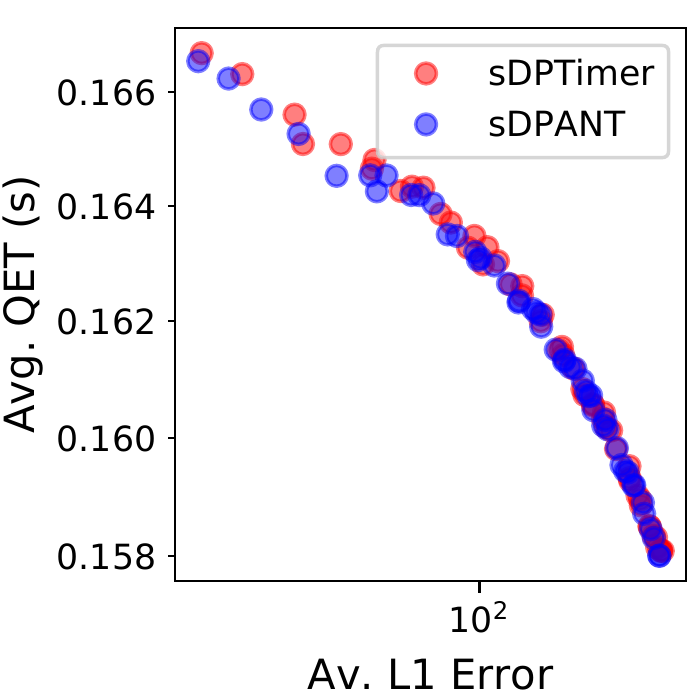}
        \caption{CPDB, $\epsilon=10$}
        \label{fig:cpdb-10-cmp}
    \end{subfigure}
    \vspace{-2mm}
   \caption{DP approaches under different workload.}
   \label{fig:dpcmp}
\end{figure}

Additionally, we also compare the two protocols with varying non-privacy parameters, i.e. $T$ and $\theta$, where we fix the $\epsilon$ then change $T$ from 1-100, and correspondingly set $\theta$ according to $T$ (As mentioned before, the average new view entries per moment are 2.7 and 9.8 for TPC-ds and CPDB data, respectively, thus we set $\theta$ to $3T$ and $10T$). We test the protocols with three privacy levels $\epsilon=0.1, 1$ and $10$ and report their comparison results in Figure~\ref{fig:dpcmp}.

\vspace{-1mm}
\boldparagraph{Observation 6. When $\epsilon$ is small, two DP protocols have different biases in terms of accuracy and performance.} According to Figure~\ref{fig:tpc-0.1-cmp} and~\ref{fig:cpdb-0.1-cmp}, when $\epsilon=0.1$, the data points for the $\ant$ locate in the upper left corner of both figures, while the $\timer$ results fall on the opposite side, in the lower right corner. This implies that when $\epsilon$ is relatively small (privacy level is high), $\ant$ tends to favor accuracy guarantees more, but at the expense of a certain level of efficiency. On the contrary, $\timer$ biases the efficiency guarantee. As per this observation, if users have strong demands regarding privacy and accuracy, then they should adopt $\ant$. However, if they have restrictive requirements for both privacy and performance, then $\timer$ is a better option. Moreover, the aforementioned deviations decrease when $\epsilon$ increases (Figure~\ref{fig:tpc-1-cmp}. In addition, when $\epsilon$ reaches a relatively large value, i.e $\epsilon=10$, both DP protocols essentially offer the same level of accuracy and efficiency guarantees. For example, for each "red" point in Figures~\ref{fig:tpc-10-cmp} and~\ref{fig:cpdb-10-cmp}, one can always find a comparative "blue" dot.

\vspace{-2mm}
\re{\subsection{Evaluation with Different $\omega$}\label{sec:exp-omega}
In this section, we investigate the effect of truncation bounds by evaluating \system under different $\omega$ values. Since the multiplicity of Q1 is 1, the $\omega$ for answering Q1 is fixed to 1. Hence, in this evaluation, we focus on Q2 over the CPDB data. We pick different $\omega$ values from the range of 2 to 32 and set the contribution budget as $b=2\omega$. The result is reported in Figure ~\ref{fig:trunc}.

\begin{figure}[ht]
\captionsetup[sub]{font=small,labelfont={bf,sf}}
    \begin{subfigure}[b]{0.48\linewidth}
    \centering    \includegraphics[width=1\linewidth]{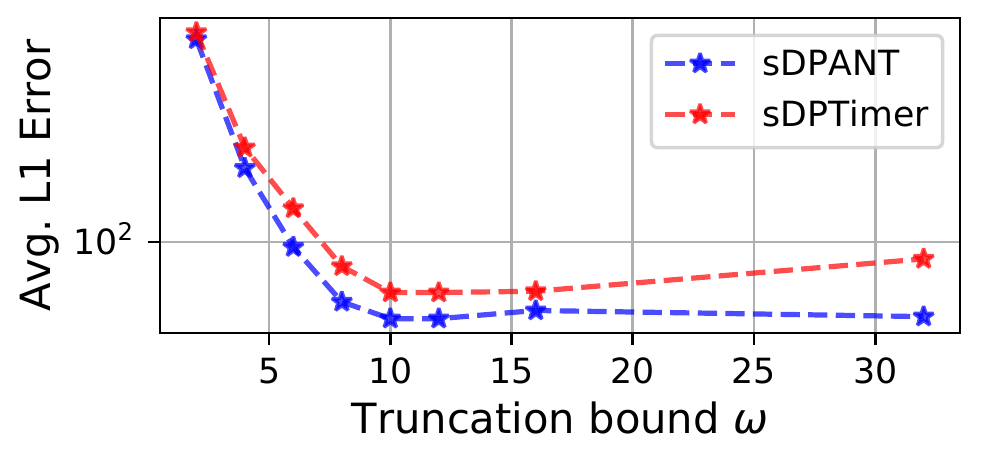}
        \caption{Query accuracy vs. $\omega$ }
        \label{fig:trunc-acc}\end{subfigure}
      \begin{subfigure}[b]{0.48\linewidth}
    \centering    \includegraphics[width=1\linewidth]{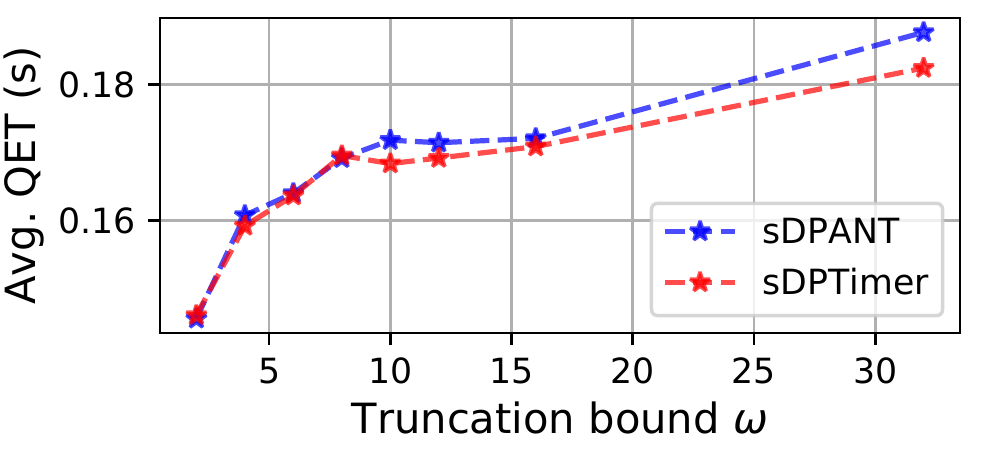}  
        \caption{Query efficiency vs. $\omega$}
        \label{fig:trunc-perf}
    \end{subfigure}
    \newline
    \begin{subfigure}[b]{0.48\linewidth}
    \centering    \includegraphics[width=1\linewidth]{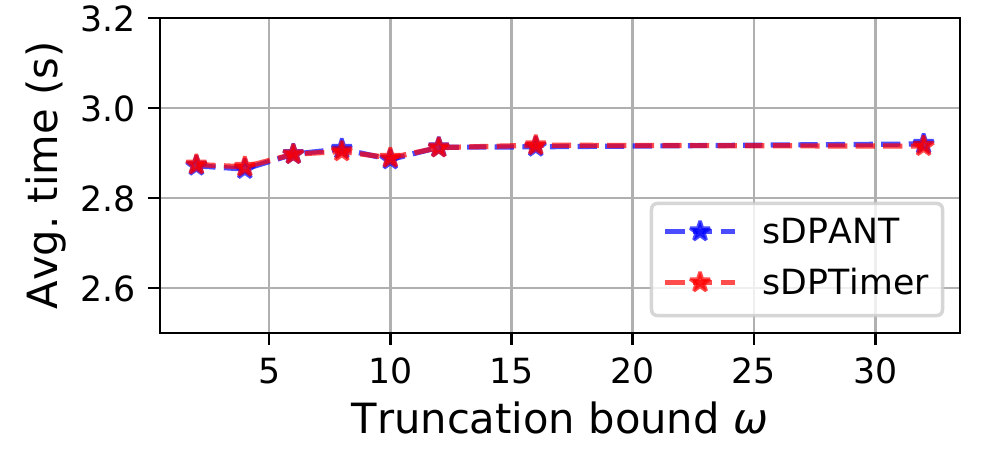}
        \caption{Avg. $\trans$ execution time }
        \label{fig:trunc-trans}\end{subfigure}
      \begin{subfigure}[b]{0.48\linewidth}
    \centering    \includegraphics[width=1\linewidth]{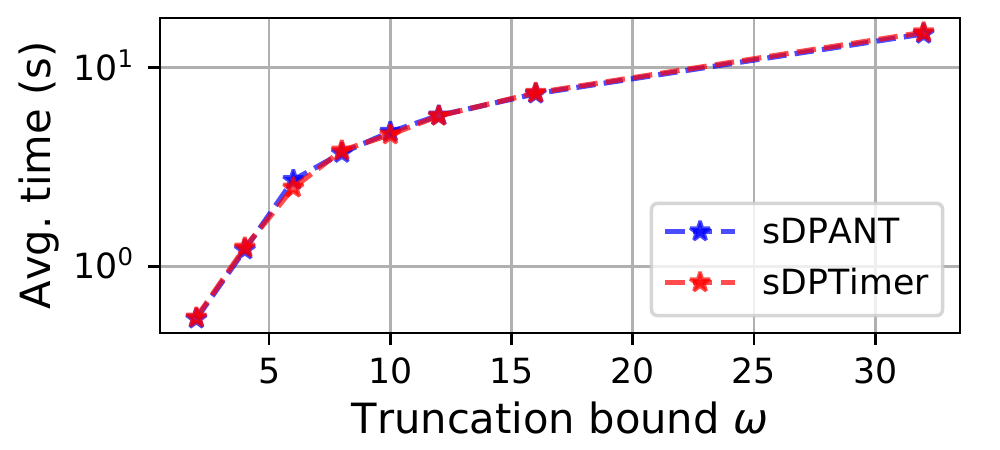}  
        \caption{Avg. $\sync$ execution time}
        \label{fig:trunc-sync}
    \end{subfigure}
    \vspace{-3mm}
   \caption{Evaluations with different truncation bound $\omega$.}
   \label{fig:trunc}
\end{figure}

\vspace{-2mm}
\boldparagraph{Observation 7. As $\omega$ grows from a small value, query accuracy increases and query efficiency decreases quickly. After $\omega$ reaches a relatively large value, $\timer$ and $\ant$ exhibit different trends in accuracy, but the same tendency in efficiency.} As per Figure~\ref{fig:trunc}, the average L1 error decreases when $\omega$ grows from $\omega=2$. This is because when $\omega$ is small, many true view entries are dropped by the $\trans$ protocol due to truncation constraint, which leads to larger L1 query errors. However, when $\omega$ reaches a relatively large value, i.e., greater than the maximum record contribution, then no real entries are discarded. At this point, increasing $\omega$ only leads to the growth of injected DP noises. As we have analyzed before, the accuracy under $\ant$ can be better for relatively large noise, but the accuracy metric will be worse under $\timer$ method. On the other hand, dropping a large number of real entries (when $\omega$ is small) leads to a smaller materialized view, which consequently improves query efficiency. When $\omega$ is greater than the maximum record contribution, based on our analysis in Observation 4, keep increasing $\omega$ leads to both methods to introduce more dummy data to the view and causes its size to keep growing. As such, the efficiency continues decreasing.

\vspace{-1mm}
\boldparagraph{Observation 8. The average $\sync$ execution time increases along with the growth of $\omega$, while the average execution time of $\trans$ tends to be approximately the same.} The reason for this tendency is fairly straightforward. The execution time of both $\trans$ and $\sync$ protocols is dominated by the oblivious input sorting. The input size of the $\trans$ protocol is only related to the size of data batches submitted by the users. Therefore, changing $\omega$ does not affect the efficiency of $\trans$ execution. However, the input size of $\sync$ is tied to $\omega$, so as $\omega$ grows, the execution time of $\sync$ increases.}

\vspace{-2mm}
\subsection{Scaling Experiments}
We continue to evaluate our framework with scaling experiments. To generate data with different scales, we randomly sample or replicate the original TPC-ds and CPDB data (We assign new primary key values to the replicated rows to prevent conflicts). According to Figure~\ref{fig:scale}, for the largest dataset, i.e., the $4\times$ groups, the total MPC time are around 24 and 6 hours, respectively for TPC-ds and CPDB. However, it is worth mentioning that for the $4\times$ group, TPC-ds has 8.8 million and 1.08 million records in the two testing tables, and CPDB has 800K and 2.6 million records for $Allegation$ and $Award$ tables, respectively.  This shows the practical scalability of our framework. In addition, the total query time for 4$\times$ TPC-ds and 4$\times$ CPDB groups are within 400 and 630 seconds, respectively.
\begin{figure}[ht]
\captionsetup[sub]{font=small,labelfont={bf,sf}}
    \begin{subfigure}[b]{0.48\linewidth}
    \centering    \includegraphics[width=1\linewidth]{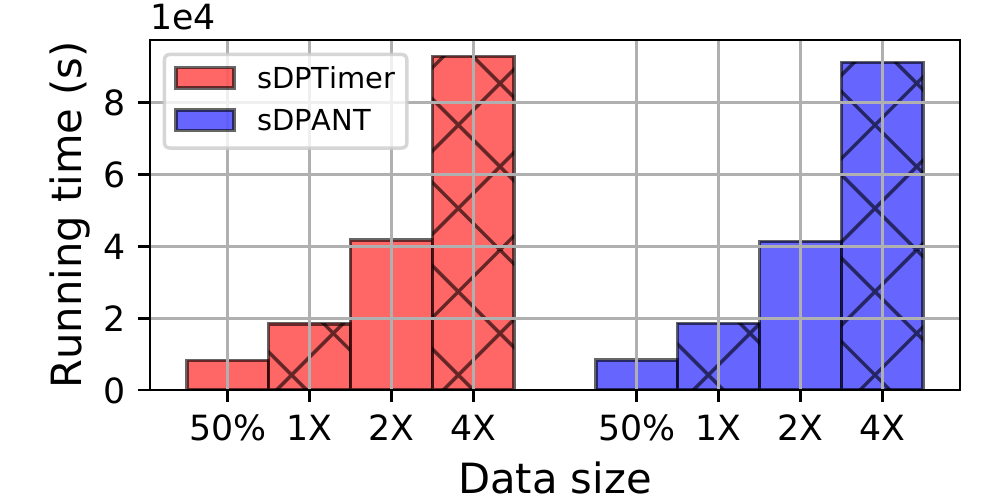}
        \caption{Total MPC time (TPC-ds)}
        \label{fig:tpc-ld-acc}\end{subfigure}
      \begin{subfigure}[b]{0.48\linewidth}
    \centering    \includegraphics[width=1\linewidth]{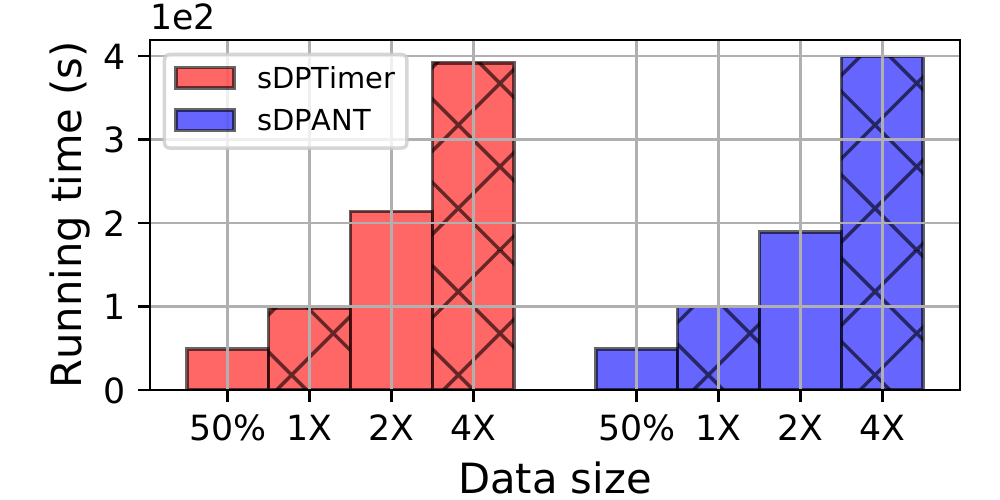}  
        \caption{Total query time (TPC-ds)}
        \label{fig:tpc-ld-perf}
    \end{subfigure}
    \newline
    \begin{subfigure}[b]{0.48\linewidth}
    \centering    \includegraphics[width=1\linewidth]{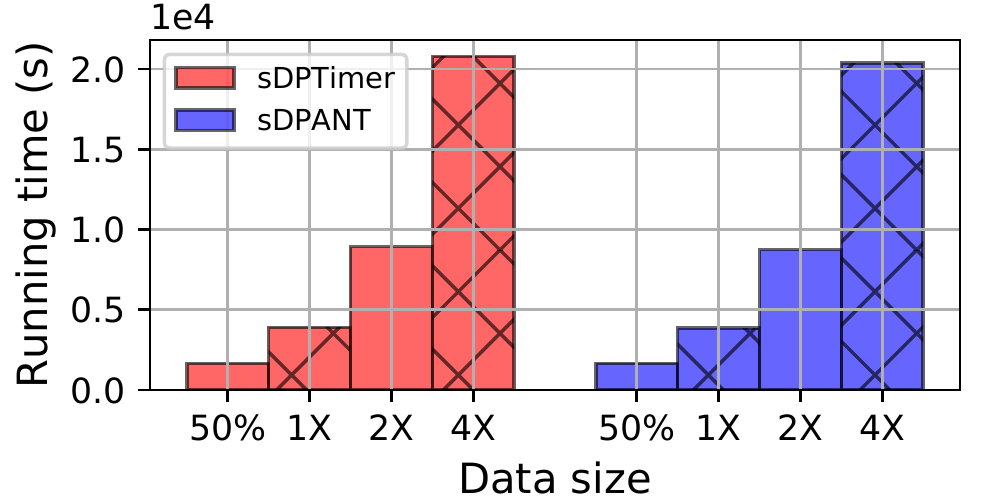}
        \caption{Total MPC time (CPDB)}
        \label{fig:tpc-ld-acc}\end{subfigure}
      \begin{subfigure}[b]{0.48\linewidth}
    \centering    \includegraphics[width=1\linewidth]{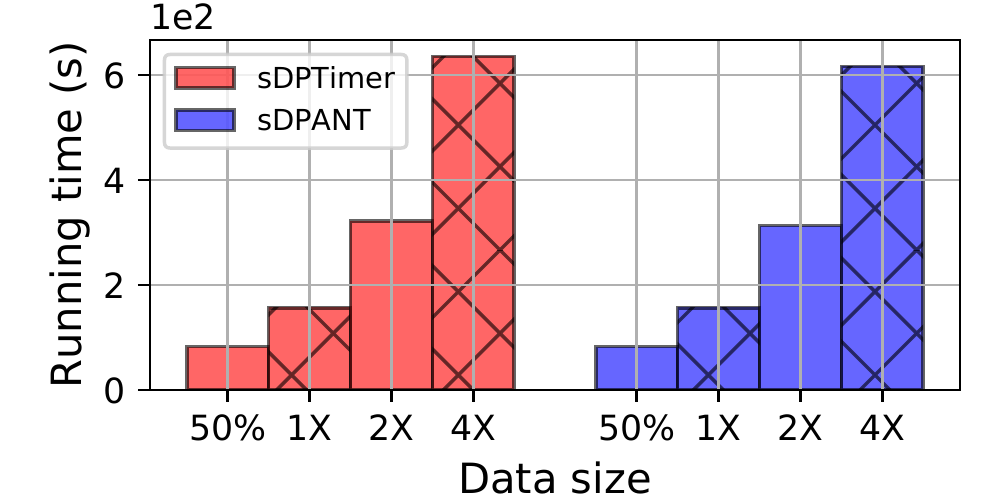}  
        \caption{Total query time (CPDB)}
        \label{fig:tpc-ld-perf}
    \end{subfigure}
    \vspace{-3mm}
\caption{Scaling experiments}
\label{fig:scale}
\end{figure}
\vspace{-3mm}
\section{Extensions}\label{sec:ext}
\label{sec:discussion}
We discuss potential extensions of the original \system design.  

\vspace{-1.5mm}
\boldparagraph{Connecting with DP-Sync.}
For ease of demonstration, in the prototype design, we assume that data owners submit a fixed amount of data at fixed intervals. However, \system is not subject to this particular record synchronization strategy. Owners can choose other private update policies such as the ones proposed in DP-Sync, and can also adapt our framework. Additionally, the view update protocol requires no changes or recompilation as long as the view definition does not change. On the other hand, privacy will still be 
ensured under the composed system that connects \system with DP-Sync. For example, assume the owner adopts a record synchronization strategy that ensures $\epsilon_1$-DP and the server is deployed with \system that guarantees $\epsilon_2$-DP with respect to the owner's data. By sequential composition theorem~\cite{dwork2014algorithmic}, revealing their combined leakage ensures $(\epsilon_1 + \epsilon_2)$-DP over the owner's data. Similarly, such composability can also be obtained in terms of the accuracy guarantee. For instance, let's denote the error bound for the selected record synchronization policy as $\alpha_r$ (total number of records not uploaded in time). Then by Theorem~\ref{lg:timer} and~\ref{one:ant}, the combined system ensures error bounds \smash{$O(b\alpha_r+\frac{2b}{\epsilon}\sqrt{k})$} and \smash{$O(b\alpha_r+\frac{16b\log{t}}{\epsilon})$} under $\timer$ and $\ant$ protocol, respectively. Interested readers may refer to our full version for the complete utility proofs.

\vspace{-1.5mm}
\boldparagraph{Support for complex query workloads.} Now we describe how to generalize the view update protocol for complex query workloads, i.e. queries that can be written as a composite of multiple relational algebra operators. Apparently, one can directly replicate the design of this paper to support complex queries by first compiling a $\trans$ protocol that produces and caches the corresponding view tuples based on the specified query plan, while a $\sync$ protocol is used independently to continuously synchronize the cached data. However, there exists another design pattern that utilizes multi-level ``Transform-and-Shrink'' protocol. For example, we can disassemble a query into a series of operators and then construct an independent "Transform-and-Shrink" protocol for each individual operator. Moreover, the output of one "Transform-and-Shrink" protocol can be the input of another one, which eventually forms a multi-level view update protocol. There are certain benefits of the multi-level design, for instance, one can optimize the system efficiency via operator level privacy allocation~\cite{bater2018shrinkwrap}. Recall that in Section~\ref{sec:dp-st} we discussed that the choice of privacy budget affects the number of dummy records processed by the system, with a higher proportion of dummy records reducing overall performance and vice versa. To maximize performance, one can construct an optimization problem that maximizes the efficiency of all operators in a given query, while maintaining the desired accuracy level. With a privacy budget allocation determined by the optimization problem, each operator can carry out its own instance of \system, minimizing the overall computation cost while satisfying desired privacy and accuracy constraints. Note that optimization details are beyond the scope of this paper but may be of independent interest and we leave the design of these techniques to future work.

\vspace{-1.5mm}
\re{\boldparagraph{Expanding to multiple servers.} Although in our prototype \system design we assume to leverage 2 non-colluding servers, the system architecture can be modified to work with multiple servers. In what follows, we summarize the major modifications that would extend our current design to a $N$ servers setup such that $N \geq 2$.
 Firstly, the owners need to share their local data using the $(N,N)$-secret-sharing scheme, and disseminate one share per participating server. In addition, for all outsourced objects, such as the secure cache, the materialized view, and parameters passed between view update protocols, must be stored on the $N$ servers in a secret shared manner. Secondly, both $\trans$ and $\sync$ protocol will be compiled as a general MPC protocol  where $N$ parties (servers) provide their confidential input and evaluate the protocol altogether.  Finally, when generating DP noises, each server needs to contribute a random bit string to the MPC protocol, which subsequently aggregates the $N$ random strings to obtain the randomness used for noise~generation. Note that our joint noise addition mechanism ensures to produce only one instance of DP noise, thus expanding to $N$ servers setting does not lead to injecting more noise. According to~\cite{keller2018overdrive, keller2013architecture}, such design can tolerate up to $N-1$ server corruptions.
}
\vspace{-2mm}

\vspace{-2mm}
\section{Related Work}\label{sec:related}
\vspace{-3mm}
\boldparagraph{Secure outsourced database and leakage abuse attacks.} There have been a series of efforts under the literature of secure outsourcing databases. Existing solutions utilize bucketization~\cite{hacigumucs2002executing, hore2004privacy, hore2012secure}, predicate encryption~\cite{shi2007multi,lu2012privacy},  property and order preserving encryption~\cite{agrawal2004order,bellare2007deterministic, boldyreva2009order, pandey2012property, boldyreva2011order, popa2012cryptdb, poddar2016arx}, symmetric searchable encryption (SSE)~\cite{curtmola2011searchable, stefanov2014practical, cash2014dynamic, kamara2012dynamic, kellaris2017accessing, kamara2018sql, kamara2019computationally, patel2019mitigating, ghareh2018new, amjad2019forward, kellaris2021accessing},  functional encryption~\cite{boneh2004public, shen2009predicate}, oblivious RAM~\cite{bater2017smcql, crooks2018obladi, demertzis2020seal, naveed2014dynamic, ishai2016private, zheng2017opaque},  multi-party secure computation (MPC)~\cite{bater2017smcql, bater2018shrinkwrap, bonawitz2017practical, tan2021cryptgpu}, trusted execution environments (TEE)~\cite{priebe2018enclavedb, eskandarian2017oblidb, vinayagamurthy2019stealthdb, xu2019hermetic} and homomorphic encryption~\cite{gentry2009fully, boneh2005evaluating, chowdhury2019crypt, samanthula2014privacy}. These designs differ in the types of supported queries and the provided security guarantees. Although the initial goal was to conceal the record values~\cite{hacigumucs2002executing, hore2004privacy, boneh2004public, shi2007multi, lu2012privacy, popa2012cryptdb, agrawal2004order,bellare2007deterministic, bater2017smcql, boldyreva2009order, pandey2012property, boldyreva2011order, popa2012cryptdb}, researchers soon discovered the shortcomings of this security assurance. Recent work has revealed that these methods may be subject to certain leakage through query patterns~\cite{zhang2016all, wang2018order}, access patterns~\cite{kellaris2016generic, dautrich2013compromising} and query response volume~\cite{kellaris2016generic, grubbs2019learning, grubbs2018pump, gui2019encrypted}, which makes them vulnerable to leakage-abuse attacks~\cite{cash2015leakage, blackstone2019revisiting}. Therefore, more recent works on secure outsourced databases not only consider concealing record values but also hiding associated leakages~\cite{stefanov2014practical, eskandarian2017oblidb, kellaris2017accessing, bater2018shrinkwrap, cash2014dynamic, kamara2018sql, kamara2019computationally, patel2019mitigating, ghareh2018new, amjad2019forward, demertzis2020seal, kellaris2021accessing, naveed2014dynamic, bater2017smcql, crooks2018obladi, ishai2016private, xu2019hermetic, zheng2017opaque}. Unfortunately, few of the aforementioned efforts consider the potential leakage when underlying data is dynamic~\cite{cash2014dynamic, kellaris2017accessing, ghareh2018new, amjad2019forward}. Wang et al.~\cite{wang2021dp} formalize a general leakage named {\it update pattern} that may affect many existing secure database schemes when outsourcing dynamic data. 

\vspace{-1mm}
\boldparagraph{Differentially-private leakage.} Existing studies on hiding database leakage with DP can be divided into two main categories: (i) safeguarding the query results from revealing sensitive information~\cite{cummings2018differential, chowdhury2019crypt, agarwal2019encrypted, lecuyer2019sage, luo2021privacy}, and (ii) obscuring side-channel leakages such as access pattern~\cite{bater2018shrinkwrap, mazloom2018secure, chen2018differentially, wagh2018differentially, shang2021obfuscated, kellaris2017accessing}, query volume~\cite{kellaris2021accessing, patel2019mitigating} and update patterns~\cite{wang2021dp}. The first category consists of works that enable DP query answering over securely provisioned (and potentially dynamic) data. Since these efforts typically focus solely on query outputs, side-channel leakages are not considered or assumed to be eliminable by existing techniques. Works in the second group focus on hiding side-channel information with DP, which is pertinent to our study. \re{Among those,~\cite{bater2018shrinkwrap} and~\cite{wang2021dp} are the two most relevant works to our study. ~\cite{bater2018shrinkwrap} extends the work of~\cite{bater2017smcql}, both of which use MPC as the main tool to architect secure outsourced databases. However,~\cite{bater2017smcql} fails to address some important leakages associated with intermediate computation results (i.e., the size of some intermediate outputs may leak sensitive information about the underlying data). Thus, ~\cite{bater2018shrinkwrap} is proposed to fill this gap. \cite{bater2018shrinkwrap} implements a similar resizing technique as \system that ensures the volume leakage per secure operator is bounded by differential privacy, however, their system is restrictively focused on processing static data. 
~\cite{wang2021dp} considers hiding update patterns when outsourcing growing data with private update strategies. However, they mandate that the update strategies must be enforced by trusted entities, while \system allows untrusted servers to privately synchronize the materialized view. Additionally,~\cite{wang2021dp} considers the standard mode that processes queries directly over outsourced data, which inevitably incurs additional performance overhead.
Interested readers may refer to Sections~\ref{sec:pre} and ~\ref{sec:dp-st}, where we provide more in-depth comparisons between \system and~\cite{bater2018shrinkwrap, wang2021dp}, and highlight our technical contributions.}
\vspace{-1mm}
\boldparagraph{Bounding privacy loss.} There is a series of work investigating approaches to constrain the privacy loss of queries or transformations with unbounded stability~\cite{mcsherry2009privacy, johnson2018towards, kotsogiannis2019architecting, kotsogiannis2019privatesql, wilson2020differentially, tao2020computing}. However these works are conducted under the scope of standard databases rather than secure outsourced databases. Moreover, most of the works consider to bound the privacy loss of a single query or one-time transformation~\cite{mcsherry2009privacy, johnson2018towards, wilson2020differentially, tao2020computing}. In this work, we consider constraining the privacy loss of a composed transformation, which may contain an infinite number of sub-transformations. 
\section{Conclusion}
In this paper, we have presented a framework \system for outsourcing growing data onto untrusted servers while retaining the query functionalities over the outsourced data. \system not only supports an efficient view-based query answering paradigm but also ensures bounded leakage in the maintenance of materialized view. This is achieved by (i) utilizing incremental MPC and differential privacy to architect the secure view update protocol and (ii) imposing constraints on record contributions to the transformation of materialized view instance.


\bibliographystyle{ACM-Reference-Format}
\bibliography{sample}


\begin{thebibliography}{89}


\ifx \showCODEN    \undefined \def \showCODEN     #1{\unskip}     \fi
\ifx \showDOI      \undefined \def \showDOI       #1{#1}\fi
\ifx \showISBNx    \undefined \def \showISBNx     #1{\unskip}     \fi
\ifx \showISBNxiii \undefined \def \showISBNxiii  #1{\unskip}     \fi
\ifx \showISSN     \undefined \def \showISSN      #1{\unskip}     \fi
\ifx \showLCCN     \undefined \def \showLCCN      #1{\unskip}     \fi
\ifx \shownote     \undefined \def \shownote      #1{#1}          \fi
\ifx \showarticletitle \undefined \def \showarticletitle #1{#1}   \fi
\ifx \showURL      \undefined \def \showURL       {\relax}        \fi
\providecommand\bibfield[2]{#2}
\providecommand\bibinfo[2]{#2}
\providecommand\natexlab[1]{#1}
\providecommand\showeprint[2][]{arXiv:#2}

\bibitem[\protect\citeauthoryear{Agarwal, Herlihy, Kamara, and Moataz}{Agarwal
  et~al\mbox{.}}{2019}]%
        {agarwal2019encrypted}
\bibfield{author}{\bibinfo{person}{Archita Agarwal}, \bibinfo{person}{Maurice
  Herlihy}, \bibinfo{person}{Seny Kamara}, {and} \bibinfo{person}{Tarik
  Moataz}.} \bibinfo{year}{2019}\natexlab{}.
\newblock \showarticletitle{Encrypted Databases for Differential Privacy}.
\newblock \bibinfo{journal}{\emph{Proceedings on Privacy Enhancing
  Technologies}} \bibinfo{volume}{2019}, \bibinfo{number}{3}
  (\bibinfo{year}{2019}), \bibinfo{pages}{170--190}.
\newblock


\bibitem[\protect\citeauthoryear{Agrawal, Kiernan, Srikant, and Xu}{Agrawal
  et~al\mbox{.}}{2004}]%
        {agrawal2004order}
\bibfield{author}{\bibinfo{person}{Rakesh Agrawal}, \bibinfo{person}{Jerry
  Kiernan}, \bibinfo{person}{Ramakrishnan Srikant}, {and}
  \bibinfo{person}{Yirong Xu}.} \bibinfo{year}{2004}\natexlab{}.
\newblock \showarticletitle{Order preserving encryption for numeric data}. In
  \bibinfo{booktitle}{\emph{Proceedings of the 2004 ACM SIGMOD international
  conference on Management of data}}. \bibinfo{pages}{563--574}.
\newblock


\bibitem[\protect\citeauthoryear{Amjad, Kamara, and Moataz}{Amjad
  et~al\mbox{.}}{2019}]%
        {amjad2019forward}
\bibfield{author}{\bibinfo{person}{Ghous Amjad}, \bibinfo{person}{Seny Kamara},
  {and} \bibinfo{person}{Tarik Moataz}.} \bibinfo{year}{2019}\natexlab{}.
\newblock \showarticletitle{Forward and backward private searchable encryption
  with SGX}. In \bibinfo{booktitle}{\emph{Proceedings of the 12th European
  Workshop on Systems Security}}. \bibinfo{pages}{1--6}.
\newblock


\bibitem[\protect\citeauthoryear{Arasu, Blanas, Eguro, Kaushik, Kossmann,
  Ramamurthy, and Venkatesan}{Arasu et~al\mbox{.}}{2013}]%
        {arasu2013orthogonal}
\bibfield{author}{\bibinfo{person}{Arvind Arasu}, \bibinfo{person}{Spyros
  Blanas}, \bibinfo{person}{Ken Eguro}, \bibinfo{person}{Raghav Kaushik},
  \bibinfo{person}{Donald Kossmann}, \bibinfo{person}{Ravishankar Ramamurthy},
  {and} \bibinfo{person}{Ramarathnam Venkatesan}.}
  \bibinfo{year}{2013}\natexlab{}.
\newblock \showarticletitle{Orthogonal Security with Cipherbase.}. In
  \bibinfo{booktitle}{\emph{CIDR}}.
\newblock


\bibitem[\protect\citeauthoryear{Batcher}{Batcher}{1968}]%
        {batcher1968sorting}
\bibfield{author}{\bibinfo{person}{Kenneth~E Batcher}.}
  \bibinfo{year}{1968}\natexlab{}.
\newblock \showarticletitle{Sorting networks and their applications}. In
  \bibinfo{booktitle}{\emph{Proceedings of the April 30--May 2, 1968, spring
  joint computer conference}}. \bibinfo{pages}{307--314}.
\newblock


\bibitem[\protect\citeauthoryear{Bater, Goel, Elliott, Kho, Eggen, and
  Rogers}{Bater et~al\mbox{.}}{2016}]%
        {bater2017smcql}
\bibfield{author}{\bibinfo{person}{Johes Bater}, \bibinfo{person}{Satyender
  Goel}, \bibinfo{person}{Gregory Elliott}, \bibinfo{person}{Abel Kho},
  \bibinfo{person}{Craig Eggen}, {and} \bibinfo{person}{Jennie Rogers}.}
  \bibinfo{year}{2016}\natexlab{}.
\newblock \showarticletitle{{SMCQL: Secure querying for federated databases}}.
\newblock \bibinfo{journal}{\emph{Proceedings of the VLDB Endowment}}
  \bibinfo{volume}{10}, \bibinfo{number}{6} (\bibinfo{year}{2016}),
  \bibinfo{pages}{673--684}.
\newblock


\bibitem[\protect\citeauthoryear{Bater, He, Ehrich, Machanavajjhala, and
  Rogers}{Bater et~al\mbox{.}}{2018}]%
        {bater2018shrinkwrap}
\bibfield{author}{\bibinfo{person}{Johes Bater}, \bibinfo{person}{Xi He},
  \bibinfo{person}{William Ehrich}, \bibinfo{person}{Ashwin Machanavajjhala},
  {and} \bibinfo{person}{Jennie Rogers}.} \bibinfo{year}{2018}\natexlab{}.
\newblock \showarticletitle{Shrinkwrap: efficient sql query processing in
  differentially private data federations}.
\newblock \bibinfo{journal}{\emph{Proceedings of the VLDB Endowment}}
  \bibinfo{volume}{12}, \bibinfo{number}{3} (\bibinfo{year}{2018}),
  \bibinfo{pages}{307--320}.
\newblock


\bibitem[\protect\citeauthoryear{Beimel}{Beimel}{2011}]%
        {beimel2011secret}
\bibfield{author}{\bibinfo{person}{Amos Beimel}.}
  \bibinfo{year}{2011}\natexlab{}.
\newblock \showarticletitle{Secret-sharing schemes: a survey}. In
  \bibinfo{booktitle}{\emph{International conference on coding and
  cryptology}}. Springer, \bibinfo{pages}{11--46}.
\newblock


\bibitem[\protect\citeauthoryear{Bellare, Boldyreva, and O’Neill}{Bellare
  et~al\mbox{.}}{2007}]%
        {bellare2007deterministic}
\bibfield{author}{\bibinfo{person}{Mihir Bellare}, \bibinfo{person}{Alexandra
  Boldyreva}, {and} \bibinfo{person}{Adam O’Neill}.}
  \bibinfo{year}{2007}\natexlab{}.
\newblock \showarticletitle{Deterministic and efficiently searchable
  encryption}. In \bibinfo{booktitle}{\emph{Annual International Cryptology
  Conference}}. Springer, \bibinfo{pages}{535--552}.
\newblock


\bibitem[\protect\citeauthoryear{Blackstone, Kamara, and Moataz}{Blackstone
  et~al\mbox{.}}{2019}]%
        {blackstone2019revisiting}
\bibfield{author}{\bibinfo{person}{Laura Blackstone}, \bibinfo{person}{Seny
  Kamara}, {and} \bibinfo{person}{Tarik Moataz}.}
  \bibinfo{year}{2019}\natexlab{}.
\newblock \showarticletitle{Revisiting Leakage Abuse Attacks.}
\newblock \bibinfo{journal}{\emph{IACR Cryptol. ePrint Arch.}}
  \bibinfo{volume}{2019} (\bibinfo{year}{2019}), \bibinfo{pages}{1175}.
\newblock


\bibitem[\protect\citeauthoryear{Bogatov, Kellaris, Kollios, Nissim, and
  O'Neill}{Bogatov et~al\mbox{.}}{2021}]%
        {kellaris2021accessing}
\bibfield{author}{\bibinfo{person}{Dmytro Bogatov}, \bibinfo{person}{Georgios
  Kellaris}, \bibinfo{person}{George Kollios}, \bibinfo{person}{Kobbi Nissim},
  {and} \bibinfo{person}{Adam O'Neill}.} \bibinfo{year}{2021}\natexlab{}.
\newblock \showarticletitle{$\epsilon$psolute : Efficiently Querying Databases
  While Providing Differential Privacy}.
\newblock \bibinfo{journal}{\emph{arXiv preprint arXiv:1706.01552}}
  (\bibinfo{year}{2021}).
\newblock


\bibitem[\protect\citeauthoryear{Boldyreva, Chenette, Lee, and
  O’neill}{Boldyreva et~al\mbox{.}}{2009}]%
        {boldyreva2009order}
\bibfield{author}{\bibinfo{person}{Alexandra Boldyreva},
  \bibinfo{person}{Nathan Chenette}, \bibinfo{person}{Younho Lee}, {and}
  \bibinfo{person}{Adam O’neill}.} \bibinfo{year}{2009}\natexlab{}.
\newblock \showarticletitle{Order-preserving symmetric encryption}. In
  \bibinfo{booktitle}{\emph{Annual International Conference on the Theory and
  Applications of Cryptographic Techniques}}. Springer,
  \bibinfo{pages}{224--241}.
\newblock


\bibitem[\protect\citeauthoryear{Boldyreva, Chenette, and O’Neill}{Boldyreva
  et~al\mbox{.}}{2011}]%
        {boldyreva2011order}
\bibfield{author}{\bibinfo{person}{Alexandra Boldyreva},
  \bibinfo{person}{Nathan Chenette}, {and} \bibinfo{person}{Adam O’Neill}.}
  \bibinfo{year}{2011}\natexlab{}.
\newblock \showarticletitle{Order-preserving encryption revisited: Improved
  security analysis and alternative solutions}. In
  \bibinfo{booktitle}{\emph{Annual Cryptology Conference}}. Springer,
  \bibinfo{pages}{578--595}.
\newblock


\bibitem[\protect\citeauthoryear{Bonawitz, Ivanov, Kreuter, Marcedone, McMahan,
  Patel, Ramage, Segal, and Seth}{Bonawitz et~al\mbox{.}}{2017}]%
        {bonawitz2017practical}
\bibfield{author}{\bibinfo{person}{Keith Bonawitz}, \bibinfo{person}{Vladimir
  Ivanov}, \bibinfo{person}{Ben Kreuter}, \bibinfo{person}{Antonio Marcedone},
  \bibinfo{person}{H~Brendan McMahan}, \bibinfo{person}{Sarvar Patel},
  \bibinfo{person}{Daniel Ramage}, \bibinfo{person}{Aaron Segal}, {and}
  \bibinfo{person}{Karn Seth}.} \bibinfo{year}{2017}\natexlab{}.
\newblock \showarticletitle{Practical secure aggregation for privacy-preserving
  machine learning}. In \bibinfo{booktitle}{\emph{proceedings of the 2017 ACM
  SIGSAC Conference on Computer and Communications Security}}.
  \bibinfo{pages}{1175--1191}.
\newblock


\bibitem[\protect\citeauthoryear{Boneh, Di~Crescenzo, Ostrovsky, and
  Persiano}{Boneh et~al\mbox{.}}{2004}]%
        {boneh2004public}
\bibfield{author}{\bibinfo{person}{Dan Boneh}, \bibinfo{person}{Giovanni
  Di~Crescenzo}, \bibinfo{person}{Rafail Ostrovsky}, {and}
  \bibinfo{person}{Giuseppe Persiano}.} \bibinfo{year}{2004}\natexlab{}.
\newblock \showarticletitle{Public key encryption with keyword search}. In
  \bibinfo{booktitle}{\emph{International conference on the theory and
  applications of cryptographic techniques}}. Springer,
  \bibinfo{pages}{506--522}.
\newblock


\bibitem[\protect\citeauthoryear{Boneh, Goh, and Nissim}{Boneh
  et~al\mbox{.}}{2005}]%
        {boneh2005evaluating}
\bibfield{author}{\bibinfo{person}{Dan Boneh}, \bibinfo{person}{Eu-Jin Goh},
  {and} \bibinfo{person}{Kobbi Nissim}.} \bibinfo{year}{2005}\natexlab{}.
\newblock \showarticletitle{Evaluating 2-DNF formulas on ciphertexts}. In
  \bibinfo{booktitle}{\emph{Theory of cryptography conference}}. Springer,
  \bibinfo{pages}{325--341}.
\newblock


\bibitem[\protect\citeauthoryear{Canetti}{Canetti}{2001}]%
        {canetti2001universally}
\bibfield{author}{\bibinfo{person}{Ran Canetti}.}
  \bibinfo{year}{2001}\natexlab{}.
\newblock \showarticletitle{Universally composable security: A new paradigm for
  cryptographic protocols}. In \bibinfo{booktitle}{\emph{Proceedings 42nd IEEE
  Symposium on Foundations of Computer Science}}. IEEE,
  \bibinfo{pages}{136--145}.
\newblock


\bibitem[\protect\citeauthoryear{Cao, Yoshikawa, Xiao, and Xiong}{Cao
  et~al\mbox{.}}{2017}]%
        {cao2017quantifying}
\bibfield{author}{\bibinfo{person}{Yang Cao}, \bibinfo{person}{Masatoshi
  Yoshikawa}, \bibinfo{person}{Yonghui Xiao}, {and} \bibinfo{person}{Li
  Xiong}.} \bibinfo{year}{2017}\natexlab{}.
\newblock \showarticletitle{Quantifying differential privacy under temporal
  correlations}. In \bibinfo{booktitle}{\emph{2017 IEEE 33rd International
  Conference on Data Engineering (ICDE)}}. IEEE, \bibinfo{pages}{821--832}.
\newblock


\bibitem[\protect\citeauthoryear{Cash, Grubbs, Perry, and Ristenpart}{Cash
  et~al\mbox{.}}{2015}]%
        {cash2015leakage}
\bibfield{author}{\bibinfo{person}{David Cash}, \bibinfo{person}{Paul Grubbs},
  \bibinfo{person}{Jason Perry}, {and} \bibinfo{person}{Thomas Ristenpart}.}
  \bibinfo{year}{2015}\natexlab{}.
\newblock \showarticletitle{Leakage-abuse attacks against searchable
  encryption}. In \bibinfo{booktitle}{\emph{Proceedings of the 22nd ACM SIGSAC
  conference on computer and communications security}}.
  \bibinfo{pages}{668--679}.
\newblock


\bibitem[\protect\citeauthoryear{Cash, Jaeger, Jarecki, Jutla, Krawczyk, Rosu,
  and Steiner}{Cash et~al\mbox{.}}{2014}]%
        {cash2014dynamic}
\bibfield{author}{\bibinfo{person}{David Cash}, \bibinfo{person}{Joseph
  Jaeger}, \bibinfo{person}{Stanislaw Jarecki}, \bibinfo{person}{Charanjit~S
  Jutla}, \bibinfo{person}{Hugo Krawczyk}, \bibinfo{person}{Marcel-Catalin
  Rosu}, {and} \bibinfo{person}{Michael Steiner}.}
  \bibinfo{year}{2014}\natexlab{}.
\newblock \showarticletitle{Dynamic searchable encryption in very-large
  databases: data structures and implementation.}. In
  \bibinfo{booktitle}{\emph{NDSS}}, Vol.~\bibinfo{volume}{14}. Citeseer,
  \bibinfo{pages}{23--26}.
\newblock


\bibitem[\protect\citeauthoryear{Chen, Lai, Reiter, and Zhang}{Chen
  et~al\mbox{.}}{2018}]%
        {chen2018differentially}
\bibfield{author}{\bibinfo{person}{Guoxing Chen}, \bibinfo{person}{Ten-Hwang
  Lai}, \bibinfo{person}{Michael~K Reiter}, {and} \bibinfo{person}{Yinqian
  Zhang}.} \bibinfo{year}{2018}\natexlab{}.
\newblock \showarticletitle{Differentially private access patterns for
  searchable symmetric encryption}. In \bibinfo{booktitle}{\emph{IEEE INFOCOM
  2018-IEEE Conference on Computer Communications}}. IEEE,
  \bibinfo{pages}{810--818}.
\newblock


\bibitem[\protect\citeauthoryear{Chowdhury, Wang, He, Machanavajjhala, and
  Jha}{Chowdhury et~al\mbox{.}}{2019}]%
        {chowdhury2019crypt}
\bibfield{author}{\bibinfo{person}{Amrita~Roy Chowdhury},
  \bibinfo{person}{Chenghong Wang}, \bibinfo{person}{Xi He},
  \bibinfo{person}{Ashwin Machanavajjhala}, {and} \bibinfo{person}{Somesh
  Jha}.} \bibinfo{year}{2019}\natexlab{}.
\newblock \showarticletitle{Crypt$epsilon $: Crypto-Assisted Differential
  Privacy on Untrusted Servers}.
\newblock \bibinfo{journal}{\emph{arXiv preprint arXiv:1902.07756}}
  (\bibinfo{year}{2019}).
\newblock


\bibitem[\protect\citeauthoryear{cpdp.co}{cpdp.co}{2006}]%
        {cpdb}
\bibfield{author}{\bibinfo{person}{cpdp.co}.} \bibinfo{year}{2006}\natexlab{}.
\newblock \bibinfo{title}{Chicago Police Database}.
\newblock
\newblock
\urldef\tempurl%
\url{https://github.com/invinst/chicago-police-data}
\showURL{%
\tempurl}


\bibitem[\protect\citeauthoryear{Crooks, Burke, Cecchetti, Harel, Agarwal, and
  Alvisi}{Crooks et~al\mbox{.}}{2018}]%
        {crooks2018obladi}
\bibfield{author}{\bibinfo{person}{Natacha Crooks}, \bibinfo{person}{Matthew
  Burke}, \bibinfo{person}{Ethan Cecchetti}, \bibinfo{person}{Sitar Harel},
  \bibinfo{person}{Rachit Agarwal}, {and} \bibinfo{person}{Lorenzo Alvisi}.}
  \bibinfo{year}{2018}\natexlab{}.
\newblock \showarticletitle{Obladi: Oblivious Serializable Transactions in the
  Cloud}. In \bibinfo{booktitle}{\emph{13th {USENIX} Symposium on Operating
  Systems Design and Implementation ({OSDI} 18)}}. \bibinfo{publisher}{{USENIX}
  Association}, \bibinfo{address}{Carlsbad, CA}, \bibinfo{pages}{727–743}.
\newblock
\showISBNx{978-1-931971-47-8}
\urldef\tempurl%
\url{https://www.usenix.org/conference/osdi18/presentation/crooks}
\showURL{%
\tempurl}


\bibitem[\protect\citeauthoryear{Cummings, Krehbiel, Lai, and
  Tantipongpipat}{Cummings et~al\mbox{.}}{2018}]%
        {cummings2018differential}
\bibfield{author}{\bibinfo{person}{Rachel Cummings}, \bibinfo{person}{Sara
  Krehbiel}, \bibinfo{person}{Kevin~A Lai}, {and} \bibinfo{person}{Uthaipon
  Tantipongpipat}.} \bibinfo{year}{2018}\natexlab{}.
\newblock \showarticletitle{Differential privacy for growing databases}.
\newblock \bibinfo{journal}{\emph{arXiv preprint arXiv:1803.06416}}
  (\bibinfo{year}{2018}).
\newblock


\bibitem[\protect\citeauthoryear{Curtmola, Garay, Kamara, and
  Ostrovsky}{Curtmola et~al\mbox{.}}{2011}]%
        {curtmola2011searchable}
\bibfield{author}{\bibinfo{person}{Reza Curtmola}, \bibinfo{person}{Juan
  Garay}, \bibinfo{person}{Seny Kamara}, {and} \bibinfo{person}{Rafail
  Ostrovsky}.} \bibinfo{year}{2011}\natexlab{}.
\newblock \showarticletitle{Searchable symmetric encryption: improved
  definitions and efficient constructions}.
\newblock \bibinfo{journal}{\emph{Journal of Computer Security}}
  \bibinfo{volume}{19}, \bibinfo{number}{5} (\bibinfo{year}{2011}),
  \bibinfo{pages}{895--934}.
\newblock


\bibitem[\protect\citeauthoryear{Dautrich~Jr and Ravishankar}{Dautrich~Jr and
  Ravishankar}{2013}]%
        {dautrich2013compromising}
\bibfield{author}{\bibinfo{person}{Jonathan~L Dautrich~Jr} {and}
  \bibinfo{person}{Chinya~V Ravishankar}.} \bibinfo{year}{2013}\natexlab{}.
\newblock \showarticletitle{Compromising privacy in precise query protocols}.
  In \bibinfo{booktitle}{\emph{Proceedings of the 16th International Conference
  on Extending Database Technology}}. \bibinfo{pages}{155--166}.
\newblock


\bibitem[\protect\citeauthoryear{Demertzis, Papadopoulos, Papamanthou, and
  Shintre}{Demertzis et~al\mbox{.}}{2020}]%
        {demertzis2020seal}
\bibfield{author}{\bibinfo{person}{Ioannis Demertzis},
  \bibinfo{person}{Dimitrios Papadopoulos}, \bibinfo{person}{Charalampos
  Papamanthou}, {and} \bibinfo{person}{Saurabh Shintre}.}
  \bibinfo{year}{2020}\natexlab{}.
\newblock \showarticletitle{$\{$SEAL$\}$: Attack Mitigation for Encrypted
  Databases via Adjustable Leakage}. In \bibinfo{booktitle}{\emph{29th
  $\{$USENIX$\}$ Security Symposium ($\{$USENIX$\}$ Security 20)}}.
\newblock


\bibitem[\protect\citeauthoryear{Dwork, Kenthapadi, McSherry, Mironov, and
  Naor}{Dwork et~al\mbox{.}}{2006}]%
        {dwork2006our}
\bibfield{author}{\bibinfo{person}{Cynthia Dwork}, \bibinfo{person}{Krishnaram
  Kenthapadi}, \bibinfo{person}{Frank McSherry}, \bibinfo{person}{Ilya
  Mironov}, {and} \bibinfo{person}{Moni Naor}.}
  \bibinfo{year}{2006}\natexlab{}.
\newblock \showarticletitle{Our data, ourselves: Privacy via distributed noise
  generation}. In \bibinfo{booktitle}{\emph{Annual International Conference on
  the Theory and Applications of Cryptographic Techniques}}. Springer,
  \bibinfo{pages}{486--503}.
\newblock


\bibitem[\protect\citeauthoryear{Dwork, Naor, Pitassi, and Rothblum}{Dwork
  et~al\mbox{.}}{2010}]%
        {dwork2010differential}
\bibfield{author}{\bibinfo{person}{Cynthia Dwork}, \bibinfo{person}{Moni Naor},
  \bibinfo{person}{Toniann Pitassi}, {and} \bibinfo{person}{Guy~N Rothblum}.}
  \bibinfo{year}{2010}\natexlab{}.
\newblock \showarticletitle{Differential privacy under continual observation}.
  In \bibinfo{booktitle}{\emph{Proceedings of the forty-second ACM symposium on
  Theory of computing}}. \bibinfo{pages}{715--724}.
\newblock


\bibitem[\protect\citeauthoryear{Dwork, Roth, et~al\mbox{.}}{Dwork
  et~al\mbox{.}}{2014}]%
        {dwork2014algorithmic}
\bibfield{author}{\bibinfo{person}{Cynthia Dwork}, \bibinfo{person}{Aaron
  Roth}, {et~al\mbox{.}}} \bibinfo{year}{2014}\natexlab{}.
\newblock \showarticletitle{The algorithmic foundations of differential
  privacy.}
\newblock \bibinfo{journal}{\emph{Foundations and Trends in Theoretical
  Computer Science}} \bibinfo{volume}{9}, \bibinfo{number}{3-4}
  (\bibinfo{year}{2014}), \bibinfo{pages}{211--407}.
\newblock


\bibitem[\protect\citeauthoryear{Eskandarian and Zaharia}{Eskandarian and
  Zaharia}{2017}]%
        {eskandarian2017oblidb}
\bibfield{author}{\bibinfo{person}{Saba Eskandarian} {and}
  \bibinfo{person}{Matei Zaharia}.} \bibinfo{year}{2017}\natexlab{}.
\newblock \showarticletitle{Oblidb: Oblivious query processing using hardware
  enclaves}.
\newblock \bibinfo{journal}{\emph{arXiv preprint arXiv:1710.00458}}
  (\bibinfo{year}{2017}).
\newblock


\bibitem[\protect\citeauthoryear{Fuhry, Bahmani, Brasser, Hahn, Kerschbaum, and
  Sadeghi}{Fuhry et~al\mbox{.}}{2017}]%
        {fuhry2017hardidx}
\bibfield{author}{\bibinfo{person}{Benny Fuhry}, \bibinfo{person}{Raad
  Bahmani}, \bibinfo{person}{Ferdinand Brasser}, \bibinfo{person}{Florian
  Hahn}, \bibinfo{person}{Florian Kerschbaum}, {and}
  \bibinfo{person}{Ahmad-Reza Sadeghi}.} \bibinfo{year}{2017}\natexlab{}.
\newblock \showarticletitle{HardIDX: Practical and secure index with SGX}. In
  \bibinfo{booktitle}{\emph{IFIP Annual Conference on Data and Applications
  Security and Privacy}}. Springer, \bibinfo{pages}{386--408}.
\newblock


\bibitem[\protect\citeauthoryear{Gentry}{Gentry}{2009}]%
        {gentry2009fully}
\bibfield{author}{\bibinfo{person}{Craig Gentry}.}
  \bibinfo{year}{2009}\natexlab{}.
\newblock \showarticletitle{Fully homomorphic encryption using ideal lattices}.
  In \bibinfo{booktitle}{\emph{Proceedings of the forty-first annual ACM
  symposium on Theory of computing}}. \bibinfo{pages}{169--178}.
\newblock


\bibitem[\protect\citeauthoryear{Ghareh~Chamani, Papadopoulos, Papamanthou, and
  Jalili}{Ghareh~Chamani et~al\mbox{.}}{2018}]%
        {ghareh2018new}
\bibfield{author}{\bibinfo{person}{Javad Ghareh~Chamani},
  \bibinfo{person}{Dimitrios Papadopoulos}, \bibinfo{person}{Charalampos
  Papamanthou}, {and} \bibinfo{person}{Rasool Jalili}.}
  \bibinfo{year}{2018}\natexlab{}.
\newblock \showarticletitle{New constructions for forward and backward private
  symmetric searchable encryption}. In \bibinfo{booktitle}{\emph{Proceedings of
  the 2018 ACM SIGSAC Conference on Computer and Communications Security}}.
  \bibinfo{pages}{1038--1055}.
\newblock


\bibitem[\protect\citeauthoryear{Goldreich}{Goldreich}{2009}]%
        {goldreich2009foundations}
\bibfield{author}{\bibinfo{person}{Oded Goldreich}.}
  \bibinfo{year}{2009}\natexlab{}.
\newblock \bibinfo{booktitle}{\emph{Foundations of cryptography: volume 2,
  basic applications}}.
\newblock \bibinfo{publisher}{Cambridge university press}.
\newblock


\bibitem[\protect\citeauthoryear{Grubbs, Lacharit{\'e}, Minaud, and
  Paterson}{Grubbs et~al\mbox{.}}{2018}]%
        {grubbs2018pump}
\bibfield{author}{\bibinfo{person}{Paul Grubbs}, \bibinfo{person}{Marie-Sarah
  Lacharit{\'e}}, \bibinfo{person}{Brice Minaud}, {and}
  \bibinfo{person}{Kenneth~G Paterson}.} \bibinfo{year}{2018}\natexlab{}.
\newblock \showarticletitle{Pump up the volume: Practical database
  reconstruction from volume leakage on range queries}. In
  \bibinfo{booktitle}{\emph{Proceedings of the 2018 ACM SIGSAC Conference on
  Computer and Communications Security}}. \bibinfo{pages}{315--331}.
\newblock


\bibitem[\protect\citeauthoryear{Grubbs, Lacharit{\'e}, Minaud, and
  Paterson}{Grubbs et~al\mbox{.}}{2019}]%
        {grubbs2019learning}
\bibfield{author}{\bibinfo{person}{Paul Grubbs}, \bibinfo{person}{Marie-Sarah
  Lacharit{\'e}}, \bibinfo{person}{Brice Minaud}, {and}
  \bibinfo{person}{Kenneth~G Paterson}.} \bibinfo{year}{2019}\natexlab{}.
\newblock \showarticletitle{Learning to reconstruct: Statistical learning
  theory and encrypted database attacks}. In \bibinfo{booktitle}{\emph{2019
  IEEE Symposium on Security and Privacy (SP)}}. IEEE,
  \bibinfo{pages}{1067--1083}.
\newblock


\bibitem[\protect\citeauthoryear{Gui, Johnson, and Warinschi}{Gui
  et~al\mbox{.}}{2019}]%
        {gui2019encrypted}
\bibfield{author}{\bibinfo{person}{Zichen Gui}, \bibinfo{person}{Oliver
  Johnson}, {and} \bibinfo{person}{Bogdan Warinschi}.}
  \bibinfo{year}{2019}\natexlab{}.
\newblock \showarticletitle{Encrypted databases: New volume attacks against
  range queries}. In \bibinfo{booktitle}{\emph{Proceedings of the 2019 ACM
  SIGSAC Conference on Computer and Communications Security}}.
  \bibinfo{pages}{361--378}.
\newblock


\bibitem[\protect\citeauthoryear{Hacig{\"u}m{\"u}{\c{s}}, Iyer, Li, and
  Mehrotra}{Hacig{\"u}m{\"u}{\c{s}} et~al\mbox{.}}{2002}]%
        {hacigumucs2002executing}
\bibfield{author}{\bibinfo{person}{Hakan Hacig{\"u}m{\"u}{\c{s}}},
  \bibinfo{person}{Bala Iyer}, \bibinfo{person}{Chen Li}, {and}
  \bibinfo{person}{Sharad Mehrotra}.} \bibinfo{year}{2002}\natexlab{}.
\newblock \showarticletitle{Executing SQL over encrypted data in the
  database-service-provider model}. In \bibinfo{booktitle}{\emph{Proceedings of
  the 2002 ACM SIGMOD international conference on Management of data}}.
  \bibinfo{pages}{216--227}.
\newblock


\bibitem[\protect\citeauthoryear{Hore, Mehrotra, Canim, and Kantarcioglu}{Hore
  et~al\mbox{.}}{2012}]%
        {hore2012secure}
\bibfield{author}{\bibinfo{person}{Bijit Hore}, \bibinfo{person}{Sharad
  Mehrotra}, \bibinfo{person}{Mustafa Canim}, {and} \bibinfo{person}{Murat
  Kantarcioglu}.} \bibinfo{year}{2012}\natexlab{}.
\newblock \showarticletitle{Secure multidimensional range queries over
  outsourced data}.
\newblock \bibinfo{journal}{\emph{The VLDB Journal}} \bibinfo{volume}{21},
  \bibinfo{number}{3} (\bibinfo{year}{2012}), \bibinfo{pages}{333--358}.
\newblock


\bibitem[\protect\citeauthoryear{Hore, Mehrotra, and Tsudik}{Hore
  et~al\mbox{.}}{2004}]%
        {hore2004privacy}
\bibfield{author}{\bibinfo{person}{Bijit Hore}, \bibinfo{person}{Sharad
  Mehrotra}, {and} \bibinfo{person}{Gene Tsudik}.}
  \bibinfo{year}{2004}\natexlab{}.
\newblock \showarticletitle{A privacy-preserving index for range queries}. In
  \bibinfo{booktitle}{\emph{Proceedings of the Thirtieth international
  conference on Very large data bases-Volume 30}}. \bibinfo{pages}{720--731}.
\newblock


\bibitem[\protect\citeauthoryear{Ishai, Kushilevitz, Lu, and Ostrovsky}{Ishai
  et~al\mbox{.}}{2016}]%
        {ishai2016private}
\bibfield{author}{\bibinfo{person}{Yuval Ishai}, \bibinfo{person}{Eyal
  Kushilevitz}, \bibinfo{person}{Steve Lu}, {and} \bibinfo{person}{Rafail
  Ostrovsky}.} \bibinfo{year}{2016}\natexlab{}.
\newblock \showarticletitle{Private large-scale databases with distributed
  searchable symmetric encryption}. In
  \bibinfo{booktitle}{\emph{Cryptographers’ Track at the RSA Conference}}.
  Springer, \bibinfo{pages}{90--107}.
\newblock


\bibitem[\protect\citeauthoryear{Johnson, Near, and Song}{Johnson
  et~al\mbox{.}}{2018}]%
        {johnson2018towards}
\bibfield{author}{\bibinfo{person}{Noah Johnson}, \bibinfo{person}{Joseph~P
  Near}, {and} \bibinfo{person}{Dawn Song}.} \bibinfo{year}{2018}\natexlab{}.
\newblock \showarticletitle{Towards practical differential privacy for SQL
  queries}.
\newblock \bibinfo{journal}{\emph{Proceedings of the VLDB Endowment}}
  \bibinfo{volume}{11}, \bibinfo{number}{5} (\bibinfo{year}{2018}),
  \bibinfo{pages}{526--539}.
\newblock


\bibitem[\protect\citeauthoryear{Kamara and Moataz}{Kamara and Moataz}{2018}]%
        {kamara2018sql}
\bibfield{author}{\bibinfo{person}{Seny Kamara} {and} \bibinfo{person}{Tarik
  Moataz}.} \bibinfo{year}{2018}\natexlab{}.
\newblock \showarticletitle{SQL on structurally-encrypted databases}. In
  \bibinfo{booktitle}{\emph{International Conference on the Theory and
  Application of Cryptology and Information Security}}. Springer,
  \bibinfo{pages}{149--180}.
\newblock


\bibitem[\protect\citeauthoryear{Kamara and Moataz}{Kamara and Moataz}{2019}]%
        {kamara2019computationally}
\bibfield{author}{\bibinfo{person}{Seny Kamara} {and} \bibinfo{person}{Tarik
  Moataz}.} \bibinfo{year}{2019}\natexlab{}.
\newblock \showarticletitle{Computationally volume-hiding structured
  encryption}. In \bibinfo{booktitle}{\emph{Annual International Conference on
  the Theory and Applications of Cryptographic Techniques}}. Springer,
  \bibinfo{pages}{183--213}.
\newblock


\bibitem[\protect\citeauthoryear{Kamara, Mohassel, and Raykova}{Kamara
  et~al\mbox{.}}{2011}]%
        {kamara2011outsourcing}
\bibfield{author}{\bibinfo{person}{Seny Kamara}, \bibinfo{person}{Payman
  Mohassel}, {and} \bibinfo{person}{Mariana Raykova}.}
  \bibinfo{year}{2011}\natexlab{}.
\newblock \showarticletitle{Outsourcing Multi-Party Computation.}
\newblock \bibinfo{journal}{\emph{IACR Cryptol. Eprint Arch.}}
  \bibinfo{volume}{2011} (\bibinfo{year}{2011}), \bibinfo{pages}{272}.
\newblock


\bibitem[\protect\citeauthoryear{Kamara, Papamanthou, and Roeder}{Kamara
  et~al\mbox{.}}{2012}]%
        {kamara2012dynamic}
\bibfield{author}{\bibinfo{person}{Seny Kamara}, \bibinfo{person}{Charalampos
  Papamanthou}, {and} \bibinfo{person}{Tom Roeder}.}
  \bibinfo{year}{2012}\natexlab{}.
\newblock \showarticletitle{Dynamic searchable symmetric encryption}. In
  \bibinfo{booktitle}{\emph{Proceedings of the 2012 ACM conference on Computer
  and communications security}}. \bibinfo{pages}{965--976}.
\newblock


\bibitem[\protect\citeauthoryear{Kellaris, Kollios, Nissim, and
  O'neill}{Kellaris et~al\mbox{.}}{2016}]%
        {kellaris2016generic}
\bibfield{author}{\bibinfo{person}{Georgios Kellaris}, \bibinfo{person}{George
  Kollios}, \bibinfo{person}{Kobbi Nissim}, {and} \bibinfo{person}{Adam
  O'neill}.} \bibinfo{year}{2016}\natexlab{}.
\newblock \showarticletitle{Generic attacks on secure outsourced databases}. In
  \bibinfo{booktitle}{\emph{Proceedings of the 2016 ACM SIGSAC Conference on
  Computer and Communications Security}}. \bibinfo{pages}{1329--1340}.
\newblock


\bibitem[\protect\citeauthoryear{Kellaris, Kollios, Nissim, and
  O'Neill}{Kellaris et~al\mbox{.}}{2017}]%
        {kellaris2017accessing}
\bibfield{author}{\bibinfo{person}{Georgios Kellaris}, \bibinfo{person}{George
  Kollios}, \bibinfo{person}{Kobbi Nissim}, {and} \bibinfo{person}{Adam
  O'Neill}.} \bibinfo{year}{2017}\natexlab{}.
\newblock \showarticletitle{Accessing data while preserving privacy}.
\newblock \bibinfo{journal}{\emph{arXiv preprint arXiv:1706.01552}}
  (\bibinfo{year}{2017}).
\newblock


\bibitem[\protect\citeauthoryear{Keller, Pastro, and Rotaru}{Keller
  et~al\mbox{.}}{2018}]%
        {keller2018overdrive}
\bibfield{author}{\bibinfo{person}{Marcel Keller}, \bibinfo{person}{Valerio
  Pastro}, {and} \bibinfo{person}{Dragos Rotaru}.}
  \bibinfo{year}{2018}\natexlab{}.
\newblock \showarticletitle{Overdrive: Making SPDZ great again}. In
  \bibinfo{booktitle}{\emph{Annual International Conference on the Theory and
  Applications of Cryptographic Techniques}}. Springer,
  \bibinfo{pages}{158--189}.
\newblock


\bibitem[\protect\citeauthoryear{Keller, Scholl, and Smart}{Keller
  et~al\mbox{.}}{2013}]%
        {keller2013architecture}
\bibfield{author}{\bibinfo{person}{Marcel Keller}, \bibinfo{person}{Peter
  Scholl}, {and} \bibinfo{person}{Nigel~P Smart}.}
  \bibinfo{year}{2013}\natexlab{}.
\newblock \showarticletitle{An architecture for practical actively secure MPC
  with dishonest majority}. In \bibinfo{booktitle}{\emph{Proceedings of the
  2013 ACM SIGSAC conference on Computer \& communications security}}.
  \bibinfo{pages}{549--560}.
\newblock


\bibitem[\protect\citeauthoryear{Kifer and Machanavajjhala}{Kifer and
  Machanavajjhala}{2011}]%
        {kifer2011no}
\bibfield{author}{\bibinfo{person}{Daniel Kifer} {and} \bibinfo{person}{Ashwin
  Machanavajjhala}.} \bibinfo{year}{2011}\natexlab{}.
\newblock \showarticletitle{No free lunch in data privacy}. In
  \bibinfo{booktitle}{\emph{Proceedings of the 2011 ACM SIGMOD International
  Conference on Management of data}}. \bibinfo{pages}{193--204}.
\newblock


\bibitem[\protect\citeauthoryear{Kotsogiannis, Tao, He, Fanaeepour,
  Machanavajjhala, Hay, and Miklau}{Kotsogiannis et~al\mbox{.}}{2019a}]%
        {kotsogiannis2019privatesql}
\bibfield{author}{\bibinfo{person}{Ios Kotsogiannis}, \bibinfo{person}{Yuchao
  Tao}, \bibinfo{person}{Xi He}, \bibinfo{person}{Maryam Fanaeepour},
  \bibinfo{person}{Ashwin Machanavajjhala}, \bibinfo{person}{Michael Hay},
  {and} \bibinfo{person}{Gerome Miklau}.} \bibinfo{year}{2019}\natexlab{a}.
\newblock \showarticletitle{Privatesql: a differentially private sql query
  engine}.
\newblock \bibinfo{journal}{\emph{Proceedings of the VLDB Endowment}}
  \bibinfo{volume}{12}, \bibinfo{number}{11} (\bibinfo{year}{2019}),
  \bibinfo{pages}{1371--1384}.
\newblock


\bibitem[\protect\citeauthoryear{Kotsogiannis, Tao, Machanavajjhala, Miklau,
  and Hay}{Kotsogiannis et~al\mbox{.}}{2019b}]%
        {kotsogiannis2019architecting}
\bibfield{author}{\bibinfo{person}{Ios Kotsogiannis}, \bibinfo{person}{Yuchao
  Tao}, \bibinfo{person}{Ashwin Machanavajjhala}, \bibinfo{person}{Gerome
  Miklau}, {and} \bibinfo{person}{Michael Hay}.}
  \bibinfo{year}{2019}\natexlab{b}.
\newblock \showarticletitle{Architecting a Differentially Private SQL Engine.}.
  In \bibinfo{booktitle}{\emph{CIDR}}.
\newblock


\bibitem[\protect\citeauthoryear{L\'{e}cuyer, Spahn, Vodrahalli, Geambasu, and
  Hsu}{L\'{e}cuyer et~al\mbox{.}}{2019}]%
        {lecuyer2019sage}
\bibfield{author}{\bibinfo{person}{Mathias L\'{e}cuyer}, \bibinfo{person}{Riley
  Spahn}, \bibinfo{person}{Kiran Vodrahalli}, \bibinfo{person}{Roxana
  Geambasu}, {and} \bibinfo{person}{Daniel Hsu}.}
  \bibinfo{year}{2019}\natexlab{}.
\newblock \showarticletitle{Privacy Accounting and Quality Control in the Sage
  Differentially Private ML Platform}. In \bibinfo{booktitle}{\emph{Proceedings
  of the 27th ACM Symposium on Operating Systems Principles}} (Huntsville,
  Ontario, Canada) \emph{(\bibinfo{series}{SOSP '19})}.
  \bibinfo{publisher}{Association for Computing Machinery},
  \bibinfo{address}{New York, NY, USA}, \bibinfo{pages}{181–195}.
\newblock
\showISBNx{9781450368735}
\urldef\tempurl%
\url{https://doi.org/10.1145/3341301.3359639}
\showDOI{\tempurl}


\bibitem[\protect\citeauthoryear{Lindell}{Lindell}{2017}]%
        {lindell2017simulate}
\bibfield{author}{\bibinfo{person}{Yehuda Lindell}.}
  \bibinfo{year}{2017}\natexlab{}.
\newblock \showarticletitle{How to simulate it--a tutorial on the simulation
  proof technique}.
\newblock \bibinfo{journal}{\emph{Tutorials on the Foundations of
  Cryptography}} (\bibinfo{year}{2017}), \bibinfo{pages}{277--346}.
\newblock


\bibitem[\protect\citeauthoryear{Liu, Chakraborty, and Mittal}{Liu
  et~al\mbox{.}}{2016}]%
        {liu2016dependence}
\bibfield{author}{\bibinfo{person}{Changchang Liu}, \bibinfo{person}{Supriyo
  Chakraborty}, {and} \bibinfo{person}{Prateek Mittal}.}
  \bibinfo{year}{2016}\natexlab{}.
\newblock \showarticletitle{Dependence Makes You Vulnberable: Differential
  Privacy Under Dependent Tuples.}. In \bibinfo{booktitle}{\emph{NDSS}},
  Vol.~\bibinfo{volume}{16}. \bibinfo{pages}{21--24}.
\newblock


\bibitem[\protect\citeauthoryear{Lu}{Lu}{2012}]%
        {lu2012privacy}
\bibfield{author}{\bibinfo{person}{Yanbin Lu}.}
  \bibinfo{year}{2012}\natexlab{}.
\newblock \showarticletitle{Privacy-preserving Logarithmic-time Search on
  Encrypted Data in Cloud.}. In \bibinfo{booktitle}{\emph{NDSS}}.
\newblock


\bibitem[\protect\citeauthoryear{Luo, Pan, Tholoniat, Cidon, Geambasu, and
  L{\'e}cuyer}{Luo et~al\mbox{.}}{2021}]%
        {luo2021privacy}
\bibfield{author}{\bibinfo{person}{Tao Luo}, \bibinfo{person}{Mingen Pan},
  \bibinfo{person}{Pierre Tholoniat}, \bibinfo{person}{Asaf Cidon},
  \bibinfo{person}{Roxana Geambasu}, {and} \bibinfo{person}{Mathias
  L{\'e}cuyer}.} \bibinfo{year}{2021}\natexlab{}.
\newblock \showarticletitle{Privacy Budget Scheduling}. In
  \bibinfo{booktitle}{\emph{15th $\{$USENIX$\}$ Symposium on Operating Systems
  Design and Implementation ($\{$OSDI$\}$ 21)}}. \bibinfo{pages}{55--74}.
\newblock


\bibitem[\protect\citeauthoryear{Mazloom and Gordon}{Mazloom and
  Gordon}{2018}]%
        {mazloom2018secure}
\bibfield{author}{\bibinfo{person}{Sahar Mazloom} {and} \bibinfo{person}{S~Dov
  Gordon}.} \bibinfo{year}{2018}\natexlab{}.
\newblock \showarticletitle{Secure computation with differentially private
  access patterns}. In \bibinfo{booktitle}{\emph{Proceedings of the 2018 ACM
  SIGSAC Conference on Computer and Communications Security}}.
  \bibinfo{pages}{490--507}.
\newblock


\bibitem[\protect\citeauthoryear{McSherry}{McSherry}{2009}]%
        {mcsherry2009privacy}
\bibfield{author}{\bibinfo{person}{Frank~D McSherry}.}
  \bibinfo{year}{2009}\natexlab{}.
\newblock \showarticletitle{Privacy integrated queries: an extensible platform
  for privacy-preserving data analysis}. In
  \bibinfo{booktitle}{\emph{Proceedings of the 2009 ACM SIGMOD International
  Conference on Management of data}}. \bibinfo{pages}{19--30}.
\newblock


\bibitem[\protect\citeauthoryear{Mironov, Pandey, Reingold, and Vadhan}{Mironov
  et~al\mbox{.}}{2009}]%
        {mironov2009computational}
\bibfield{author}{\bibinfo{person}{Ilya Mironov}, \bibinfo{person}{Omkant
  Pandey}, \bibinfo{person}{Omer Reingold}, {and} \bibinfo{person}{Salil
  Vadhan}.} \bibinfo{year}{2009}\natexlab{}.
\newblock \showarticletitle{Computational differential privacy}. In
  \bibinfo{booktitle}{\emph{Annual International Cryptology Conference}}.
  Springer, \bibinfo{pages}{126--142}.
\newblock


\bibitem[\protect\citeauthoryear{Mohassel and Zhang}{Mohassel and
  Zhang}{2017}]%
        {mohassel2017secureml}
\bibfield{author}{\bibinfo{person}{Payman Mohassel} {and}
  \bibinfo{person}{Yupeng Zhang}.} \bibinfo{year}{2017}\natexlab{}.
\newblock \showarticletitle{Secureml: A system for scalable privacy-preserving
  machine learning}. In \bibinfo{booktitle}{\emph{2017 IEEE symposium on
  security and privacy (SP)}}. IEEE, \bibinfo{pages}{19--38}.
\newblock


\bibitem[\protect\citeauthoryear{Naveed, Prabhakaran, and Gunter}{Naveed
  et~al\mbox{.}}{2014}]%
        {naveed2014dynamic}
\bibfield{author}{\bibinfo{person}{Muhammad Naveed}, \bibinfo{person}{Manoj
  Prabhakaran}, {and} \bibinfo{person}{Carl~A Gunter}.}
  \bibinfo{year}{2014}\natexlab{}.
\newblock \showarticletitle{Dynamic searchable encryption via blind storage}.
  In \bibinfo{booktitle}{\emph{2014 IEEE Symposium on Security and Privacy}}.
  IEEE, \bibinfo{pages}{639--654}.
\newblock


\bibitem[\protect\citeauthoryear{Pandey and Rouselakis}{Pandey and
  Rouselakis}{2012}]%
        {pandey2012property}
\bibfield{author}{\bibinfo{person}{Omkant Pandey} {and} \bibinfo{person}{Yannis
  Rouselakis}.} \bibinfo{year}{2012}\natexlab{}.
\newblock \showarticletitle{Property preserving symmetric encryption}. In
  \bibinfo{booktitle}{\emph{Annual International Conference on the Theory and
  Applications of Cryptographic Techniques}}. Springer,
  \bibinfo{pages}{375--391}.
\newblock


\bibitem[\protect\citeauthoryear{Patel, Persiano, Yeo, and Yung}{Patel
  et~al\mbox{.}}{2019}]%
        {patel2019mitigating}
\bibfield{author}{\bibinfo{person}{Sarvar Patel}, \bibinfo{person}{Giuseppe
  Persiano}, \bibinfo{person}{Kevin Yeo}, {and} \bibinfo{person}{Moti Yung}.}
  \bibinfo{year}{2019}\natexlab{}.
\newblock \showarticletitle{Mitigating leakage in secure cloud-hosted data
  structures: Volume-hiding for multi-maps via hashing}. In
  \bibinfo{booktitle}{\emph{Proceedings of the 2019 ACM SIGSAC Conference on
  Computer and Communications Security}}. \bibinfo{pages}{79--93}.
\newblock


\bibitem[\protect\citeauthoryear{Poddar, Boelter, and Popa}{Poddar
  et~al\mbox{.}}{2016}]%
        {poddar2016arx}
\bibfield{author}{\bibinfo{person}{Rishabh Poddar}, \bibinfo{person}{Tobias
  Boelter}, {and} \bibinfo{person}{Raluca~Ada Popa}.}
  \bibinfo{year}{2016}\natexlab{}.
\newblock \showarticletitle{Arx: A Strongly Encrypted Database System.}
\newblock \bibinfo{journal}{\emph{IACR Cryptol. ePrint Arch.}}
  \bibinfo{volume}{2016} (\bibinfo{year}{2016}), \bibinfo{pages}{591}.
\newblock


\bibitem[\protect\citeauthoryear{Popa, Redfield, Zeldovich, and
  Balakrishnan}{Popa et~al\mbox{.}}{2012}]%
        {popa2012cryptdb}
\bibfield{author}{\bibinfo{person}{Raluca~Ada Popa},
  \bibinfo{person}{Catherine~MS Redfield}, \bibinfo{person}{Nickolai
  Zeldovich}, {and} \bibinfo{person}{Hari Balakrishnan}.}
  \bibinfo{year}{2012}\natexlab{}.
\newblock \showarticletitle{CryptDB: processing queries on an encrypted
  database}.
\newblock \bibinfo{journal}{\emph{Commun. ACM}} \bibinfo{volume}{55},
  \bibinfo{number}{9} (\bibinfo{year}{2012}), \bibinfo{pages}{103--111}.
\newblock


\bibitem[\protect\citeauthoryear{Priebe, Vaswani, and Costa}{Priebe
  et~al\mbox{.}}{2018}]%
        {priebe2018enclavedb}
\bibfield{author}{\bibinfo{person}{Christian Priebe}, \bibinfo{person}{Kapil
  Vaswani}, {and} \bibinfo{person}{Manuel Costa}.}
  \bibinfo{year}{2018}\natexlab{}.
\newblock \showarticletitle{Enclavedb: A secure database using SGX}. In
  \bibinfo{booktitle}{\emph{2018 IEEE Symposium on Security and Privacy (SP)}}.
  IEEE, \bibinfo{pages}{264--278}.
\newblock


\bibitem[\protect\citeauthoryear{REPOSITORY}{REPOSITORY}{[n.d.]}]%
        {tpcds}
\bibfield{author}{\bibinfo{person}{RELATIONAL~DATASET REPOSITORY}.}
  \bibinfo{year}{[n.d.]}\natexlab{}.
\newblock \bibinfo{title}{TPCDS}.
\newblock
\newblock
\urldef\tempurl%
\url{https://relational.fit.cvut.cz/dataset/TPCDS}
\showURL{%
\tempurl}


\bibitem[\protect\citeauthoryear{Samanthula, Jiang, and Bertino}{Samanthula
  et~al\mbox{.}}{2014}]%
        {samanthula2014privacy}
\bibfield{author}{\bibinfo{person}{Bharath~Kumar Samanthula},
  \bibinfo{person}{Wei Jiang}, {and} \bibinfo{person}{Elisa Bertino}.}
  \bibinfo{year}{2014}\natexlab{}.
\newblock \showarticletitle{Privacy-preserving complex query evaluation over
  semantically secure encrypted data}. In \bibinfo{booktitle}{\emph{European
  Symposium on Research in Computer Security}}. Springer,
  \bibinfo{pages}{400--418}.
\newblock


\bibitem[\protect\citeauthoryear{Shang, Oya, Peter, and Kerschbaum}{Shang
  et~al\mbox{.}}{2021}]%
        {shang2021obfuscated}
\bibfield{author}{\bibinfo{person}{Zhiwei Shang}, \bibinfo{person}{Simon Oya},
  \bibinfo{person}{Andreas Peter}, {and} \bibinfo{person}{Florian Kerschbaum}.}
  \bibinfo{year}{2021}\natexlab{}.
\newblock \showarticletitle{Obfuscated Access and Search Patterns in Searchable
  Encryption}.
\newblock \bibinfo{journal}{\emph{arXiv preprint arXiv:2102.09651}}
  (\bibinfo{year}{2021}).
\newblock


\bibitem[\protect\citeauthoryear{Shen, Shi, and Waters}{Shen
  et~al\mbox{.}}{2009}]%
        {shen2009predicate}
\bibfield{author}{\bibinfo{person}{Emily Shen}, \bibinfo{person}{Elaine Shi},
  {and} \bibinfo{person}{Brent Waters}.} \bibinfo{year}{2009}\natexlab{}.
\newblock \showarticletitle{Predicate privacy in encryption systems}. In
  \bibinfo{booktitle}{\emph{Theory of Cryptography Conference}}. Springer,
  \bibinfo{pages}{457--473}.
\newblock


\bibitem[\protect\citeauthoryear{Shi, Bethencourt, Chan, Song, and Perrig}{Shi
  et~al\mbox{.}}{2007}]%
        {shi2007multi}
\bibfield{author}{\bibinfo{person}{Elaine Shi}, \bibinfo{person}{John
  Bethencourt}, \bibinfo{person}{TH~Hubert Chan}, \bibinfo{person}{Dawn Song},
  {and} \bibinfo{person}{Adrian Perrig}.} \bibinfo{year}{2007}\natexlab{}.
\newblock \showarticletitle{Multi-dimensional range query over encrypted data}.
  In \bibinfo{booktitle}{\emph{2007 IEEE Symposium on Security and Privacy
  (SP'07)}}. IEEE, \bibinfo{pages}{350--364}.
\newblock


\bibitem[\protect\citeauthoryear{Song, Wang, and Chaudhuri}{Song
  et~al\mbox{.}}{2017}]%
        {song2017pufferfish}
\bibfield{author}{\bibinfo{person}{Shuang Song}, \bibinfo{person}{Yizhen Wang},
  {and} \bibinfo{person}{Kamalika Chaudhuri}.} \bibinfo{year}{2017}\natexlab{}.
\newblock \showarticletitle{Pufferfish privacy mechanisms for correlated data}.
  In \bibinfo{booktitle}{\emph{Proceedings of the 2017 ACM International
  Conference on Management of Data}}. \bibinfo{pages}{1291--1306}.
\newblock


\bibitem[\protect\citeauthoryear{Srivastava, Dar, Jagadish, and
  Levy}{Srivastava et~al\mbox{.}}{1996}]%
        {srivastava1996answering}
\bibfield{author}{\bibinfo{person}{Divesh Srivastava}, \bibinfo{person}{Shaul
  Dar}, \bibinfo{person}{Hosagrahar~V Jagadish}, {and} \bibinfo{person}{Alon~Y
  Levy}.} \bibinfo{year}{1996}\natexlab{}.
\newblock \showarticletitle{Answering queries with aggregation using views}. In
  \bibinfo{booktitle}{\emph{VLDB}}, Vol.~\bibinfo{volume}{96}.
  \bibinfo{pages}{318--329}.
\newblock


\bibitem[\protect\citeauthoryear{Stefanov, Papamanthou, and Shi}{Stefanov
  et~al\mbox{.}}{2014}]%
        {stefanov2014practical}
\bibfield{author}{\bibinfo{person}{Emil Stefanov}, \bibinfo{person}{Charalampos
  Papamanthou}, {and} \bibinfo{person}{Elaine Shi}.}
  \bibinfo{year}{2014}\natexlab{}.
\newblock \showarticletitle{Practical Dynamic Searchable Encryption with Small
  Leakage.}. In \bibinfo{booktitle}{\emph{NDSS}}, Vol.~\bibinfo{volume}{71}.
  \bibinfo{pages}{72--75}.
\newblock


\bibitem[\protect\citeauthoryear{Tan, Knott, Tian, and Wu}{Tan
  et~al\mbox{.}}{2021}]%
        {tan2021cryptgpu}
\bibfield{author}{\bibinfo{person}{Sijun Tan}, \bibinfo{person}{Brian Knott},
  \bibinfo{person}{Yuan Tian}, {and} \bibinfo{person}{David~J Wu}.}
  \bibinfo{year}{2021}\natexlab{}.
\newblock \showarticletitle{CRYPTGPU: Fast Privacy-Preserving Machine Learning
  on the GPU}.
\newblock \bibinfo{journal}{\emph{arXiv preprint arXiv:2104.10949}}
  (\bibinfo{year}{2021}).
\newblock


\bibitem[\protect\citeauthoryear{Tao, He, Machanavajjhala, and Roy}{Tao
  et~al\mbox{.}}{2020}]%
        {tao2020computing}
\bibfield{author}{\bibinfo{person}{Yuchao Tao}, \bibinfo{person}{Xi He},
  \bibinfo{person}{Ashwin Machanavajjhala}, {and} \bibinfo{person}{Sudeepa
  Roy}.} \bibinfo{year}{2020}\natexlab{}.
\newblock \showarticletitle{Computing Local Sensitivities of Counting Queries
  with Joins}. In \bibinfo{booktitle}{\emph{Proceedings of the 2020 ACM SIGMOD
  International Conference on Management of Data}}. \bibinfo{pages}{479--494}.
\newblock


\bibitem[\protect\citeauthoryear{Vinayagamurthy, Gribov, and
  Gorbunov}{Vinayagamurthy et~al\mbox{.}}{2019}]%
        {vinayagamurthy2019stealthdb}
\bibfield{author}{\bibinfo{person}{Dhinakaran Vinayagamurthy},
  \bibinfo{person}{Alexey Gribov}, {and} \bibinfo{person}{Sergey Gorbunov}.}
  \bibinfo{year}{2019}\natexlab{}.
\newblock \showarticletitle{Stealthdb: a scalable encrypted database with full
  SQL query support}.
\newblock \bibinfo{journal}{\emph{Proceedings on Privacy Enhancing
  Technologies}} \bibinfo{volume}{2019}, \bibinfo{number}{3}
  (\bibinfo{year}{2019}), \bibinfo{pages}{370--388}.
\newblock


\bibitem[\protect\citeauthoryear{Wagh, Cuff, and Mittal}{Wagh
  et~al\mbox{.}}{2018}]%
        {wagh2018differentially}
\bibfield{author}{\bibinfo{person}{Sameer Wagh}, \bibinfo{person}{Paul Cuff},
  {and} \bibinfo{person}{Prateek Mittal}.} \bibinfo{year}{2018}\natexlab{}.
\newblock \showarticletitle{Differentially private oblivious ram}.
\newblock \bibinfo{journal}{\emph{Proceedings on Privacy Enhancing
  Technologies}} \bibinfo{volume}{2018}, \bibinfo{number}{4}
  (\bibinfo{year}{2018}), \bibinfo{pages}{64--84}.
\newblock


\bibitem[\protect\citeauthoryear{Wang, Bater, Nayak, and Machanavajjhala}{Wang
  et~al\mbox{.}}{2021}]%
        {wang2021dp}
\bibfield{author}{\bibinfo{person}{Chenghong Wang}, \bibinfo{person}{Johes
  Bater}, \bibinfo{person}{Kartik Nayak}, {and} \bibinfo{person}{Ashwin
  Machanavajjhala}.} \bibinfo{year}{2021}\natexlab{}.
\newblock \showarticletitle{DP-Sync: Hiding Update Patterns in Secure
  OutsourcedDatabases with Differential Privacy}.
\newblock \bibinfo{journal}{\emph{arXiv preprint arXiv:2103.15942}}
  (\bibinfo{year}{2021}).
\newblock


\bibitem[\protect\citeauthoryear{Wang and Zhao}{Wang and Zhao}{2018}]%
        {wang2018order}
\bibfield{author}{\bibinfo{person}{Xingchen Wang} {and} \bibinfo{person}{Yunlei
  Zhao}.} \bibinfo{year}{2018}\natexlab{}.
\newblock \showarticletitle{Order-revealing encryption: file-injection attack
  and forward security}. In \bibinfo{booktitle}{\emph{European Symposium on
  Research in Computer Security}}. Springer, \bibinfo{pages}{101--121}.
\newblock


\bibitem[\protect\citeauthoryear{Wilson, Zhang, Lam, Desfontaines,
  Simmons-Marengo, and Gipson}{Wilson et~al\mbox{.}}{2020}]%
        {wilson2020differentially}
\bibfield{author}{\bibinfo{person}{Royce~J Wilson},
  \bibinfo{person}{Celia~Yuxin Zhang}, \bibinfo{person}{William Lam},
  \bibinfo{person}{Damien Desfontaines}, \bibinfo{person}{Daniel
  Simmons-Marengo}, {and} \bibinfo{person}{Bryant Gipson}.}
  \bibinfo{year}{2020}\natexlab{}.
\newblock \showarticletitle{Differentially private sql with bounded user
  contribution}.
\newblock \bibinfo{journal}{\emph{Proceedings on privacy enhancing
  technologies}} \bibinfo{volume}{2020}, \bibinfo{number}{2}
  (\bibinfo{year}{2020}), \bibinfo{pages}{230--250}.
\newblock


\bibitem[\protect\citeauthoryear{Xiao and Xiong}{Xiao and Xiong}{2015}]%
        {xiao2015protecting}
\bibfield{author}{\bibinfo{person}{Yonghui Xiao} {and} \bibinfo{person}{Li
  Xiong}.} \bibinfo{year}{2015}\natexlab{}.
\newblock \showarticletitle{Protecting locations with differential privacy
  under temporal correlations}. In \bibinfo{booktitle}{\emph{Proceedings of the
  22nd ACM SIGSAC Conference on Computer and Communications Security}}.
  \bibinfo{pages}{1298--1309}.
\newblock


\bibitem[\protect\citeauthoryear{Xu, Papadimitriou, Haeberlen, and Feldman}{Xu
  et~al\mbox{.}}{2019}]%
        {xu2019hermetic}
\bibfield{author}{\bibinfo{person}{Min Xu}, \bibinfo{person}{Antonis
  Papadimitriou}, \bibinfo{person}{Andreas Haeberlen}, {and}
  \bibinfo{person}{Ariel Feldman}.} \bibinfo{year}{2019}\natexlab{}.
\newblock \showarticletitle{Hermetic: Privacy-preserving distributed analytics
  without (most) side channels}.
\newblock \bibinfo{journal}{\emph{External Links: Link Cited by}}
  (\bibinfo{year}{2019}).
\newblock


\bibitem[\protect\citeauthoryear{Zhang, Katz, and Papamanthou}{Zhang
  et~al\mbox{.}}{2016}]%
        {zhang2016all}
\bibfield{author}{\bibinfo{person}{Yupeng Zhang}, \bibinfo{person}{Jonathan
  Katz}, {and} \bibinfo{person}{Charalampos Papamanthou}.}
  \bibinfo{year}{2016}\natexlab{}.
\newblock \showarticletitle{All your queries are belong to us: The power of
  file-injection attacks on searchable encryption}. In
  \bibinfo{booktitle}{\emph{25th $\{$USENIX$\}$ Security Symposium
  ($\{$USENIX$\}$ Security 16)}}. \bibinfo{pages}{707--720}.
\newblock


\bibitem[\protect\citeauthoryear{Zheng, Dave, Beekman, Popa, Gonzalez, and
  Stoica}{Zheng et~al\mbox{.}}{2017}]%
        {zheng2017opaque}
\bibfield{author}{\bibinfo{person}{Wenting Zheng}, \bibinfo{person}{Ankur
  Dave}, \bibinfo{person}{Jethro~G Beekman}, \bibinfo{person}{Raluca~Ada Popa},
  \bibinfo{person}{Joseph~E Gonzalez}, {and} \bibinfo{person}{Ion Stoica}.}
  \bibinfo{year}{2017}\natexlab{}.
\newblock \showarticletitle{Opaque: An oblivious and encrypted distributed
  analytics platform}. In \bibinfo{booktitle}{\emph{14th $\{$USENIX$\}$
  Symposium on Networked Systems Design and Implementation ($\{$NSDI$\}$ 17)}}.
  \bibinfo{pages}{283--298}.
\newblock


\end{thebibliography}

\appendix
\section{Additional Implementation Details}

\subsection{Truncated view transformation}
In addition to the oblivious sort-merge join discussed in Example~\ref{exmp:smj}. We continue to provide two additional instantiations of the truncated view transformation, namely the truncated selection and truncated nested loop join.

\subsubsection{Oblivious selection (filter)} Since each input record can only contribute to the output of the selection operator at most once. Therefore, it does not require us to have additional implementations to constraint the record contributions. To ensure obliviousness, the operator will return all input data as the output. However, only records that satisfy the selection predicate will have its {\it isView} bit set to 1. As a result, records that do not satisfy the selection predicate are treated as dummy tuples ({\it isView=0}).

\subsubsection{Truncated (oblivious) nested-loop join}. The truncated nested-loop join is similar to a normal nested-loop join, where the operator scans the first table (outer table), say $T_1$, and joins each of tuples in $T_1$ to the rows in the second table $T_2$ (inner table). However, additional operations are required to ensure obliviousness and the bounded record contribution. Algorithm~\ref{algo:nlj} illustrates the details of this truncated nested-loop join method.

\begin{algorithm}[]
\caption{$\trans$ protocol}
\begin{algorithmic}[1]
\Statex
\textbf{Input}: Two tables $T_1$ and $T_2$; Truncation bound $b$
\State $\mathsf{assign\_budget}(T_1 \cup T_2, b)$
\State ${\bf o} = \mathsf{init\_secure\_array}()$
\For{$tup^1_i \in T_1$}
\State ${\bf o}_i = \mathsf{init\_secure\_array}()$
\For{$tup^2_j \in T_2$}
\If{$\mathsf{budget}(tup^1_i) > 0$ {\bf and} $\mathsf{budget}(tup^2_j) > 0$}
\If{$tup^1_i.key == tup^2_j.key$}
\State ${\bf o}_i.\mathsf{append}(tup^1_i||tup^2_j || isView=1) $
\State $\mathsf{consume\_budget}(tup^1_i, tup^2_j, 1)$
\EndIf
\Else
\State ${\bf o}_i.\mathsf{append}(dummy)$
\EndIf
\EndFor
\State $\hat{{\bf o}_i} \gets \mathsf{Oblisort}({\bf o}_i)$
\State $\hat{{\bf o}_i} \gets \hat{{\bf o}_i}[0,1,2,...,b-1]$
\State ${\bf o}.\mathsf{append}(\hat{{\bf o}_i})$
\EndFor
\end{algorithmic}
\label{algo:nlj}
\end{algorithm}
Initially, the operator assign a contribution budget to each tuple in $T_1$ and $T_2$. This can be achieved by append a fixed point number (i.e. 32-bit) after each tuple. Then for each tuple $tup^1_i \in T_1$, the operator joins it with each tuple $tup^2_j \in T_2$. A join tuple is generated if and only if (i) both $tup^1_i$ and $tup^2_j$ have remaining budgets (Algo~\ref{algo:nlj}:6) and (ii) the two tuples share the same join key. Once a join tuple is generated, the operator consumes the budgets of both join tuples, subtracting their remaining budgets by one (Algo~\ref{algo:nlj}:9). Otherwise it generates a dummy tuple. Additionally, at the end of each inner loop (Algo~\ref{algo:nlj}:12-13), the operator obliviously sort the intermediate tuples ${\bf o}_i$ and picks only the first $b$ tuples stored in ${\bf o}_i$. Since the contribution bound is $b$, thus the total number of true joins  in ${\bf o}_i$ must not exceed $b$. And by applying these steps (Algo~\ref{algo:nlj}:12-13) could help to reduce the cache I/O burden. 

\subsection{Generating secret shares inside MPC}\label{sec:proof:ss}
We provide an implementation example on how to generate $k$-out-of-$k$ XOR secret shares inside an MPC protocol follow by a sketch proof for it's security.

We assume there are $k$ participating owners $P_1, P_2, ..., P_k$, and there is a secret value $c$ that is computed inside MPC. We assume the secret share scheme is over the ring $\mathbb{Z}_m$ and we denote the MPC protocol that computes $c$ internally as $\mathsf{protocol}$. 
\begin{enumerate}
    \item $\forall~P_i$, randomly samples $k-1$ values $(z_i^1, z_i^2, ..., z_i^{(k-1)})\xleftarrow[]{\text{rd}}\mathbb{Z}_{m}$, from ring $\mathbb{Z}_m$.
    \item Each $P_i$ inputs the sampled values to $\mathsf{protocol}$.
    \item $c$ is computed inside $\mathsf{protocol}$.
    \item Before revealing outputs, the $\mathsf{protocol}$ computes $\forall~z^j,~ z^j\gets z_1^j\oplus z_2^j...\oplus z_k^j$, internally.
    \item The $\mathsf{protocol}$ computes secret shares $(x_1, x_2, ..., x_k)$, such that $\forall~ j<(k-1),~x_j \gets z^j$, and $x_k \gets c\oplus z^1 \oplus z^2 ...\oplus z^{k-1}$.
    \item The $\mathsf{protocol}$ reveals one secret share $x_i$ to one participating owner.
\end{enumerate}
\boldparagraph{Availability}. For all shares $\share{x} \gets (x_1, x_2, ..., x_k)$, $$\textup{Pr}[\rec(\share{x}) = c] = \textup{Pr}[x_1 \oplus x_2...\oplus x_k = c]=1$$

\boldparagraph{Confidentiality}. We introduce the following lemma that defines the security of secret sharing with adversary. 
\begin{lemma}\label{mpcsec-shares}A $t$-out-of-$n$ secret sharing scheme $\langle\gs, \rec\rangle$ over ring $\mathbb{Z}_m$ is perfectly secure if for any adversary $\mathcal{A}$, $\forall~S\subseteq(1, 2,..., n)$ such that $|S| < t$, and for any two messages $m$, and $m'$, the following holds:
\begin{align*}
\textup{Pr}\left[\begin{array}{l}
\mathcal{A}(x_i~|~i\in S) = 1 ~:\\
(x_1, x_2, ...x_k)\gets \gs(m)
\end{array}\right] =  \textup{Pr}\left[\begin{array}{l}
\mathcal{A}(x'_i~|~i\in S) = 1 ~:\\
(x'_1, x'_2, ...x'_k)\gets \gs(m')
\end{array}\right]
\end{align*}
\end{lemma}

In our setting, we consider the adversary is able to obtain up to $k-1$ out of $k$ secret shares, and we prove the security by illustrating our implemented approach satisfies Lemma~\ref{mpcsec-shares}.

Let $X(m)$ to be the secret shares obtained by the adversary, $\mathcal{A}$, and let $X'(m)$ to be the secret shares such that $X'\subseteq\gs(m)$ and $X(m)\cap X'(m)=\emptyset$. Since $\mathcal{A}$ is able to control $k-1$ parties, thus we consider the following 2 cases: (i) if $m\oplus z^1...\oplus z^{k-1} \notin X(m)$, then $X(m)$ is independent from input message $m$, therefore $\mathcal{A}$ can not distinguish the shares for two different messages. (ii) if $m\oplus z^1...\oplus z^{k-1} \in X(m)$, since $X'(m)\neq \emptyset$, thus $\exists z^i,~z^i \notin X(m)$. Since $z^i$ uses randomness that  can not controlled by $\mathcal{A}$, thus $\textup{Pr}[\mathcal{A}(z^i)=1]=\textup{Pr}[\mathcal{A}(z'\xleftarrow{\mathsf{rd}}\mathbb{Z}_m)=1]$. Moreover, as $m\oplus z^1...\oplus z^{k-1}$ is masked by $z^i$, thus for any two messages, the adversary $\mathcal{A}$ can not distinguish between their secret shares. 
\section{Proof of Theorems}
\subsection{Proof of Theorem~\ref{tm:bcompose}}
\begin{proof}
Assume two neighboring databases $\mathcal{D}$, $\mathcal{D}'$, differ by one record $u$. Let ${\bf o} = \{o_i\}_{i\geq0}$, and ${\bf o}' = \{o'_i\}_{i\geq0}$ to be the output of $\mathcal{M}(\mathcal{D})$ and $\mathcal{M}(\mathcal{D}')$, respectively, where $o_i$ and $o'_i$ denotes $\mathcal{M}_i$'s output. We use $T(U_i)$ and $T(U'_i)$ to denote the corresponding input of $\mathcal{M}_i$ when $\mathcal{M}$ applies over $\mathcal{D}$ and $\mathcal{D}'$, respectively. We know that $ \epsilon =  \ln{\left(\frac{\textup{Pr}[\mathcal{M}(\mathcal{D}) = {\bf o}]}{\textup{Pr}[\mathcal{M}(\mathcal{D}') = {\bf o}]}\right)}  = \ln{\left(\prod_{i ~:~ \tau_i(u) > 0}\frac{\textup{Pr}[\mathcal{M}_i(T(U_i)) = o_i]}{\textup{Pr}[\mathcal{M}_i(T(U'_i)) = o_i]}\right)}$, therefore we obtain \smash{$\epsilon \leq \max_{u}{\ln{\left(\prod_{i ~:~ \tau_i(u) > 0} e^{q\epsilon_i}\right)}} \leq$ Eq.~\ref{eq:compose}}. 
\eat{
Now, we compute the following
\begin{equation}
    \begin{split}
        \epsilon & =  \ln{\left(\frac{\textup{Pr}[\mathcal{M}(\mathcal{D}) = {\bf o}]}{\textup{Pr}[\mathcal{M}(\mathcal{D}') = {\bf o}]}\right)}  = \ln{\left(\prod_{i\geq 0}\frac{\textup{Pr}[\mathcal{M}_i(T_i(\mathcal{D})) = o_i]}{\textup{Pr}[\mathcal{M}_i(T_i(\mathcal{D}')) = o_i]}\right)}\\
        & =   \ln{\left(\prod_{i ~:~ \tau_i(u) > 0}\frac{\textup{Pr}[\mathcal{M}_i(T(U_i)) = o_i]}{\textup{Pr}[\mathcal{M}_i(T(U'_i)) = o_i]}\right)}\\
        & \leq \max_{u}{\ln{\left(\prod_{i ~:~ \tau_i(u) > 0} e^{q\epsilon_i}\right)}} \leq \max_{u, \mathcal{D}}\left(\sum_{i ~:~ \tau_i(u) > 0} q_i\epsilon_i\right)\\
    \end{split}
\end{equation}}
\end{proof}
\subsection{Proof of Theorem~\ref{lg:timer},~\ref{perf:timer}}
\begin{lemma}\label{lemma:sum}
Given $k$ independent and identically distributed Laplace random variables, $Y_1, Y_2,...,Y_k$, where each $Y_i$ is sampled from the distribution $\textup{Lap}(\frac{\Delta}{\epsilon})$, where $\Delta$ denotes the sensitivity. Let $Y = \sum_i^k Y_i$, and $0 < \alpha \leq k\frac{\Delta}{\epsilon}$, the following inequality holds
$$\textup{Pr}\left[~ Y \geq \alpha \right] \leq e^{\left( \frac{-\alpha^2\epsilon^2}{4k\Delta^2} \right)}$$
\end{lemma}
\begin{proof}The complete proof of Lemma~\ref{lemma:sum} can be found in the Appendix C.1 of~\cite{wang2021dp} and in~\cite{dwork2014algorithmic}.
\end{proof}
\begin{corollary}\label{col:sum}
Given $k$ independent and identically distributed Laplace random variables, $Y_1, Y_2,...,Y_k$, where each $Y_i$ is sampled from the distribution $\textup{Lap}(\frac{\Delta}{\epsilon})$. Let $Y=\sum_{i=1}^k Y_i$, and $\beta \in (0,1)$, the following inequality holds
$$\textup{Pr}\left[~ Y \geq 2\frac{\Delta}{\epsilon}\sqrt{k\log{\frac{1}{\beta}}} ~\right] \leq \beta $$
\end{corollary}
\begin{proof}
Continue with Lemma~\ref{lemma:sum}, let $e^{\frac{-\alpha^2\epsilon^2}{4k\Delta^2}} = {\beta}$, then $\alpha = 2\frac{\Delta}{\epsilon}\sqrt{k\log{\frac{1}{\beta}}}$, when $k > 4\log{\frac{1}{\beta}}$ the corollary holds.
\end{proof}
\begin{proof}(\textbf{Theorem~\ref{lg:timer}}). Let $\tilde{c}_k$ denotes the total number of records synchronized by $\sync$ protocol after $k$ times update (without cache flush), and let $c_k$ denotes the true cardinality of materialized view after $k$ updates (at time $kT$). Knowing that $\tilde{c}_k\gets c_k + \sum_{i=1}^{k} Y_i$, where each $Y_i$ is an i.i.d Laplace rnadom variable drawn from the distribution $\lap(\frac{b}{\epsilon})$, where $b$ is the contribution upper bound for each record. Given $\beta \in (0,1)$, according to Corollary~\ref{col:sum} we obtain that $\textup{Pr}\left[c_k - \tilde{c}_k \geq \alpha \right]\leq \beta$, such that $\alpha\gets\frac{2b}{\epsilon}\sqrt{k\log{\frac{1}{\beta}}}$. Knowing that $c_k - \tilde{c}_k$ computes the total number of records delayed after $k$-th updates, thus the theorem holds.
\end{proof}
\begin{proof}(\textbf{Theorem~\ref{perf:timer}}). Similarly, let $\tilde{c}_k$ denotes the total number of records synchronized by $\sync$ protocol after $k$ times update, and let $c_k$ denotes the true cardinality of materialized view after $k$ updates (at time $kT$). When considering cache flush, we compute $\tilde{c}_k \gets c_k + \sum_{i}^{k}Y_i + \sum_{i}^{k'} s $, where $k'$ denotes the number of cache flushes occurred since $t=0$ and $s$ is the cache flush size. 

Knowing that $\tilde{c}_k\gets c_k + \sum_{i=1}^{k} Y_i$, where each $Y_i$ is an i.i.d Laplace rnadom variable drawn from the distribution $\lap(\frac{b}{\epsilon})$, where $b$ is the contribution upper bound for each record. Given $\beta \in (0,1)$, according to Corollary~\ref{col:sum} we obtain that $\textup{Pr}\left[c_k - \tilde{c}_k \geq \alpha \right]\leq \beta$, such that $\alpha\gets\frac{2b}{\epsilon}\sqrt{k\log{\frac{1}{\beta}}}$. Knowing that $\tilde{c}_k - c_k$ computes the total number of records delayed after $k$-th updates, thus the theorem holds. Knowing that $k'\gets\lfloor\frac{kT}{f}\rfloor\leq \frac{kT}{f}$, and according to Corollary~\ref{col:sum}, $\sum_{i}^{k} Y_i$ is bounded by \smash{$O(\frac{2b\sqrt{k}}{\epsilon})$}, we conclude that the dummy data after $k$-th updates is bounded by $O(\frac{2b\sqrt{k}}{\epsilon} ) + \frac{skT}{f}$.
\end{proof}

\subsection{Proof of Theorem~\ref{one:ant}}
\begin{proof} 


Let $t$ denotes the current time, $c_i$ counts how many records received since last update at every time $i$. Assuming there $k$ updates happened before current time $t$, and thus we have $k$ noisy thresholds $\tilde{\theta}_1, \tilde{\theta}_2, ...\tilde{\theta}_k$.  Let $A =\{a_1, a_2, ..., a_t\}$ as the collection of $ant$'s outputs, where $a_i \in A$ is either $\perp$ (no updates) or equals to $c_j + \lap(\frac{2b}{\epsilon})$. According the {\it Fact 3.7} in~\cite{dwork2014algorithmic}, such that 
\begin{equation}
\begin{split}
 ~~\textup{Pr}\left[ \forall_{j}|\tilde{\theta}_j - \theta| \geq \frac{\alpha}{4} \right] = e^{-\frac{\epsilon\alpha}{16}}
\Rightarrow  ~~\textup{Pr}\left[ \sum_{j=1}^{k}|\tilde{\theta}_j - \theta| \geq \frac{\alpha}{4} \right] = k\times e^{-\frac{\epsilon\alpha}{16}}\\
\end{split}
\end{equation}
Let $k\times e^{-\frac{\epsilon\alpha}{16b}}$ to be at most $\beta/4$, then $\alpha \geq \frac{16b(\log{k} + \log{(4/\beta}))}{\epsilon}$.
Similarly, for each time $i$, we know that $\tilde{c}_i - c_i = \lap(\frac{8b}{\epsilon})$, where $\tilde{c}_i$ is the value used to compare with the noisy threshold, it satisfies:
\begin{equation}
\begin{split}
 & ~~\textup{Pr}\left[ \forall_{0 < i \leq t} |\tilde{c}_i - c_i| \geq \frac{\alpha}{2} \right] \leq e^{-\frac{\epsilon\alpha}{16}}\\
\Rightarrow & ~~\textup{Pr}\left[ \sum_{i=1~:~a_i\neq\bot}^{t} |\tilde{c}_i - c_i| \geq \frac{\alpha}{2} \right] \leq \sum_{j=1}^{k} (t^{j} - t^{j-1})\times e^{-\frac{\epsilon\alpha}{16}} \leq t\times e^{-\frac{\epsilon\alpha}{16}}
\end{split}
\end{equation}
where $t^{j}$ denotes the time stamp for $j^{th}$ update, and let  $t \times e^{-\frac{\epsilon\alpha}{16}}$ to be at most $\frac{\beta}{2}$, we have $\alpha \geq \frac{16(\log{t} + \log{(2/\beta}))}{\epsilon}$. Finally, we set the following conditions $\forall_{i~:~ a_i\neq\bot} |a_i - c_i = \lap(\frac{2b}{\epsilon})| \geq \alpha$ holds with probability at most $\beta/4$, we obtain $\alpha \geq \frac{4b\log{(4/\beta)}}{\epsilon}$. By combining the above analysis, we can obtain if set $\alpha \geq \frac{16b(\log{t} + \log{(2/\beta}))}{\epsilon}$ the following holds.
\begin{equation}\label{eq:final}
\begin{split}
 & ~~\textup{Pr}\left[ \left(\sum_{\forall_{i ~:~ a_i \neq \bot}}^{t} c_i - \sum_{\forall_{i ~:~ a_i \neq \bot}}^{t} a_i \right) \geq {\alpha} \right] \leq ~~\textup{Pr}\left[ \sum_{\forall i ~:~ a_i \neq \bot }^{t} |a_i - c_i|  \geq {\alpha}\right] \\
 \leq & ~~\textup{Pr}\left[ \left(\sum_{j=1}^{k}|\tilde{\theta}_j - \sum_{i=1~:~a_i\neq\bot}^{t} |\tilde{c}_i - c_i| + \sum_{\forall i~:~ a_i \neq \bot} |a_i - c_i|  \right) \geq {\alpha}\right] \leq \beta. 
\end{split}
\end{equation}
according to Eq.~\ref{eq:final},  with probability at most $\beta$, the number of deferred data, $\sum_{\forall_{i ~:~ a_i \neq \bot}}^{t} (c_i - a_i)$ is greater than  $\alpha \geq \frac{16b(\log{t} + \log{(2/\beta}))}{\epsilon}$. thus the total number of deferred data is bounded by $O(\frac{16b\log{t}}{\epsilon})$.

\end{proof}

\DeclareRobustCommand{\rchi}{{\mathpalette\irchi\relax}}
\newcommand{\irchi}[2]{\raisebox{\depth}{$#1\chi$}}

\section{Security Proof}
In this section we continue to provide the complete formal security proof for \system framework. We first provide privacy proofs for mechanisms $\mathcal{M}_{\mathsf{timer}}$ and $\mathcal{M}_{\mathsf{ant}}$ provided in Theorem~\ref{tm:pf-timer} and~\ref{tm:pfant}.

\begin{theorem} $\mathcal{M}_{\mathsf{timer}}$ provided in Theorem~\ref{tm:pf-timer} satisfies $\epsilon$-DP.
\end{theorem}
\begin{proof}
First we construct $\mathcal{M}_{\mathsf{unit}}(X, \epsilon)\gets f(X)+\lap(\frac{\Delta f}{\epsilon})$ where $f = \sum_{\forall x_i \in X} 1 | x_i \neq \emptyset$, and:
\begin{equation}
    \begin{split}
        \Delta f = & \max_{\forall U_x,U_y\in\rchi' \wedge  ||U_x - U_y||_1 \leq 1}| f(U_x) - f(U_y)|
    \end{split}
\end{equation} 

Let $\rchi'$ denotes all possible inputs, and let $U_x \in \rchi'$, and $U_y \in \rchi$, denote two neighboring inputs (differ by one tuple). Next, let $p_x$, $p_y$ denote the density functions of $\mathcal{M}_{\mathsf{unit}}(U_x, \epsilon)$, and $\mathcal{M}_{\mathsf{unit}}(U_y, \epsilon)$, respectively. We compare the two terms under arbitrary point $z$:
\begin{equation}
\begin{split}
    \frac{p_x(z)}{p_y(z)} = & ~\frac{\frac{1}{2b}e^{\frac{-|f(U_x) - z|}{b}}}{\frac{1}{2b}e^{\frac{-|f(U_y) - z|}{b}}} = e^{\frac{|f(U_y) - z| - |f(U_x) - z|}{b}} \leq e^{\frac{|f(U_y) - f(U_x)|}{b}}
\end{split}
\end{equation}

\eat{
\begin{equation}
\begin{split}
    \frac{p_x(z)}{p_y(z)} = & ~\Pi_{i=1}^{k}\left(\frac{e^{-\epsilon(\frac{| f(U_x) - z_i|}{\Delta f})}}{e^{-\epsilon(\frac{| f(U_y) - z_i|}{\Delta f})}}\right) \\
    = & ~\Pi_{i=1}^{k}e^{\left( \epsilon(\frac{|f(U_y) - z_i| - |f(U_x) - z_i|}{\Delta f})\right)}\\
    \leq & ~\Pi_{i=1}^{k}e^{\left( \epsilon(\frac{|f(U_y) - f(U_x)|}{\Delta f})\right)} = e^{\left( \epsilon(\frac{||f(U_y) - f(U_x)||_1}{\Delta f})\right)} \leq e^{\epsilon}
\end{split}
\end{equation}}
Now, we set $b = \frac{\Delta f}{\epsilon}$, and thus:
\begin{equation}
\begin{split}
 e^{\frac{|f(U_y) - f(U_x)|}{b}} \leq e^{\frac{\Delta f}{\frac{1}{\epsilon}}} = e^{\epsilon} \rightarrow \frac{p_x(z)}{p_y(z)} \leq e^{\epsilon}.
\end{split}
\end{equation}
Note that the ratio $\frac{p_x(z)}{p_y(z)} \geq e^{-\epsilon}$ follows by symmetry. As a conclusion, $\mathcal{M}_{\mathsf{unit}}$ satisfies $\epsilon$-DP. Then $\mathcal{M}_{\mathsf{timer}}$ can be formulated as the composition of  $\mathcal{M}_{\mathsf{unit}}$ as follows

$$\mathcal{M}_{\mathsf{timer}} (\mathcal{D}) = \begin{cases}    \mathcal{M}_{\mathsf{unit}}(\Delta_{t-T}^{t}\mathcal{D})  &  \text{if}~ 0\equiv t \Mod T ,\\    
        0               &  \text{otherwise}.\end{cases}$$
        
where $\Delta_{t-T}^{t}\mathcal{D} = \mathcal{D}_t - \mathcal{D}_{t-T}$ denotes the logical updates (records) received between time $(t-T, t]$. Since output a constant is independent of the input data, thus proving the privacy of $\mathcal{M}_{\mathsf{timer}} (\mathcal{D})$ is reduced to proving the privacy of the composed mechanism $\mathcal{M}(\mathcal{D}) = \{\mathcal{M}_{\mathsf{unit}}(\Delta_{iT}^{(i+1)T}\mathcal{D})\}_{i=1,2,3...}$

Now consider two neighboring databases $\mathcal{D}$ and $\mathcal{D}'$ that differs in one logical update $u_t$. Then we compute 
\begin{equation}\label{eq:compose-grow}
    \begin{split}
         &   \ln{\left(\frac{\textup{Pr}[\mathcal{M}(\mathcal{D}) = {\bf o}]}{\textup{Pr}[\mathcal{M}(\mathcal{D}') = {\bf o}]}\right)} \\ 
         & = \ln{\left(\prod_{i\geq 1}\frac{\textup{Pr}[\mathcal{M}_{\mathsf{unit}}(\Delta_{iT}^{(i+1)T}\mathcal{D}) = o_i]}{\textup{Pr}[\mathcal{M}_{\mathsf{unit}}(\Delta_{iT}^{(i+1)T}\mathcal{D}') = o_i]}\right)}\\
        & =   \ln{\left(\prod_{i : u_t \in \Delta_{iT}^{(i+1)T}\mathcal{D}}\frac{\textup{Pr}[\mathcal{M}_{\mathsf{unit}}(\Delta_{iT}^{(i+1)T}\mathcal{D}) = o_i]}{\textup{Pr}[\mathcal{M}_{\mathsf{unit}}(\Delta_{iT}^{(i+1)T}\mathcal{D}') = o_i]}\right)} \leq \epsilon\\
    \end{split}
\end{equation}
Therefore, $\mathcal{M}_{\mathsf{timer}}$ satisfies $\epsilon$-DP.
\end{proof}

\begin{theorem}
$\mathcal{M}_{\mathsf{ant}}$ provided in Theorem~\ref{tm:pfant} satisfies $\epsilon$-DP.
\end{theorem}
\begin{proof}
First we provide mechanism NANT ( Numeric Above Noisy Threshold) in Algorithm~\ref{algo:nant}, and proves it's privacy guarantee.

\begin{algorithm}[]
\caption{Numeric Above Noisy Threshold}
\begin{algorithmic}[1]
\Statex
\textbf{Input}: data stream $X$, privacy budget $\epsilon$, threshold $\theta$.
\State $\epsilon_1 \gets \frac{1}{2}\epsilon, \epsilon_2 \gets \frac{1}{2}\epsilon$
\State $ \tilde{\theta} \leftarrow \theta + \lap(\frac{2\Delta f}{\epsilon_1})$, $c \leftarrow 0$
\For{$t \leftarrow 1,2,...$} 
\State $v_t \leftarrow \lap(\frac{4\Delta f}{\epsilon_1})$
\State $c \gets f(X, t)$
\If{$c + v_t \geq \tilde{\theta}$}
\State \textbf{output} $\tilde{c}\gets c + \lap(\frac{2\Delta f}{\epsilon_2}) $, \textbf{return}
\Else 
\State \textbf{output} $0$
\EndIf
\EndFor
\end{algorithmic}
\label{algo:nant}
\end{algorithm}

where $ f(X, t) \gets \sum_{i\gets 1}^{t} 1$ $ ~|~(x_t \in X \wedge x_t\neq \emptyset)$. We now prove it's privacy as follows:

We start with a modified mechanism $\mathcal{M}_{\mathsf{unit}}$ of NANT such that where it outputs $\top$ once the condition $c + c_t >\tilde{\theta}$ is satisfied (Alg~\ref{algo:nant}:6),  and outputs $\bot$ for all other cases (Alg~\ref{algo:nant}:9). We write the output of $\mathcal{M}_{\mathsf{unit}}$ as $O = \{o_1, o_2,...,o_m \}$, where $\forall~ 1 \leq i < m$, $o_i = \bot$, and $o_m = \top$. Now given two neighboring database $X$ and $X'$, and for all $i$, $\textup{Pr}\left[\tilde{c}_i < x \right] \leq \textup{Pr}\left[\tilde{c}'_i < x+1 \right]$ is satisfied, where $\tilde{c}_i$ and $\tilde{c}'_i$ denotes the $i^{th}$ noisy count when applying  $\mathcal{M}_{\mathsf{unit}}$ over $X$ and $X'$ respectively, such that:
\begin{equation}
\begin{split}
  & \textup{Pr}\left[~ \mathcal{M'}_{\mathsf{sparse}}(U) = O \right] \\
  =  \int_{-\infty}^{\infty} & \textup{Pr}\left[\tilde{\theta} = x \right]\left( \prod_{1 \leq i < m}\textup{Pr}\left[\tilde{c}_i < x\right] \right) \textup{Pr}\left[\tilde{c}_m \geq x \right]dx\\
\leq \int_{-\infty}^{\infty} & e^{\epsilon/2}\textup{Pr}\left[\tilde{\theta} = x + 1 \right]\left( \prod_{1 \leq i < m}\textup{Pr}\left[\tilde{c}'_i < x + 1\right] \right) \\
  \times & e^{\epsilon/2} \textup{Pr}\left[ v_m + c'_m \geq x + 1 \right]dx\\
  = \int_{-\infty}^{\infty} & e^{\epsilon}\textup{Pr}\left[\tilde{\theta} = x + 1 \right]\left( \prod_{1 \leq i < m}\textup{Pr}\left[\tilde{c}'_i < x + 1\right] \right) \textup{Pr}\left[ \tilde{c}'_m \geq x + 1 \right]dx\\
  = e^{\epsilon}\textup{Pr} & [\mathcal{M'}_{\mathsf{sparse}}(U') = O]
\end{split}
\end{equation}
Thus, $\mathcal{M}_{\mathsf{unit}}$ satisfies $\epsilon$-DP. Moreover, mechanism NANT can be expressed as the composition of a $\mathcal{M}_{\mathsf{unit}}$ and a Laplace mechanism, each with a privacy parameter of $\frac{\epsilon}{2}$. Therefore by sequential composition, NANT satisfies $\epsilon$-DP. Similar, $\mathcal{M}_{\mathsf{ant}}$ can be treated as repeatedly running NANT over disjoint data, thus by Eq.~\ref{eq:compose-grow}. $\mathcal{M}_{\mathsf{ant}}$ satisfies $\epsilon$-DP.
\end{proof}

\begin{definition}[Secure 2-Party Computation~\cite{lindell2017simulate}]\label{def:s2pc} Let $f = (f1, f2)$ be a functionality and let, $\pi$ to be a 2 party protocol that computes $f$. We say that $\pi$ securely computes $f$ in the presence of semi-honest adversaries if there exists {\it p.p.t.} simulator $\mathcal{S}_1$ and $\mathcal{S}_2$:
\begin{equation}
    \begin{split}
        \{\mathcal{S}_1(x, f_1(x,y)), f(x,y)\} \myeq \{\vi_1^{\pi}(x,y), \mathsf{output}^{\pi}(x,y)\}\\
        \{\mathcal{S}_2(y, f_2(x,y)), f(x,y)\} \myeq \{\vi_2^{\pi}(x,y), \mathsf{output}^{\pi}(x,y)\}
    \end{split}
\end{equation}
where $\myeq$ means computational indistinguishable, $\vi^{\pi}$ and $\mathsf{output}^{\pi}$ denotes the views and outputs when evaluating protocol $\pi$. 
\end{definition}

\begin{theorem} If there exists secure 2-PC protocols that satisfy Definition~\ref{def:s2pc} and $(2,2)$-secret sharing scheme that satisfy Lemma~\ref{mpcsec-shares}, then \system implemented with $\timer$ and $\ant$ view update protocol satisfies $\epsilon$-SIM-CDP.
\end{theorem}
\begin{proof}
In this section, we focus on proving the simulator provided in Table~\ref{tab:sim1} yields computationally indistinguishable outputs compared to the execution of the real view update protocols. Let $f_{t} (x, y)$ to be the functionality of truncated view transformation and  $f_{s} (x, y)$ to be the functionality of synchronizing data from secure cache to materialized view. $\pi_t$ and $\pi_s$ are the protocols that securely computes these two functionalities ($\trans$ and $\sync$). In general, we assume the secure cache and the materialized view are the secret-shared objects across the 2-PC participants. 

Therefore, at each time $t$, the adversary's view against the entire view update protocol can be formulated as:
$$\vi_j^{\pi}(x_j, x_{1-j}, t) = \{\gamma^{u}_j, \gamma^{v}_j, \gamma^{s}_j, \gamma^{f}_j, c_j, \theta_j\}$$
where $\gamma^{u}_j, \gamma^{v}_j, \gamma^{s}_j, \gamma^{f}_j, c_j$, and $\theta_j$ denotes the corresponding secret shares party $j$ obtains for user uploaded data, transformed view tuples, synchronized data, flushed data, cardinality counter and the noisy threshold (this parameter is not included in the view of $\timer$ protocol). Since $\gamma^u_j$ is the secret shared data generated by owners, thus by Lemma~\ref{mpcsec-shares},
$$\gamma^u_j \myeq B \xleftarrow[]{\text{rd}}\mathbb{Z}_{m}~\textbf{if}~~|\gamma^u_j|=|B|$$
Since $\gamma^{v}_j$ is computed from the secure 2PC protocol $\pi$, thus by Theorem~\ref{def:s2pc}, there must exists simulator such that $\mathcal{S}(x_j, f_j(x_j, x_{1-j}))\myeq \gamma^{v}_j$. In addition, since $\gamma^{v}_j$ is assumed to be secret-shared data, thus 
$$\exists \mathcal{S}, ~s.t.~ \mathcal{S}(B \xleftarrow[]{\text{rd}}\mathbb{Z}_{m}) \myeq \gamma^{v}_j~\textbf{if}~~|\gamma^v_j|=|B|$$
 Similarly, we can also obtain 
 $$\exists \mathcal{S}, ~s.t.~ \mathcal{S}(a, b \xleftarrow[]{\text{rd}}\mathbb{Z}_{m}) \myeq c_j, \theta_j$$
 Finally, since $f_2$ obliviously sorts then fetches from the recovered input data, thus the output of $f_2$ should be computational indistinguishable from the random sampling over the recovered input data. This also applies to cache flush, and therefore we can obtain
  $$\exists \mathcal{S}, ~s.t.~ \mathcal{S}(B, B' \gets \mathsf{rd\_sample}(x)) \myeq \gamma^{s}_j, \gamma^{f}_j~\textbf{if}~~|B|=|\gamma^{s}_j|\wedge|B'|=|\gamma^{f}_j|$$
  As per the aforementioned analysis, the simulator provided in Table~\ref{tab:sim1} yields computational indistinguishable transcripts in comparison with the real protocol execution. Since the simulator only takes in the outputs of differentially-private mechanisms and public parameters, thus the view update protocol satisfies $\epsilon$-SIM-CDP.
\end{proof}
\section{Extension continued}
\subsection{Connecting with DP-Sync}
We continue to provide more details regarding the utility guarantees when combining DP-Sync and \system. 
\begin{theorem}[Logical gap~\cite{wang2021dp}]  For each time $t$, the logical gap, $LG_t$ between the outsourced and logical database is defined as the total number of records that have been received by the owner but have not been outsourced to the server. 
\end{theorem}
In~\cite{wang2021dp}, logical gap is used as the major utility metric and typically a large logical gap indicates a relatively large error for queries to the outsourced database. Similar we derive a logical gap at time $t$ for the materialized view as $LG^{\mathcal{V}}_t$, which denotes the number of view tuples delayed by the respective mechanisms (record synchronization strategy and view update protocol).
\begin{theorem}[($\alpha,\beta$)-accurate sync strategy] A record synchronization strategy $r\_sync$ over growing data $\mathcal{D}$ is ($\alpha,\beta$)-accurate if there exists $\alpha > 0$, and $0<\beta<1$, such that the logical gap when outsourcing $\mathcal{D}$ with $r\_sync$ satisfies, $\forall t$
$$\textup{Pr}\left[ LG_t > \alpha \right] < \beta$$
\end{theorem}





\begin{theorem}\label{thm:logicalgap}
Applying \system over the outsourced data uploaded by an $(\alpha, \beta)$-accurate private synchronization strategy $r\_sync$, results in error bounds \smash{$O(b\alpha+\frac{2b}{\epsilon}\sqrt{k})$} and \smash{$O(b\alpha+\frac{16b\log{t}}{\epsilon})$}, respectively for $\timer$ and $\ant$ protocol.
\end{theorem}
\begin{proof} We provide the proof of the error bound under $\timer$ protocol and the bound under $\ant$ can be proved using the same technique. Let $\beta_1 \in (0,1)$, and let $\xi_k$ denotes the total number of cached view tuples that are delayed for synchronization after $k^{th}$ view update. According to Theorem~\ref{lg:timer}, $\textup{Pr}\left[\xi_k > \frac{2b}{\epsilon}\sqrt{k\log{\frac{1}{\beta_1}}}\right]<\beta_1$, where $b$ is contribution bound and $\epsilon$ is privacy parameter. Let $r\_sync$ is an $(\alpha_2, \beta_2)$-accurate sync strategy, then the following holds
\begin{equation}
    \begin{split}
        & \textup{Pr}\left[LG_t^{\mathcal{V}} \geq \alpha \right] \leq \textup{Pr}\left[b\times LG_t + \xi_k \geq \alpha \right]\\
    \end{split}
\end{equation}
By union bound, $\textup{Pr}\left[b\times LG_t + \xi_k \geq b\alpha_2 + \frac{2b}{\epsilon}\sqrt{k\log{\frac{1}{\beta_1}}} \right] \leq \beta_1 + \beta_2$, thus we obtain $\textup{Pr}\left[LG_t^{\mathcal{V}}\geq b\alpha_2 + \frac{2b}{\epsilon}\sqrt{k\log{\frac{1}{\beta_1}}} \right] \leq \beta_1 + \beta_2$. This indicates the error is bounded by $O(b\alpha_2 + \frac{2b}{\epsilon}\sqrt{k})$. The same proof technique can be used to prove the error bound of $\ant$.
\end{proof}
In general, as the logical gap of materialized view is resulted by (i) the total data delayed by $r\_sync$ and (ii) the total view entries delayed by $\sync$ protocol. Thus the logical gaps of the two mechanisms are  additive. 

\subsection{Connecting with DP-Sync}
To further define the effect of dummy records on overall computation cost, we introduce two efficiency metrics.

\begin{definition}[Filter Efficiency] Given a Filter operator $O$ with input $O_{1}$ of size $n_{1}$, let the number of dummy records in $O_{1}$ be $Y_{1}(\epsilon_{1})$, where $\epsilon_{1}$ is defined in privacy budget allocation $P$ = {$\epsilon_{1},\ldots,\epsilon_l$}. The efficiency of $O$ is defined as:
$$E(P) = 1 - (Y_1(\epsilon_1)/n_1)$$
\end{definition}

\begin{definition}[Join Efficiency] Given a Join operator $O$ whose inputs $O_1$ and $O_2$ are of size $n_1$ and $n_2$, respectively. Let the number of dummy records in $O_1$ and $O_2$ be $Y_1(\epsilon_1)$ and $Y_2(\epsilon_2)$, respectively, where $\epsilon_1$ and $\epsilon_2$ are defined in privacy budget allocation $P$ = {$\epsilon_1,\ldots,\epsilon_l$}. The efficiency of $O$ is defined as:
$$E(P) = 1 - (Y_1(\epsilon_1) + Y_2(\epsilon_2))/(n_1 + n_2)$$
\end{definition}

The total efficiency of a given query $Q$ is defined as:

\begin{definition}[Query Efficiency] Given a query $Q$ comprised of operators {$O_1,\ldots,O_l$} with efficiencies {$E_1,\ldots,E_l$} and operator output cardinalities $|O_i|$, respectively. The efficiency of $Q$ with a privacy budget allocation $P$ and total output size $|O_{total}|$ is defined as:
$$E_Q(P) = \sum_{i=1}^l \frac{|O_i|}{|O_{total}|} E_i(P)$$
\end{definition}

Given a maximum privacy budget $\epsilon$ and a maximum logical gap $LG$, we can now define our optimization problem as follows:
\begin{eqnarray}
& \max_{P} E_Q(P) ~~s.t. & \sum_{i=1}^{l} \epsilon_i \leq \epsilon,  \sum_{i=1}^{l} LG_i \leq LG_{total}, \nonumber \\
&&\epsilon_i  \geq 0~~ \forall i=1,\ldots,l
\end{eqnarray}

Note that in order to obtain the optimal privacy budget allocation, we require the true number of dummy records $d_1$ and $d_2$ in the inputs to each operator $O$. However, revealing this information compromises our privacy guarantee. Instead, we can utilize estimates of $d_1$ and $d_2$ learned from the DP volume information released by our materialized joins, as seen in Figure~\ref{fig:cache}.  

\end{document}